\documentclass[a4paper,10pt]{article}

\usepackage[utf8]{inputenc}
\usepackage[T1]{fontenc}
\usepackage{textcomp}
\usepackage{fancyhdr}
\usepackage{fullpage}
\usepackage{lmodern}
\usepackage{amsmath,amssymb,bm,bbm}
\usepackage{algorithm}
\usepackage{algorithmic}
\usepackage{float}
\usepackage{graphicx}
\usepackage[shortlabels]{enumitem}
\usepackage{extarrows}
\usepackage{caption}
\usepackage{graphicx, subfig}
\usepackage{titling}
\usepackage{yhmath}
\usepackage{mathrsfs}
\usepackage{footnote}
\usepackage{amsthm}
\usepackage{color}
\usepackage{slashed}
\usepackage{verbatim}
\usepackage[colorlinks=true,linkcolor=blue,citecolor=blue]{hyperref}
\usepackage{amsfonts,euscript}
\usepackage{latexsym, multicol}
\usepackage{epstopdf}
\usepackage{psfrag}
\usepackage{mathrsfs}
\usepackage{xcolor}
\usepackage{comment}
\usepackage{authblk}
\usepackage{indentfirst}
\usepackage{parskip}

\usepackage{lipsum}
\usepackage[
backend=biber,
style=numeric,
maxcitenames=5,
maxbibnames=9]{biblatex}

\renewbibmacro{in:}{}

\numberwithin{equation}{section}
\usepackage{thmtools}
\declaretheorem[name=Theorem, numberwithin=section]{thm}
\declaretheorem[name=Proposition,sibling=thm]{prop}

\declaretheorem[name=Lemma,sibling=thm]{lem}

\declaretheorem[name=Remark,sibling=thm]{rem}

\theoremstyle{definition}
\declaretheorem[name=Definition,sibling=thm]{defi}

\newcommand{\nn}{\nonumber}
\newcommand{\anshatm}{\frac{\Delta^2}{\ubar^7}\sum_{|m|\leq2}Q_{m,\sfrak}Y_{m,\sfrak}^{-\sfrak}(\cos\theta)e^{im\phi_+}}
\newcommand{\ansatzp}{\frac{1}{\ubar^{7}}\sum_{|m|\leq2}A_m(r)Q_{m,\sfrak}Y_{m,\sfrak}^{+\sfrak}(\cos\theta)e^{im\phi_{+}}}

\newcommand{\teukm}{\teuk_{-2}}

\newcommand{\teukhat}{\widehat{\teuk}}
\newcommand{\intS}{\int_{S(u,\ubar)}}

\newcommand{\errhatm}{\mathrm{Err}[\Psihatm]
}
\newcommand{\drin}{\partial_{r_{\mathrm{in}}}}
\newcommand{\drout}{\partial_{r_{\mathrm{out}}}}

\newcommand{\Real}{\mathfrak{R}}

\newcommand{\mch}{\mathcal{H}}

\newcommand{\mcq}{\mathcal{Q}}
\newcommand{\mcu}{\mathcal{U}}

\newcommand{\parr}[1]{\frac{\partial #1}{\partial r}\!}
\newcommand{\parrdeux}[1]{\frac{\partial^2 #1}{\partial r^2}\!}
\newcommand{\parubar}[1]{\frac{\partial #1}{\partial \ubar}\!}
\newcommand{\parubarj}[1]{\frac{\partial^j #1}{\partial \ubar^j}\!}

\newcommand{\wbar}{\underline{w}}

\newcommand{\ubar}{{\underline{u}}}

\newcommand{\errp}{\mathrm{Err}[\Psip]}
\newcommand{\errm}{\mathrm{Err}[\Psim]}

\newcommand{\carterm}{\mcq_{-2}}

\newcommand{\nablabar}{\overline{\nabla}}

\newcommand{\R}{\mathbb R}

\newcommand{\dee}{\mathrm{d}}

\newcommand{\ch}{\mathcal{CH_+}}
\newcommand{\eh}{\mathcal{H_+}}

\newcommand{\ener}{\mathbf{e}}

\newcommand{\drond}{\mathring{\eth}}
\newcommand{\teuk}{\mathbf{T}}
\newcommand{\Psihat}{\widehat{\psi}}
\newcommand{\psihat}{\widehat{\psi}}

\newcommand{\ethat}{\widehat{e_3}}
\newcommand{\eqhat}{\widehat{e_4}}
\newcommand{\un}{\mathbf{I}}
\newcommand{\deux}{\mathbf{II}}
\newcommand{\trois}{\mathbf{III}}

\newcommand{\rb}{r_\mathfrak{b}}

\newcommand{\Psip}{\psi_{+2}}

\newcommand{\Psim}{\psi_{-2}}
\newcommand{\psim}{\psi_{-2}}

\newcommand{\Psihatp}{\Psihat_{+2}}
\newcommand{\psihatp}{\psihat_{+2}}

\newcommand{\Psihatm}{\Psihat_{-2}}
\newcommand{\psihatm}{\Psihat_{-2}}

\newcommand{\sfrak}{2}

\newcommand{\ansatzm}{\frac{1}{\ubar^7}\sum_{|m|\leq2}Q_{m,\sfrak}Y_{m,\sfrak}^{-\sfrak}(\cos\theta)e^{im\phi_{+}}}
\newcommand{\ansatzhatp}{\frac{\Delta^{-2}(u,\ubar)}{\ubar^{7}}\sum_{|m|\leq2}A_m(r_-)Q_{m,\sfrak}e^{2imr_{mod}}Y_{m,\sfrak}^{+\sfrak}(\cos\theta)e^{im\phi_{-}}}

\newcommand{\sph}{\mathbb S}

\title{Generic linearized curvature singularity at the perturbed\\Kerr Cauchy horizon}
\author{Sebastian Gurriaran\footnote{Laboratoire Jacques-Louis Lions, Sorbonne Universit\'e, 75005 Paris, France \par \indent\hspace{0.26cm} Email adress: sebastian.gurriaran@sorbonne-universite.fr}}
\date{March 31, 2025}

\begin{document}
\maketitle
\begin{abstract}
We prove the precise asymptotics of the spin $-2$ Teukolsky field in the interior and along the Cauchy horizon of a subextremal Kerr black hole. Together with the oscillatory blow-up asymptotics of the spin $+2$ Teukolsky field proven in our previous work \cite{G24}, our result suggests that generic perturbations of a Kerr black hole build up to form a coordinate-independent curvature singularity at the Cauchy horizon. This supports the Strong Cosmic Censorship conjecture in Kerr spacetimes. Unlike in the spin $+2$ case, the spin $-2$ Teukolsky field is regular on the Cauchy horizon and the first term in its asymptotic development vanishes. As a result, the derivation of a precise lower bound for the spin $-2$ field is more delicate than in the spin $+2$ case, and relies on a novel ODE method based on a decomposition of the Teukolsky operator between radial and time derivatives.

\end{abstract}
{
  \hypersetup{linkcolor=black}
  \tableofcontents

}
\section{Introduction}
\subsection{Kerr black holes and Strong Cosmic Censorship}
\subsubsection{The Kerr black hole interior}

The Kerr metric describes a spacetime around and inside a non-charged rotating black hole. Its expression in Boyer-Lindquist coordinates $(t,r,\theta,\phi)$ is 
$$\mathbf{g}_{a,M}=-\frac{(\Delta-a^2\sin^2\theta)}{\Sigma}\dee t^2-\frac{4aMr}{\Sigma}\sin^2\theta \dee t \dee \phi+\frac{\Sigma}{\Delta}\dee r^2+\Sigma\dee \theta^2+\frac{(r^2+a^2)^2-a^2\sin^2\theta\Delta}{\Sigma}\sin^2\theta \dee \phi^2\:,$$
where $M>0$ and $aM$ are respectively the mass and the angular momentum of the black hole, and where  $$\Delta:=r^2-2rM+a^2,\quad \Sigma:=r^2+a^2\cos^2\theta.$$
The main features of the Kerr metric is that it is a 2-parameter family of stationary, axisymmetric, and asymptotically flat solutions of the Einstein vacuum equation of general relativity that writes
\begin{align}\label{eq:eve}
    \mathbf{R}_{\mu\nu}[\mathbf{g}]=0,
\end{align}
where $\mathbf{R}_{\mu\nu}[\mathbf{g}]$ is the Ricci tensor of the Lorentzian metric $\mathbf{g}$. 

In the context of this work, we will only consider subextremal Kerr black holes with non-zero angular momentum, i.e. such that 
$0<|a|<M$. In this setting, $\Delta=(r-r_+)(r-r_-)$ has two distincts real roots 
$$r_\pm=M\pm\sqrt{M^2-a^2},\quad 0<r_-<r_+.$$
The sets $\{r=r_\pm\}$ are in Boyer-Lindquist coordinate degeneracies. Defining instead Eddington-Finkelstein coordinates
$$u=r^*+t,\quad\ubar=r^*-t,\quad\phi_{\pm}=\phi\pm r_{mod},$$
where $\dee r^*/\dee r=(r^2+a^2)/\Delta$, $\dee r_{mod}/\dee r=a/\Delta$, one can prove\footnote{See for example \cite[Chapt. 2]{oneill}.} that the hypersurfaces
$$\mch'_+:=\{\ubar=-\infty\},\quad\eh:=\{u=-\infty\},\quad\ch:=\{\ubar=+\infty\},\quad\mathcal{CH}'_+:=\{u=+\infty\}$$
can be properly attached to the Lorentzian manifold $$(\R_u\times\R_\ubar\times\mathbb{S}^2, \mathbf{g}_{a,M})\simeq((r_-,r_+)\times\R_t\times\mathbb{S}^2, \mathbf{g}_{a,M}).$$ We call the resulting manifold the Kerr black hole interior, that admits as boundaries the event horizon of the black hole
$$\eh\cup\mch'_+=\{u=-\infty\}\cup\{\ubar=+\infty\}=\{r=r_+\},$$
and the Cauchy horizon 
$$\ch\cup\mathcal{CH}'_+=\{u=+\infty\}\cup\{\ubar=+\infty\}=\{r=r_-\}.$$
\subsubsection{Strong Cosmic Censorship conjecture}\label{section:scccccc}
The Kerr spacetime is globally hyperbolic only up to the Cauchy horizon. Indeed, it admits infinitely many extensions across the Cauchy horizon as a smooth solution of the Einstein vacuum equations \eqref{eq:eve}, see Figure \ref{fig:penrosediagram}. This gives rise to a failure of determinism as in that case, the data on any complete spacelike hypersurface inside an exact Kerr black hole does not determine its future development.

\begin{figure}[h!]
    \centering
    \includegraphics[scale=0.43]{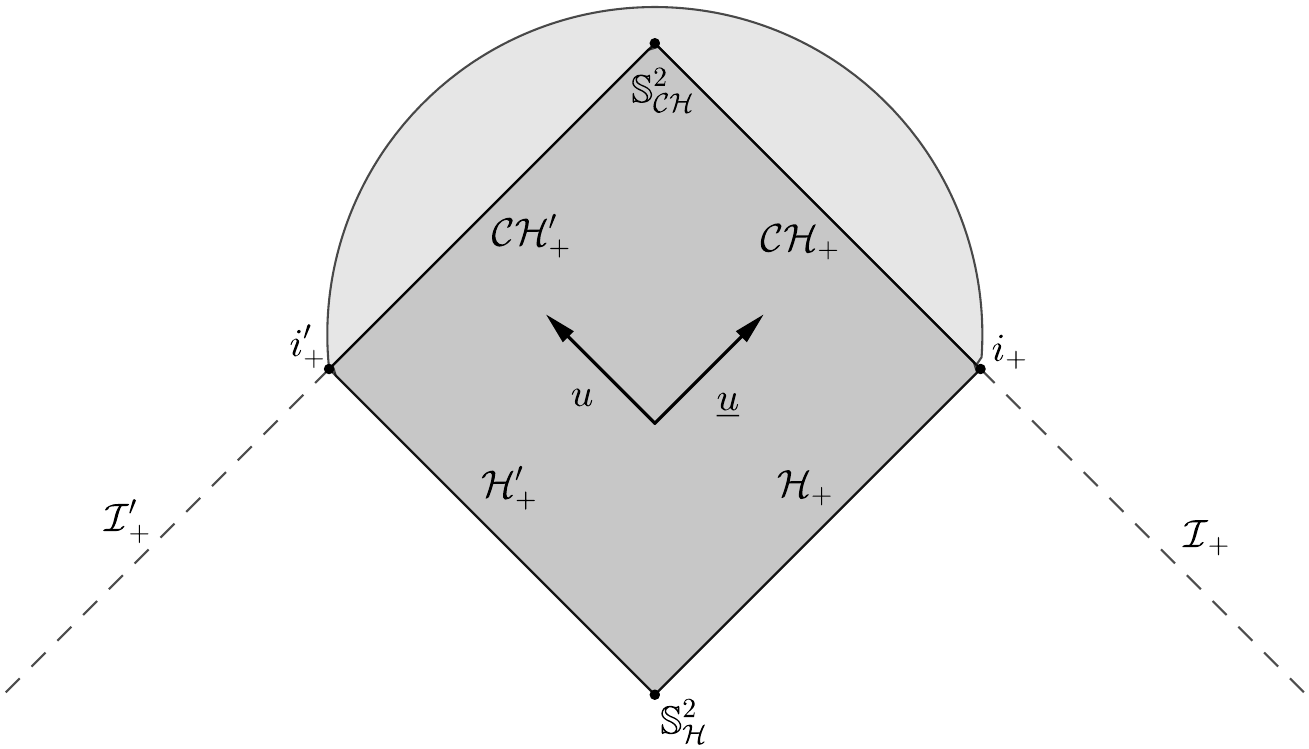}
    \caption{The Penrose diagram of Kerr spacetime, with the exterior of the black hole in white, the interior of the black hole in the darker shaded grey region, and the region with infinitely many extensions of Kerr in the lightest shaded grey region.}
    \label{fig:penrosediagram}
\end{figure}

It is expected that this behavior is an artifact of exact Kerr spacetimes, and vanishes upon small perturbations. This is the statement of the Strong Cosmic Censorship (SCC) conjecture in the case of Kerr. First formulated by Penrose \cite{sccpenrose}\footnote{For more modern versions of the Strong Cosmic Censorship conjecture, see \cite{christo, chruscc}.}, the SCC conjecture roughly states that the maximal globally hyperbolic development of generic initial data for the Einstein equations is inextendible as a \textit{regular} Lorentzian manifold. In other words, it states that the non-deterministic behavior of spacetimes with non-empty Cauchy horizons, like Kerr and Reissner-Nordström, is non-generic and vanishes upon small perturbations. Note that the $C^0$ version of SCC was disproved in Kerr by Dafermos and Luk \cite{stabC0}. In fact, they prove the $C^0$ stability of the Cauchy horizon and the proof suggests that the perturbed Cauchy horizon should be a weak null singularity as in \cite{weakluk}.

 The results of the present paper\footnote{Namely, the precise asymptotics of the spin $-2$ Teukolsky field near $\ch$.}, together with the results of our previous work \cite{G24}, suggest a precise and coordinate-invariant linearized version of this expected instability of the Kerr Cauchy horizon, see the discussion in Section \ref{section:curvatureinstab} for more details.

\subsection{Teukolsky equations}

In this work, we consider a linearized gravitational perturbation of a Kerr black hole. A classical result by Teukolsky \cite{teukolsky} states that two linearized curvature components decouple from the linearized gravity system in the Newman-Penrose formalism and satisfy wave equations called the \textit{Teukolsky equations}. 

To write the Teukolsky equations, we first define the pair of null vector fields 
\begin{align}\label{eq:gaugeteuk}
    n:=\left[(r^2+a^2)\partial_t-\Delta\partial_r+a\partial_\phi\right]/(2\Sigma),\quad l:=\frac{r^2+a^2}{\Delta}\partial_t+\partial_r+\frac{a}{\Delta}\partial_\phi,
\end{align}
which are aligned with the principal null directions of Kerr, and the complex vector field
\begin{align}\label{eq:mframe}
    m:=\frac{1}{\sqrt{2}(r+ia\cos\theta)}\left(ia\sin\theta\partial_t+\partial_\theta+\frac{i}{\sin\theta}\partial_\phi\right).
\end{align}
They satisfy
\begin{align}\label{eq:prodln}
    \mathbf{g}_{a,M}(l,n)=-1,\quad\mathbf{g}_{a,M}(m,\overline{m})=1,
\end{align}
and all other scalar products are zero. Then, denoting by $\dot{\mathbf{W}}$ the linearized Weyl tensor, Teukolsky discovered that the scalars
$$\psihatp:=\dot{\mathbf{W}}_{lmlm},\quad\psihatm:=(r-ia\cos\theta)^4\dot{\mathbf{W}}_{n\overline{m}n\overline{m}},$$
satisfy the following wave equations, for $s=\pm 2$:
\begin{align}
    & -\left[\frac{\left(r^2+a^2\right)^2}{\Delta}-a^2 \sin ^2 \theta\right] \partial_t^2\psihat_s -\frac{4 M a r}{\Delta}\partial_t\partial_\phi\psihat_s  -\left[\frac{a^2}{\Delta}-\frac{1}{\sin ^2 \theta}\right] \partial_\phi^2\psihat_s \nonumber \\
    & +\Delta^{-s} \partial_r\left(\Delta^{s+1} \partial_r \psihat_s\right)+\frac{1}{\sin \theta} \partial_\theta\left(\sin \theta \partial_\theta \psihat_s\right)+2 s\left[\frac{a(r-M)}{\Delta}+\frac{i \cos \theta}{\sin ^2 \theta}\right] \partial_\phi\psihat_s  \nonumber\\
    & +2 s\left[\frac{M\left(r^2-a^2\right)}{\Delta}-r-i a \cos \theta\right]\partial_t\psihat_s -\left[\frac{s^2 \cos ^2 \theta}{\sin ^2 \theta}-s\right]\psihat_s=0.\label{eq:premteuk}
\end{align}
Note that $n$ and $l$ are regular on\footnote{This can be seen using outgoing Eddington-Finkelstein coordinates.} $\ch$, thus $\Psihatm$ and $\Psihatp$ are projections of the linearized curvature tensor on a tetrad $(n, l,m,\overline{m})$ which is regular on $\ch$. We also define the rescaled scalars 
$$\Psip:=\Delta^2\Psihatp, \quad\Psim:=\Delta^{-2}\Psihatm,$$
which are projections of the linearized Weyl tensor on the tetrad $(\Delta^{-1}n,\Delta l,m,\overline{m})$ that is regular on $\eh$. These rescaled scalars satisfy a rescaled version of the Teukolsky equation, see Section \ref{section:system}.
 
Note that in \cite{KS21}, the proof of the non-linear stability of the exterior of slowly rotating Kerr black holes relied on decay estimates for solutions of the Teukolsky equations in the non-linear setting proven in \cite{KSwaveeq}. For a review of the literature concerning the Teukolsky equations, see the introduction in \cite{KSwaveeq}. In the present work, we will use the Teukolsky equations to prove a precise quantitative and coordinate-independent curvature blow-up on the perturbed Kerr Cauchy horizon, in the linearized setting.

\subsection{Coordinate-independent curvature singularity}
\label{section:curvatureinstab}
As mentioned in Section \ref{section:scccccc}, the Kerr Cauchy horizon was proven to be $C^0$ stable\footnote{More precisely, it was proven that small perturbations of the Kerr interior have a future null boundary (the perturbed Cauchy horizon) in a neighborhood of timelike infinity, across which the metric extends continuously. Moreover the perturbed metric is $C^0$-close to the Kerr metric.} in \cite{stabC0}. Thus, the expected instability of the Kerr Cauchy horizon must be at a higher regularity level. In the recent paper by Sbierski \cite{sbierskiinextdernier}, it is shown that if one manages to prove, in the non-linear setting, the blow-up of the curvature tensor contracted with small perturbations of some continuous fixed vector fields\footnote{These contractions of the curvature tensor should be thought as perturbations of $\Psihatp$ in the non-linear setting.}, integrated on a spacetime region approaching the perturbed Cauchy horizon, then the spacetime is $C^{0,1}_{loc}$-inextendible across the Cauchy horizon. This result suggests that a proof of the $C^{0,1}_{loc}$ version of SCC conjecture in Kerr spacetimes could be obtained by proving the generic (integrated) blow-up of the non-linear analog of the spin $+2$ Teukolsky scalar (which is a projection of the curvature on some fixed vector fields, and where the blow-up would have to hold when projecting the curvature tensor on perturbations of these vector fields). 

\textbf{In this section, we outline a different strategy to address the instability of the Kerr Cauchy horizon, by obtaining a generic coordinate-invariant curvature blow-up result at the perturbed Cauchy horizon, from the precise control of both spin $\pm 2$ Teukolsky scalars.}

The most classical way to detect coordinate-invariant singularities in general relativity is by proving blow-up of the Kretschmann scalar\footnote{The most famous example is the blow-up of the Kretschmann scalar of Schwarzschild spacetime at $r=0$, proving that the singularity $\{r=0\}$ of Schwarzschild is a true singularity.}, which is the coordinate-independent curvature scalar
$$\mathbf{K}:=\mathbf{R}^{\alpha\beta\gamma\delta}\mathbf{R}_{\alpha\beta\gamma\delta},$$
where $\mathbf{R}_{\alpha\beta\gamma\delta}$ is the curvature tensor of $\mathbf{g}$. One can express $\mathbf{K}$ as projections of the curvature tensor on a null frame, and this expression is by definition \textit{frame independent}.

Let $(n,l,m,\overline{m})$ be any null tetrad satisfying \eqref{eq:prodln}. Then defining the Weyl scalars
\begin{align*}
    \widehat{\Psi}_{-2}:=\mathbf{W}_{n\overline{m}n\overline{m}},\quad\widehat{\Psi}_{-1}:=\mathbf{W}_{ln\overline{m}n},\quad\widehat{\Psi}_0:=\mathbf{W}_{lm\overline{m}n},\quad\widehat{\Psi}_{+1}:=\mathbf{W}_{lnlm},\quad\widehat{\Psi}_{+2}:=\mathbf{W}_{lmlm},
\end{align*}
where $\mathbf{W}$ is the Weyl tensor, we have in a vacuum spacetime\footnote{Note that if the metric solves the Einstein vacuum equation \eqref{eq:eve}, then $\mathbf{R}_{\alpha\beta\gamma\delta}=\mathbf{W}_{\alpha\beta\gamma\delta}$.}
$$\mathbf{K}=16\Real(\widehat{\Psi}_{+2}\widehat{\Psi}_{-2}+3\widehat{\Psi}_0^2-4\widehat{\Psi}_{+1}\widehat{\Psi}_{-1}).$$
Moreover, we expect that in generic perturbations of the Kerr black hole interior, for a frame $(n,l,m,\overline{m})$ which is a perturbation of \eqref{eq:gaugeteuk}-\eqref{eq:mframe}, the curvature terms $\widehat{\Psi}_0^2$ and $\widehat{\Psi}_{+1}\widehat{\Psi}_{-1}$ are negligible\footnote{See \cite{stabC0} for estimates supporting this claim in the non-linear setting, and \cite{orikretschmann} for a discussion in the linearized setting.} compared to $\widehat{\Psi}_{+2}\widehat{\Psi}_{-2}$ near $\ch$ so that we have
$$\mathbf{K}= 16\Real(\widehat{\Psi}_{+2}\widehat{\Psi}_{-2}(1+o(1)))$$
near the Cauchy horizon. In the linearized setting, we have by definition $\dot{\widehat{\Psi}}_{+2}=\Psihatp$, $\dot{\widehat{\Psi}}_{-2}=(r-ia\cos\theta)^{-4}\Psihatm$. \textbf{The product $\Psihatp\Psihatm$ should thus be thought as the leading-order term of the linearized Kretschmann curvature invariant.} Although we know from \cite{G24} that in the gauge \eqref{eq:gaugeteuk}, $\Psihatp$ generically blows-up at $\ch$ while violently oscillating, this discussion highlights the importance of proving in addition that $\Psihatm$ generically does not vanish on $\ch$. This fact is a direct consequence of the main result of this work, which implies that the product $\Psihatp\Psihatm$ generically blows-up at $\ch$ while oscillating violently, see \eqref{eq:anstK}. 

In conclusion, the linearized analysis of \cite{G24} for $\Psihatp$ and of the present paper for $\Psihatm$ thus supports the conjecture that generic perturbations of Kerr black holes present a Kretschmann curvature singularity on the perturbed Cauchy horizon, see \cite{orikretschmann,simunum} for corresponding heuristics and numerical results. We give more details on the expected oscillating behavior of the curvature blow-up at the end of Section \ref{subsection:mainthm}.

\subsection{Rough version of the main theorem and consequences}\label{subsection:mainthm}

We now state the simplified version of our main result, see Theorem \ref{thm:main} for the precise formulation.
\begin{thm}[Main theorem, rough version]\label{thm:mainrough} We consider an initially smooth and compactly supported linearized gravitational perturbation of a Kerr black hole such that $0<|a|<M$. We denote $\Psihatm$ the spin $-2$ Teukolsky scalar obtained in a principal null frame regular at $\ch$. Then $\Psihatm$ is regular on $\ch$, and satisfies the following asymptotic behavior near $\ch$: 
\begin{align}\label{eq:roughasympt}
        \Psihatm\sim\frac{1}{u^8}\sum_{|m|\leq 2}C_mY_{m,2}^{-2}(\cos\theta)e^{im\phi_-},
    \end{align}
where the constants $C_m$ are defined in \eqref{eq:cm}, depend only on the initial data and black hole parameters and are generically non-zero, and the functions $Y_{m,2}^{-2}(\cos\theta)e^{im\phi_-}$ are the spin $-2$ spherical harmonics which are regular on $\ch$ (see Sections \ref{section:spinweighted} and \ref{section:mainthm} for more details). 
\end{thm}
\begin{rem} A few remarks are in order: 
\begin{enumerate}
    \item At first glance, the regularity of $\Psihatm$ at $\ch$ may seem to support a curvature \textit{stability} statement for the Kerr Cauchy horizon. But the asymptotic behavior \eqref{eq:roughasympt} implies in particular that $\Psihatm$ is generically non-zero near $\ch\cap\{u\ll -1\}$. As discussed in Section \ref{section:curvatureinstab}, together with the blow-up asymptotics of $\Psihatp$ near $\ch$  proven in \cite{G24}, this suggests the presence of a coordinate-independent curvature singularity at the Cauchy horizon of a generically perturbed Kerr black hole.
    \item From Price's law results in \cite{MZ23}, \cite{Millet23}, which predict a generic asymptotic behavior in $1/\ubar^7$ on the event horizon for $\Psim$, and the analysis in \cite{G24} for $\Psihatp$, one could expect that the generic asymptotic behavior for $\Psihatm$ in the black hole interior would be in $1/u^7$. But it turns out that, unlike for the spin $+2$ case \cite{G24}, the first term in the asymptotic development of $\Psihatm$ in inverse powers of $u$ vanishes at $\ch$. The crucial part of the proof of Theorem \ref{thm:mainrough} is thus to track the next term in the asymptotic development of $\Psihatm$, and to prove that it is generically non-zero at $\ch$. Note that the asymptotic behavior $\Psihatm\sim1/u^8$ was first heuristically predicted by Ori \cite{ori}.
    \item The proof is based on the energy method introduced in \cite{scalarMZ} and on a novel ODE method that relies on a decomposition of the Teukolsky operator between radial and time derivatives. For more details and for a comparison with the method used for $\Psihatp$, see Section \ref{section:structure}.
    \item Just like in the spin $+2$ case in \cite{G24}, we do not assume higher order decay of the higher modes $(\Psim)_{\ell\geq 3}$ on the event horizon. 
\end{enumerate}
\end{rem}

We now complete the discussion of Section \ref{section:curvatureinstab} by analyzing the strength of the curvature singularity on the Cauchy horizon suggested by the asymptotic behaviors of $\Psihatp$ and $\Psihatm$. The main result of \cite{G24} yields
\begin{align}\label{eq:twomainresults1}
    \Psihatp\sim\ansatzhatp
\end{align}
near $\ch$, where $r_{mod}$, $Q_{m,2}$, $A_m(r)$ are defined in \cite{G24}, and satisfy $r_{mod}\sim a\log(r-r_-)$ near $\ch$, $A_m(r_-)\neq 0$ for $m=\pm 1,\pm2$ and the constants $Q_{m,2}$ depend on the initial data and are generically non-zero. Using \eqref{eq:twomainresults1} and \eqref{eq:roughasympt} we get 
\begin{align}\label{eq:anstK}
    \Psihatp\Psihatm\sim\frac{\Delta^{-2}(u,\ubar)}{\ubar^{7}u^8}\left(\sum_{|m|\leq2}A_m(r_-)Q_{m,\sfrak}e^{2imr_{mod}}Y_{m,\sfrak}^{+\sfrak}(\cos\theta)e^{im\phi_{-}}\right)\left(\sum_{|m|\leq 2}C_mY_{m,2}^{-2}(\cos\theta)e^{im\phi_-}\right),
\end{align}
near $\ch$, which generically blows-up at $\ch$ (because $\Delta^{-2}\sim e^{\frac{r_+-r_-}{2Mr_-}\ubar}$ on $\{u=cst\}$) while oscillating violently. In view of the discussion in Section \ref{section:curvatureinstab}, this supports the following predictions by Ori \cite{orikretschmann}:
\begin{enumerate}
    \item Generic perturbations of a Kerr black hole have a Kretschmann scalar which is singular on the perturbed Cauchy horizon.
    \item On every slice $\{u=u_0\}$ for sufficiently negative values of $u_0\ll-1$, the generic curvature blow-up profile is a violently oscillating exponential factor in $\ubar$ divided by a polynomial in $\ubar$.
\end{enumerate}
\subsection{Perturbations of black hole interiors}
\subsubsection{Price's law-type results}
To prove the precise asymptotic behavior of solutions of the spin $-2$ Teukolsky equation in the Kerr black hole interior, the starting point is a precise version of Price's law for Teukolsky.
This law gives polynomial upper and lower bounds on the event horizon for solutions of wave equations arising from compactly supported initial data, see the original works on Price's law \cite{price1, price2, price3, price4}. A version of Price's law for Teukolsky equations in Kerr was heuristically found by Barack and Ori in \cite{barackori}. 

In this paper we use the Price's law-type precise asymptotics of the spin $-2$ Teukolsky field on $\mch_+$ proven by Ma and Zhang \cite{MZ23}\footnote{The Price's law in \cite{MZ23} holds for $|a|\ll M$, and for $|a|<M$ conditionally on an energy-Morawetz bound. This energy-Morawetz estimate has since been proved for $|a|<M$ by Teixeira da Costa and Shlapentokh-Rothman in \cite{TDCSR2, TDCSR1}, so that the Price's law in \cite{MZ23} holds for the full subextremal range $|a|<M$.}. For a proof of the precise asymptotics of solutions to the Teukolsky equations in the full subextremal range $|a|<M$ with a spectral point of view, see \cite{Millet23}. For a complete review of Price's law-type results, see the introduction in \cite{MZ23}.

\subsubsection{Previous results on the interior of perturbed black holes}

The first studies of linear perturbations of Kerr and Reissner-Nordström black holes interiors explicited specific perturbations that become unbounded in some way at the Cauchy horizon, see for example \cite{mac}. A heuristic power law for scalar waves in Kerr spacetimes was obtained by Ori in \cite{oriscalar}. For the Teukolsky equations, the oscillatory blow-up asymptotic of the spin $+2$ Teukolsky scalar in the interior of Kerr black holes was first predicted heuristically by Ori \cite{ori}. In the same paper, the asymptotic behavior $\Psihatm\sim 1/u^8$ for the spin $-2$ Teukolsky scalar was also predicted, writing the azimuthal $m$-mode of the solution as a late-time expansion ansatz. 

Next, inside Reissner-Nordström black holes, a rigorous boundedness statement for solutions of the scalar wave equation  was proven in \cite{franzen}. For a scattering approach to Cauchy horizon instability in Reissner-Nordström, as well as an application to mass inflation, see \cite{lukk}. The blow-up of the energy of generic scalar waves inside Reissner-Nordström black holes was obtained in \cite{RNscalar}. The non-linear problem of Cauchy horizon instability in spherical symmetry was treated in \cite{dafscc}, with results later extended in \cite{RNNLI, RNNLII}. Analog results for the Einstein-Maxwell-Klein-Gordon system were proven in \cite{vdmmmm, vdminsta}, see also \cite{weakvdm}.

Now we presents results for the scalar wave in the interior of Kerr black holes. A generic blow-up result on the Cauchy horizon for the energy of scalar waves was obtained in \cite{kerrwave}. It was proven in \cite{hintzkerrwave} in the slowly rotating case that scalar waves remain bounded up to the Cauchy horizon. This result was then extended to the full subextremal range in \cite{franzen2}. The paper \cite{daf3} constructs solutions that remain bounded but have infinite energy at the Cauchy horizon. More recently, \cite{scalarMZ} introduced a robust physical-space method to obtain the precise asymptotics of the scalar field in the interior of Kerr.

Concerning the Teukolsky equations in the Kerr black hole interior, the method of proof of \cite{kerrwave} was extended to the spin $+2$ Teukolsky equation in \cite{sbierski}, where the generic blow-up of a weighted $L^2$ norm on a hypersurface transverse to the Cauchy horizon was proven, relying on frequency analysis. The precise pointwise oscillatory blow-up asymptotic behavior of the spin $+2$ Teukolsky scalar was proven by the author in \cite{G24}, extending the method of proof of \cite{scalarMZ}.

For a more complete account of the results related to black hole interior perturbations, for example in the cosmological setting or in Schwarzschild spacetime, we refer to the introduction in \cite{sbierski}.

\subsection{Structure of the proof}\label{section:structure}
The proof of \eqref{eq:roughasympt} is based on energy estimates, on a precise decomposition of the Teukolsky operator, and on the following Price's law-type result on the event horizon:
\begin{align}
    \left|\partial_t^j\left(\Psim-\ansatzm\right)\right|\lesssim\ubar^{-7-j-\delta}\label{eq:decayerrhorizon}
\end{align}
on $\eh\cap\{\ubar\geq 1\}$, where the constants $Q_{m,2}$ depend on the initial data, and where the functions $Y^{-2}_{m,2}(\cos\theta)e^{im\phi_+}$ are the spin $-2$ spherical harmonics. It was proven in \cite{MZ23} that \eqref{eq:decayerrhorizon} holds for initially smooth and compactly supported solutions of the spin $-2$ Teukolsky equation in Kerr spacetime. On top of \eqref{eq:decayerrhorizon}, fixing $\rb\in(r_-,r_+)$ close to $r_-$ and $\gamma>0$ small, we will use a decomposition of the Kerr black hole interior in different regions: 
\begin{itemize}
    \item The redshift region $\un=\{\rb\leq r\leq r_+\}\cap\{\ubar\geq 1\}$.
    \item The blueshift region $\deux\cup\trois=\{r_-\leq r\leq \rb\}\cap\{w\leq w_{\rb,\gamma}\}$, where $w=u-r+r_-$,
    $$\deux=\{r_-\leq r\leq \rb\}\cap\{2r^*\leq\ubar^\gamma\}\cap\{w\leq w_{\rb,\gamma}\},\quad\trois=\{r_-\leq r\leq \rb\}\cap\{2r^*\geq\ubar^\gamma\}\cap\{w\leq w_{\rb,\gamma}\},$$
    and where $w_{\rb,\gamma}=2\rb^*-(2\rb^*)^{1/\gamma}-\rb+r_-$, see Figure \ref{fig:regions-2} which illustrates regions $\un,\deux,\trois$. 
\end{itemize}
The proof is then structured as follows:
\begin{enumerate}
    \item First, in Proposition \ref{prop:lapremiere} we use the redshift energy estimates results already proven in \cite{G24} for the spin $-2$ Teukolsky equation in redshift region $\un$, to propagate the decay \eqref{eq:decayerrhorizon} from $\eh$ to the whole region $\un$.\label{step:un}
    \item Next, in Section \ref{section:upperbounds} we implement an energy method for the rescaled spin $-2$ Teukolsky operator in the blueshift region $\deux\cup\trois$. Combining the symmetries of Kerr spacetime and these energy estimates yields pointwise upper bound results for the spin $-2$ Teukolsky field near the Kerr Cauchy horizon.
    \item Using this energy method, in Propositions \ref{prop:upperboundII} and \ref{prop:almostsharpdrout} we obtain a bound of the type, for $j\geq 0$:
    \begin{align}\label{eq:controlstructure}
        \left|\partial_t^j\left(\Psihatm-\anshatm\right)\right|\lesssim|u|^{-7-j-\delta},
    \end{align}
    in the blueshift region $\deux\cup\trois$. Like in \cite{G24}, this relies on the following key computational fact:
    $$\teuk_{-2}\left(\Psim-\ansatzm\right)=-\teuk_{-2}\left(\ansatzm\right)=O(\ubar^{-8}),$$
    where $\teuk_{-2}$ is the spin $-2$ Teukolsky operator. We also use in particular the fact that the scalar $-\Delta$ decays exponentially in $\ubar$ in region $\trois$. The bound \eqref{eq:controlstructure} is precise in any region with $r$ bounded away from $r_-$, but is imprecise near $\ch$ (as $\Delta^2/\ubar^7$ becomes negligible compared to $|u|^{-7-\delta}$ and vanishes on $\ch$).
    \item\label{step:4} The crux of the proof of the main theorem is to then re-inject the bounds \eqref{eq:decayerrhorizon} in $\un$, and \eqref{eq:controlstructure} in $\deux\cup\trois$, in the Teukolsky equation to recover a lower bound for $\Psihatm$ which is precise at $\ch$. We do this in Section \ref{section:precise} in the following way: we first rewrite the Teukolsky equation in $\un$ as
    $$\mathbf{T}_{-2}^{[\partial_r]}\Psim+\mathbf{T}_{-2}^{[\partial_t]}\Psim=0,$$
    where $\mathbf{T}_{-2}^{[\partial_t]}$ is the part of the Teukolsky operator which contains $\partial_t$ derivatives, and $\mathbf{T}_{-2}^{[\partial_r]}$ is the remaining part, see Section \ref{section:decdtdr} for the precise expressions of $\mathbf{T}_{-2}^{[\partial_r]},\mathbf{T}_{-2}^{[\partial_t]}$. We then define 
    $$\errm:=\Psim-\ansatzm.$$
    It can be checked using \eqref{eq:decayerrhorizon} (which holds in $\un$ by Step \ref{step:un}) that the Teukolsky equation implies
    \begin{align}\label{eq:donneedo}
        \mathbf{T}_{-2}^{[\partial_r]}\errm=-\mathbf{T}_{-2}^{[\partial_t]}\left(\ansatzm\right)+O(\ubar^{-8-\delta}),\quad\text{ in }\un.
    \end{align}
    The first term on the right-hand-side of \eqref{eq:donneedo} is a source term which can be precisely computed and which behaves like $1/\ubar^8$. It turns out that by using the precise expression of the operator $\mathbf{T}_{-2}^{[\partial_r]}$, when projecting \eqref{eq:donneedo} on a spin weighted $(\ell,m)$ mode, we obtain an ODE in $r$ for the $(\ell,m)$ component of $\errm$, with a precise source term in $1/\ubar^8$. This ODE can then be integrated to obtain a precise expression for $(\errm)_{\ell,m}$ in $\un$. Then, we do a similar procedure in region $\deux\cup\trois$ for the quantity $\errhatm=\Delta^2\errm$ using \eqref{eq:controlstructure}, and matching the solutions of the ODEs on the boundary $\{r=\rb\}$ between $\un$ and $\deux\cup\trois$, we obtain a precise expression in $1/u^8$ for $(\errhatm)_{\ell=2}$ (and thus also for $(\Psihatm)_{\ell=2}$) near $\ch$, and the following precise upper bound:
    $$|(\Psihatm)_{\ell\geq 3}|\lesssim |u|^{-8-\delta/2}+r-r_-.$$
    Note that for technical reasons, we deal with the $\ell=3$ mode and each $\ell\geq 4$ mode of $\Psihatm$ separately. Moreover, the summation of the bounds obtained for each $\ell\geq 4$ mode must be done with care, see the first item of Remark \ref{rem:jaimebien}.
\end{enumerate}
We conclude this section by outlining the difference of treatments between $\Psihatm$ in this paper and $\Psihatp$ in \cite{G24}. It was proven by the author in \cite{G24} that the quantities
\begin{align*}
    \errm\quad\text{and}\quad\errp:=\Psip-\ansatzp,
\end{align*}
are bounded by $\ubar^{-7-\delta}$ on $\{r=\rb\}\cap\{\ubar\geq 1\}$. Roughly stated, the energy method adapted from \cite{scalarMZ} yields
$$|\errp|\lesssim\ubar^{-7-\delta},\quad|\Delta^2\errm|\lesssim|u|^{-7-\delta},$$
in $\deux\cup\trois$ (see \cite{G24} for the first bound and Section \ref{section:upperbounds} for the second bound) which implies in $\deux\cup\trois$, after renormalization:
\begin{align}
    &\Psihatp=\ansatzhatp+O(\Delta^{-2}\ubar^{-7-\delta}),\label{eq:depoin}\\
    &\Psihatm=\anshatm+O(|u|^{-7-\delta}).\label{eq:troipoin}
\end{align}
Statement \eqref{eq:depoin} is the main theorem of \cite{G24} and implies oscillatory blow-up of $\Psihatp$ at $\ch$, but on the other hand we see that \eqref{eq:troipoin} is not precise enough, as near $\ch$ the error term $O(|u|^{-7-\delta})$ becomes the leading-order term in \eqref{eq:troipoin}. This is why we need a much finer analysis based on the ODE method in Section \ref{section:precise} (outlined in Step \ref{step:4} above), to prove the precise $1/u^8$ asymptotics of this error term. Moreover, unlike in the spin $+2$ case in \cite{G24} where there was no need for mode projection, in the present paper we need to track the decay of the different angular modes of $\Psihatm$. More precisely, we have to provide a separate analysis for the modes $\ell=2$, $\ell=3$, and $\ell\geq 4$ of $\Psihatm$.

\subsection{Overview of the paper}
In Section \ref{section:preliminaries}, we introduce the Kerr black hole interior geometry, as well as some preliminary results. Next, in Section \ref{section:mainthm}, we state the assumptions and the precise version of the main theorem of the paper, Theorem \ref{thm:main}. We begin the proof of Theorem \ref{thm:main} in Section \ref{section:upperbounds} by implementing the energy method for the spin $-2$ Teukolsky equation near $\ch$. Finally, in Section \ref{section:precise}, we find the precise asymptotics of $(\Psihatm)_{\ell=2}$, and precise upper bounds for $(\Psihatm)_{\ell=3}$ and $(\Psihatm)_{\ell\geq 4}$ near $\ch$.

\subsection{Acknowledgments}
The author would like to thank Jérémie Szeftel for helpful suggestions and for encouraging him to work on this problem. The author is partially supported by ERC-2023 AdG 101141855 BlaHSt.
\section{Preliminaries}\label{section:preliminaries}
\subsection{Notations}
We begin by introducing some notations.
\begin{itemize}
    \item Throughout the paper, ‘RHS' and ‘LHS' refer respectively to ‘right-hand side' and ‘left-hand side'.
    \item  If $L$ is an operator acting on a spin-weighted scalar $\psi$, and if $N$ is any norm, for $k\in\mathbb{N}$ we use the convention
    $$N(L^{\leq k}\psi):=\sum_{i=0}^kN(L^{i}\psi),$$
    so that having a bound for $N(L^{\leq k}\psi)$ is equivalent\footnote{In this paper, $k$ will be chosen such that $k\leq k_0$ for some fixed large enough $k_0$.} to having a bound for $N(L^i\psi)$, for all $0\leq i\leq k$.
    \item If $f$ and $g$ are two non-negative scalars, we write $f\lesssim g$ if there is a constant $C>0$ which depends only on the black hole parameters $a,M$, on the initial data, and on the constants $\rb$, $\gamma$, such that $f\leq Cg$ in the region considered. We write $f=O(g)$ when $|f|\lesssim|g|$. We write $f\sim g$ when $f\lesssim g$ and $g\lesssim f$.
\end{itemize}
\subsection{Geometric preliminaries}\label{section:geometricprel}
\subsubsection{Vector fields}
We begin by defining a rescaled and more convenient version of the null pair $(n,l)$ defined in \eqref{eq:gaugeteuk}:
\begin{align*}
    e_3:=\frac{1}{2}\left(-\partial_r+\frac{r^2+a^2}{\Delta}\partial_t+\frac{a}{\Delta}\partial_\phi\right),\quad\quad e_4:=\frac{1}{2}\left(\frac{\Delta}{r^2+a^2}\partial_r+\partial_t+\frac{a}{r^2+a^2}\partial_\phi\right).
\end{align*}
This pair satisfies
$$\mathbf{g}_{a,M}(e_3,e_3)=\mathbf{g}_{a,M}(e_4,e_4)=0,\quad \mathbf{g}_{a,M}(e_3,e_4)=-\frac{\Sigma}{2(r^2+a^2)},$$
and is regular on $\mch_+$, as can be seen by expressing $e_3$ and $e_4$ with respect to the ingoing Eddington-Finkelstein coordinate vector fields, see Section \ref{section:coordsss}. Let $\mu$ be the scalar function
$$\mu:=\frac{\Delta}{r^2+a^2}$$
where $\Delta=r^2-2Mr+a^2=(r-r_+)(r-r_-)$. Notice that both $\mu$ and $\Delta$ are non-positive in the Kerr black hole interior and vanish only on the event and Cauchy horizons.  We now define the following rescaled null pair, which is regular on $\ch$:
$$\ethat:=(-\mu)e_3,\quad\quad\eqhat:=(-\mu)^{-1}e_4.$$
For convenience, we introduce the following notations for the Kerr Killing vector fields, expressed in Boyer-Lindquist coordinates,
$$T:=\partial_t,\quad\quad\Phi:=\partial_\phi.$$
\subsubsection{Scalar functions and coordinates}\label{section:coordsss}
The Boyer-Lindquist (B-L) coordinates are singular at both the event and Cauchy horizons. In this paper, we will use both the Eddington-Finkelstein coordinates and double null-like coordinates $(u,\ubar,\theta,\phi_{\pm})$. The classical definition of the tortoise coordinate $r^*$ in the Kerr black hole interior is:
$$r^*(r)=\int_M^r\mu^{-1}(r')\dee r'.$$
We have $r^*\to -\infty$ as $r\to r_+$ and $r^*\to +\infty$ as $r\to r_-$. Indeed, defining $\kappa_+$ and $\kappa_-$  the surface gravities of the event and Cauchy horizons:
$$\kappa_+:=\frac{r_+-r_-}{4Mr_+}>0,\quad\kappa_-:=\frac{r_--r_+}{4Mr_-}<0,$$
then the asymptotics of $r^*$ at the horizons are given by
\begin{align}\label{eqn:rstar}
    r^*=\frac{1}{2\kappa_+}\ln(r_+-r)+h_+(r)=\frac{1}{2\kappa_-}\ln(r-r_-)+h_-(r),
\end{align}
where $h_\pm(r)$ has a finite limit as $r\to r_\pm$. Notice that \eqref{eqn:rstar} implies that, for $r$ close to $r_+$, 
$$-\Delta\sim\exp(-2\kappa_+|r^*|),$$
while for $r$ close to $r_-$, 
$$-\Delta\sim\exp(-2|\kappa_-|r^*).$$
Next, we define the coordinates
$$\ubar:=r^*+t,\quad u:=r^*-t.$$
The range of the coordinates $u,\ubar,r^*$ is indicated on Figure \ref{fig:range}.
\begin{figure}[h!]
    \centering
    \includegraphics[scale=0.55]{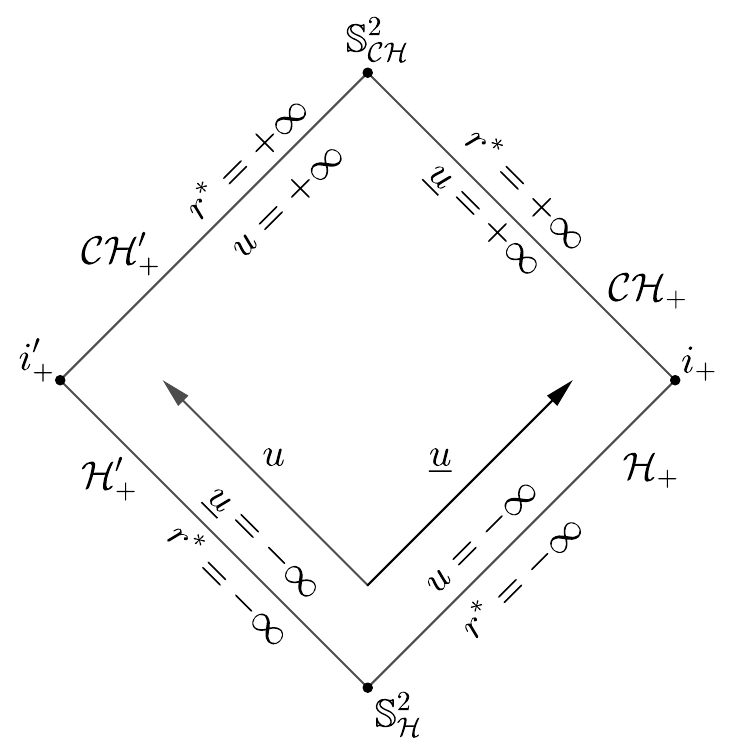}
    \caption{The range of the coordinates $u,\ubar,r^*$ in the Kerr black hole interior.}
    \label{fig:range}
\end{figure}

As in \cite{scalarMZ},  we define the function $r_{mod}$ by 
$$r_{mod}(r)=\int_M^r\frac{a}{\Delta(r')}\dee r'.$$
We then define the ingoing and outgoing angular coordinates
$$\phi_{+}:=\phi+r_{mod}\:\:\text{mod}\:2\pi,\quad \phi_{-}:=\phi-r_{mod}\:\:\text{mod}\:2\pi.$$
The coordinates $\ubar$ and $\phi_+$ are regular on $\mch_+=\{u=-\infty\}$ and $\mathcal{CH}'_+=\{u=+\infty\}$, while $u$ and $\phi_-$ are regular on $\mch'_+=\{\ubar=-\infty\}$ and $\ch=\{\ubar=+\infty\}$.

\textbf{Ingoing Eddington-Finkelstein coordinates.} The ingoing Eddington-Finkelstein coordinates are $(\ubar_{\mathrm{in}}=\ubar, r_{\mathrm{in}}=r, \theta_{\mathrm{in}}=\theta, \phi_{\mathrm{in}}=\phi_{+})$. It is a set of coordinates regular on $\eh$ and $\mathcal{CH}'_+$. The coordinate vector fields are 
\begin{align}\label{eqn:coordEFin}
    \partial_{r_\mathrm{in}}=-2e_3,\quad\partial_{\ubar_\mathrm{in}}=T,\quad  \partial_{\theta_{\mathrm{in}}}=\partial_\theta,\quad \partial_{\phi_{\mathrm{in}}}=\Phi.
\end{align}

\textbf{Outgoing Eddington-Finkelstein coordinates.} The outgoing Eddington-Finkelstein coordinates are $(u_{\mathrm{out}}=u, r_{\mathrm{out}}=r, \theta_{\mathrm{out}}=\theta,\phi_{\mathrm{out}}= \phi_{-})$. It is a set of coordinates regular on $\ch$ and $\mathcal{H}'_+$. The coordinate vector fields are 
\begin{align}\label{eqn:coordEFout}
    \partial_{r_\mathrm{out}}=-2\eqhat,\quad\partial_{u_\mathrm{out}}=T,\quad  \partial_{\theta_{\mathrm{out}}}=\partial_\theta,\quad \partial_{\phi_{\mathrm{out}} }=\Phi.
\end{align}

\textbf{Double null-like coordinate systems.} We use the ingoing double null-like coordinates $(\ubar,u,\theta,\phi_{+})$, as in \cite{scalarMZ}. The coordinate vector fields are respectively
\begin{align}\label{eqn:coordDNLin}
    \partial_\ubar=e_4-\frac{a}{r^2+a^2}\Phi,\quad\partial_u=-\mu e_3, \quad\partial_\theta,\quad\partial_{\phi_{+}}=\Phi.
\end{align}
The equivalent outgoing double null-like coordinates are $(\ubar,u,\theta,\phi_{-})$, with coordinate vector fields
\begin{align}\label{eqn:coordDNLout}
    \partial_\ubar=e_4,\quad\partial_u=-\mu e_3+\frac{a}{r^2+a^2}\Phi, \quad\partial_\theta,\quad\partial_{\phi_{-}}=\Phi.
\end{align}
We will only use the ingoing double null-like coordinate system in region $\un$, and the outgoing double null-like coordinate system in region $\deux\cup\trois$ so we use the same notations for the ingoing and outgoing double null-like coordinate vector fields $\partial_u$, $\partial_\ubar$, as there is no danger of confusion. See Section \ref{section:structure} for the definitions of regions $\un,\deux,\trois$.

\textbf{Constant $\wbar$ and $w$ spacelike hypersurfaces.} We will use two families of hypersurfaces defined as the level sets of the scalar functions
$$\wbar:=\ubar-r+r_+,\quad w:=u-r+r_-.$$
The constant $\wbar$ and $w$ hypersurfaces are spacelike, unlike the constant $\ubar, u$ hypersurfaces\footnote{We have 
$$\mathbf{g}_{a,M}(\nabla\ubar,\nabla\ubar)=\mathbf{g}_{a,M}(\nabla u,\nabla u)=\frac{a^2\sin^2\theta}{\Sigma}.$$}. Indeed, see \cite{scalarMZ}, we have
$$\mathbf{g}_{a,M}(\nabla\wbar,\nabla\wbar)=\mathbf{g}_{a,M}(\nabla w,\nabla w)=-\frac{r^2+2Mr+a^2\cos^2\theta}{\Sigma}<-1.$$
Notice that in a region with $\ubar$ large enough, we have $\wbar\sim\ubar.$ Recall that in the Kerr black hole interior, the hypersurfaces $\{t=t_0\}$ are timelike and the hypersurfaces $\{r=r_0\}$ are spacelike for $r_0\in(r_-,r_+)$ and null for $r_0=r_{\pm}$.
\begin{figure}[h!]
    \centering
    \includegraphics[scale=0.5]{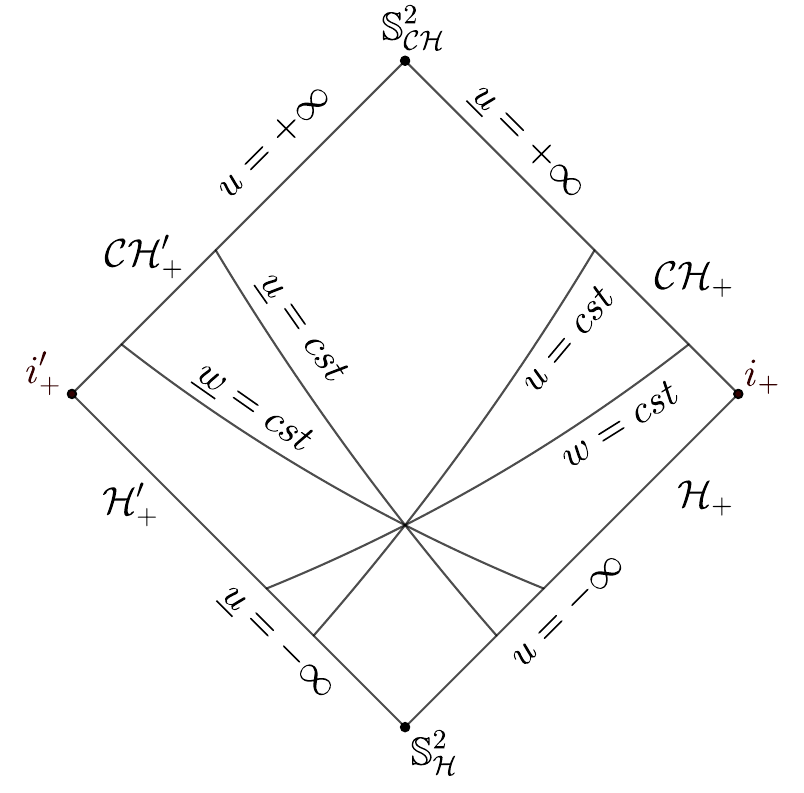}
    \caption{Examples of constant $u,\ubar,w,\wbar$ hypersurfaces and their causal nature.}
    \label{fig:range2}
\end{figure}
\subsection{Spin weighted scalars and mode projection}\label{section:spinweighted}
\subsubsection{Spin-weighted scalars and spin-weighted spherical operators.} In this section, we consider the round sphere $\sph^2$ equipped with its volume element, in coordinates $(\theta,\phi)\in(0,\pi)\times[0,2\pi)$:  
$$\dee\nu:= \sin\theta\dee\theta\dee\phi.$$

Let $s$ be an integer. Note that in this work we will only consider $s=-2$. A spin $s$ scalar is a scalar function that has zero boost weight and proper spin weight, as defined by Geroch, Held and Penrose \cite{spin}. Roughly speaking, a spin $-2$ scalar is a the contraction $\alpha(\overline{\eta},\overline{\eta})$ of a symmetric tensor $\alpha$ with $(\overline{\eta},\overline{\eta})$ where $\eta=\partial_\theta+\frac{i}{\sin\theta}\partial_\phi$. See \cite{sbierski} for a rigorous presentation of spin-weighted scalars spaces in the Kerr interior, and  \cite{millet2} for a precise review of the geometric background of the Teukolsky equation.

We define the topological spheres of Kerr spacetime by:
$$S(u,\ubar):=\{r=r(u,\ubar)\}\cap\{t=t(u,\ubar)\}.$$
The volume element induced on $S(u,\ubar)$ by the metric $\mathbf{g}_{a,M}$ is 
$$\sqrt{(r^2+a^2)^2-a^2\sin^2\theta\Delta}\sin\theta\dee\theta\dee\phi$$
where $\sqrt{(r^2+a^2)^2-a^2\sin^2\theta\Delta}\sim 1$ in the Kerr black hole interior. Thus, although they are not round, we still use the round volume element $\dee\nu$ on the Kerr spheres $S(u,\ubar)$ to define $L^2(S(u,\ubar))$ norms. 
\begin{defi}
    For $\psi$ a spin-weighted scalar in the Kerr black hole interior, we define
$$\|\psi\|_{L^2(S(u,\ubar))}:=\left(\int_{S(u,\ubar)}|\psi^2|\dee\nu\right)^{1/2}.$$
\end{defi}
Notice that in the ingoing double null-like coordinates, as $\phi_+-\phi=r_{mod}$ is constant on $S(u,\ubar)$, the definition of the $L^2(S(u,\ubar))$ norm gives: 
$$\|\psi\|_{L^2(S(u,\ubar))}^2=\int_0^\pi\int_0^{2\pi}|\psi|^2(u,\ubar,\theta,\phi_+)\sin\theta\dee\theta\dee\phi_+.$$
For the same reason, in outgoing double null-like coordinates we also have: 
$$\|\psi\|_{L^2(S(u,\ubar))}^2=\int_0^\pi\int_0^{2\pi}|\psi|^2(u,\ubar,\theta,\phi_-)\sin\theta\dee\theta\dee\phi_-.$$
\begin{defi}
    We recall the definition of the following standard spin-weighted differential operators, called the spherical eth operators:
    \begin{align}
        \drond&:=\partial_\theta+\frac{i}{\cos\theta}\partial_\phi-s\cot\theta,\quad \drond':=\partial_\theta-\frac{i}{\cos\theta}\partial_\phi+s\cot\theta.
    \end{align}
\end{defi}
A classical fact is that the spherical eth operators modify the spin when applied to a spin-weighted scalar. More precisely, $\drond'$ decreases the spin by $1$ while $\drond$ increases the spin by $1$.
\begin{defi}
    We define the spin-weighted Laplacian as
    \begin{align}
        \drond\drond'=\frac{1}{\sin\theta}\partial_\theta(\sin\theta\partial_\theta)+\frac{1}{\sin^2\theta}\partial_\phi^2+\frac{2is\cos\theta}{\sin^2\theta}\partial_\phi-(s^2\cot^2\theta+s).
    \end{align}
\end{defi}
\begin{rem}
    We also have the following identity
        \begin{align}
        \drond'\drond=\frac{1}{\sin\theta}\partial_\theta(\sin\theta\partial_\theta)+\frac{1}{\sin^2\theta}\partial_\phi^2+\frac{2is\cos\theta}{\sin^2\theta}\partial_\phi-(s^2\cot^2\theta-s)=\drond\drond'+2s.
    \end{align}
\end{rem}

\subsubsection{Spin-weighted spherical harmonics.} Let $s$ be a fixed spin. The spin-weighted spherical harmonics are the eigenfunctions of the spin-weighted Laplacian, which is self-adjoint on $L^2(\sph^2)$. They are given by the following family:
$$Y_{m,\ell}^s(\cos\theta)e^{im\phi},\quad\ell\geq |s|, -\ell\leq m\leq \ell$$
of spin $s$ scalars on the sphere $\sph^2$, and form a complete orthonormal basis of the space of spin $s$ scalars on $\sph^2$. When considering the spheres of Kerr interior, as the B-L coordinate $\phi$ is singular at the horizons, we need a slightly modified family. Let $r_0\in(r_-,r_+)$. Then for $(u,\ubar)$ such that $r(u,\ubar)\in [r_0,r_+]$, the family
$$Y_{m,\ell}^s(\cos\theta)e^{im\phi_+},\quad\ell\geq |s|, -\ell\leq m\leq \ell$$
is a complete orthonormal basis of the space of spin $s$ scalars on $S(u,\ubar)$. Similarly, if $r(u,\ubar)\in [r_-,r_0]$, the family
$$Y_{m,\ell}^s(\cos\theta)e^{im\phi_-},\quad\ell\geq |s|, -\ell\leq m\leq \ell$$
form a complete orthonormal basis of the space of spin $s$ scalars on $S(u,\ubar)$. We recall the following standard facts:
\begin{align}
    \drond'\drond(Y_{m,\ell}^s(\cos\theta)e^{im\phi_\pm})&=-(\ell-s)(\ell+s+1)Y_{m,\ell}^s(\cos\theta)e^{im\phi_\pm},\label{eq:dronddrond'spin}\\
    \drond\drond'(Y_{m,\ell}^s(\cos\theta)e^{im\phi_\pm})&=-(\ell+s)(\ell-s+1)Y_{m,\ell}^s(\cos\theta)e^{im\phi_\pm},\label{eq:annulemode}\\
    \drond(Y_{m,\ell}^s(\cos\theta)e^{im\phi_\pm})&=-\sqrt{(\ell-s)(\ell+s+1)}Y_{m,\ell}^{s+1}(\cos\theta)e^{im\phi_\pm},\label{eq:drooo}\\
    \drond'(Y_{m,\ell}^s(\cos\theta)e^{im\phi_\pm})&=\sqrt{(\ell+s)(\ell-s+1)}Y_{m,\ell}^{s-1}(\cos\theta)e^{im\phi_\pm}.
\end{align}
\begin{defi}
    Let $\psi$ be a spin $-2$ scalar. We define
    $$(\psi)^{\pm}_{m,\ell}:=\int_{S(u,\ubar)}\psi\cdot \overline{Y_{m,\ell}^{-2}(\cos\theta)e^{im\phi_\pm}}\dee\nu.$$
\end{defi}
\begin{defi}We denote by $\mathbf{P}_{\ell}$
    the orthogonal projection on the $\ell$ mode for spin $-2$ scalars. It is the linear operator defined by: 
    $$\mathbf{P}_{\ell}(\psi)=\sum_{|m|\leq\ell}(\psi)^{\pm}_{m,\ell}Y_{m,\ell}^{-2}(\cos\theta)e^{im\phi_\pm}.$$
    Note that this definition does not depend on the chosen sign $\pm$. Next, for $\ell'\geq 2$ we define 
$$\mathbf{P}_{\ell\geq\ell'}:=\sum_{\ell\geq\ell'}\mathbf{P}_{\ell}=\mathrm{Id}-\sum_{2\leq\ell<\ell'}\mathbf{P}_{\ell}.$$
     We will also use the notations 
    $$(\psi)_{\ell}=\mathbf{P}_{\ell}(\psi),\quad(\psi)_{\ell\geq\ell'}=\mathbf{P}_{\ell\geq\ell'}(\psi) .$$
\end{defi}
The following result is a standard fact, and is responsible for the mode coupling that appears when projecting the Teukolsky equation onto $\ell$ modes in Kerr spacetimes.
\begin{prop}\label{prop:modecoupling}
    There are constants $\{c_{m,\ell}^{-2}\}_{|m|\leq\ell}$ and $\{b_{m,\ell}^{-2}\}_{|m|\leq\ell}$ such that for any spin ${-2}$ scalar $\psi$ and any $\ell\geq 2$, we have 
\begin{align}
    &(\cos\theta\psi)^\pm_{m,\ell}=\sum_{\ell'=\ell-1}^{\ell+1}b_{m,\ell'}^{-2}\psi^\pm_{m,\ell'},\label{eq:projcostheta}\\
    &(\sin^2\theta\psi)^\pm_{m,\ell}=\sum_{\ell'=\ell-2}^{\ell+2}c_{m,\ell'}^{-2}\psi^\pm_{m,\ell'}.
\end{align}

\end{prop}
\begin{proof}
    See for example \cite[Prop. 2.13]{MZ23}.
\end{proof}
\begin{defi}
    We define the following spin-weighted operator:
\begin{align}\label{eq:defU}
    \mcu:=\frac{1}{\sin\theta}\Phi+a\sin\theta T+is\cot\theta\frac{\Sigma}{r^2+a^2}.
\end{align}
\end{defi}
This operator $\mcu$ is useful to write a Poincaré-type inequality in Kerr, see \cite{G24}. For a spin-weighted scalar $\psi$ we have $\mcu\psi\in L^\infty(S(u,\ubar))$ (see for example \cite[(2.34)]{sbierski}). The following integration by parts lemma is \cite[Lem. 2.10]{G24}:
\begin{lem}\label{lem:U}
    Let $\psi,\varphi$ be spin $s$ scalars. Then 
    $$\intS\overline{\varphi}\:\mcu\psi\dee\nu=T\left(\intS a\sin\theta\overline{\varphi}\psi\dee\nu\right)-\intS\overline{\mcu\varphi}\psi\dee\nu.$$
\end{lem}
We can express the Killing vector fields $T$ and $\Phi$ with $e_3$, $e_4$ and $\mcu$ in the following way:
\begin{align}\label{eq:expreT}
        T&=\frac{r^2+a^2}{\Sigma}(\mu e_3+e_4)-\frac{a\sin\theta}{\Sigma}\mcu+\frac{ias\cos\theta}{r^2+a^2},\\        \Phi&=\frac{\sin\theta(r^2+a^2)}{\Sigma}(\mcu-a\sin\theta e_4-a\sin\theta\mu e_3)-is\cos\theta.\label{eq:exprephi}
    \end{align}
    In particular, $T=O(\mu)e_3+O(1)e_4+O(\sin\theta)\mcu+O(1)$, and $\Phi=O(\mu)e_3+O(1)e_4+O(\sin\theta)\mcu+O(1)$.
\subsection{Analytic preliminaries}
\begin{defi}
We define the spin-weighted Carter operator 
\begin{align}\label{eq:carter}
    \mcq_s:=a^2\sin^2\theta T^2-2ias\cos\theta T +\drond'\drond.
\end{align}
\end{defi}
This operator is crucial in our analysis because it commutes with the Teukolsky operator (see \eqref{eq:commutateurs}), and thus will help control angular derivatives of solutions of the Teukolsky equation. We now state a Sobolev inequality on $S(u,\ubar)$ for spin-weighted scalars. It will be used throughout the paper to deduce pointwise decay from $L^2(S(u,\ubar))$ decay of Carter derivatives and $T$ derivatives.
\begin{prop}\label{prop:sobolev}
    Let $\psi$ be a spin $\pm 2$ scalar. For any $(u,\ubar)$ we have
\begin{align}\label{eq:sobolev}
    \|\psi\|^2_{L^\infty(S(u,\ubar))}{\lesssim}\int_{S(u,\ubar)}\left(|T^{\leq 2}\psi|^2+|\mcq_s\psi|^2\right)\:\dee\nu.
\end{align}
\end{prop}
\begin{proof}
     The proof is a simple re-writing of the standard Sobolev inequality on $\mathbb{S}^2$, see \cite[Prop. 2.16]{G24}.
\end{proof}
\begin{prop}\label{prop:salva}
    Let $\psi$ be a spin $-2$ scalar. We have for $\ell\geq 2$, $|m|\leq\ell$,
    $$|(\psi)_{m,\ell}^\pm|\lesssim\ell^{-2}\|\mcq_{-2}^{\leq 1}T^{\leq 2}\psi\|_{L^\infty(S(u,\ubar))}.$$
\end{prop}
\begin{proof}
    We use the following identity, which is deduced from \eqref{eq:dronddrond'spin}, and the fact that $\drond'\drond$ is self-adjoint on the sphere:
    $$(\psi)_{m,\ell}^\pm=\frac{1}{(\ell+2)(\ell-1)}(\drond'\drond\psi)_{m,\ell}^\pm,$$
    and all is left to do is to replace the spin-weighted Laplacian $\drond'\drond$ with the Carter operator and $T$ derivatives using \eqref{eq:carter}.
\end{proof}
Next, we state a Poincaré-type inequality which will be used to absorb the $0$ order terms in the energy estimate for Teukolsky.
\begin{defi}\label{defi:energy}
    Let $\psi$ be a spin weighted scalar. We define respectively the energy density and degenerate energy density of $\psi$ as 
    $$\ener[\psi]:=|e_3\psi|^2+|e_4\psi|^2+|\mcu\psi|^2+|\partial_\theta\psi|^2,$$
    $$\ener_{deg}[\psi]:=\mu^2|e_3\psi|^2+|e_4\psi|^2+|\mcu\psi|^2+|\partial_\theta\psi|^2.$$
\end{defi}
\begin{prop}
Let $\psi$ be a spin $-2$ scalar. For any $(u,\ubar)$ we have the Poincaré inequality
\begin{align}\label{eq:poincare}
    \|\psi\|^2_{L^2(S(u,\ubar))}{\lesssim}\int_{S(u,\ubar)}\mathbf{e}_{deg}[\psi]\dee\nu.
\end{align}
\end{prop}
\begin{proof}
    The proof is a simple re-writing of the standard Poincaré inequality on $\mathbb{S}^2$, see \cite[Prop. 2.14]{G24}.
\end{proof}
The following result is used to deduce decay of the energy from an energy estimate.
\begin{lem}\label{lem:decay}
    Let $p>1$ and $0\leq\gamma<1$ and let $f:[1,+\infty) \rightarrow[0,+\infty)$ be a continuous function. Assume that there are constants $C_0>0, C_1>0, C_2 \geq 0$, $C_3 \geq 0$ such that for $1 \leq x_1<x_2$,
$$
f(x_2)+C_1 \int_{x_1}^{x_2}  x^{-\gamma}f(x) \mathrm{d} x \leq C_0 f(x_1)+C_2 \int_{x_1}^{x_2} x^{-p} \mathrm{~d} x+C_3 x_1^{-p} .
$$

Then for any $x_1 \geq 1$,
$$
f(x_1) \leq C x_1^{-p+\gamma}
$$
where $C$ is a constant that depends only on $f(1), C_0, C_1, C_2, C_3,p$, and $\gamma$.
\end{lem}
\begin{proof}
    See \cite[Lemma 3.4]{scalarMZ}.
\end{proof}

\subsection{The Teukolsky equations}\label{section:system}
\subsubsection{Definition and expressions of the Teukolsky operators}
We define the Teukolsky operator $\teukhat_s$ as the operator on the LHS of the Teukolsky equation \eqref{eq:premteuk}, namely
\begin{align*}
     \teukhat_s:=-&\left[\frac{\left(r^2+a^2\right)^2}{\Delta}-a^2 \sin ^2 \theta\right] \partial_t^2 -\frac{4 M a r}{\Delta}\partial_t\partial_\phi  -\left[\frac{a^2}{\Delta}-\frac{1}{\sin ^2 \theta}\right] \partial_\phi^2 \nonumber \\
    & +\Delta^{-s} \partial_r\left(\Delta^{s+1} \partial_r \right)+\frac{1}{\sin \theta} \partial_\theta\left(\sin \theta \partial_\theta \right)+2 s\left[\frac{a(r-M)}{\Delta}+\frac{i \cos \theta}{\sin ^2 \theta}\right] \partial_\phi  \nonumber\\
    & +2 s\left[\frac{M\left(r^2-a^2\right)}{\Delta}-r-i a \cos \theta\right]\partial_t -\left[\frac{s^2 \cos ^2 \theta}{\sin ^2 \theta}-s\right].
\end{align*}
We also define the rescaled Teukolsky operator $\teuk_s$ by:
\begin{align}
    \teuk_s:=\Delta^{s}\teukhat_s(\Delta^{-s}\cdot).\label{eq:rescaledequation}
\end{align}
In the rest of the paper, $\Psihatm$ satisfies the Teukolsky equation $\teukhat_{-2}\Psihatm=0$. Thus, denoting $$\Psim:=\Delta^{-2}\Psihatm,$$
the Teukolsky equation rewrites $\teuk_{-2}\Psim=0$, where the expression of $\teuk_s$ in B-L coordinates is:
\begin{align*}
\teuk_s=  -&\left[\frac{\left(r^2+a^2\right)^2}{\Delta}-a^2 \sin ^2 \theta\right] T^2 -\frac{4 M a r}{\Delta} T \Phi -\left[\frac{a^2}{\Delta}-\frac{1}{\sin ^2 \theta}\right] \Phi^2  \\
& +\Delta^{-s} \partial_r\left(\Delta^{s+1} \partial_r \right)+\frac{1}{\sin \theta} \partial_\theta\left(\sin \theta \partial_\theta \right)+2 s\left[\frac{a(r-M)}{\Delta}+\frac{i \cos \theta}{\sin ^2 \theta}\right] \Phi  \\
& +2 s\left[\frac{M\left(r^2-a^2\right)}{\Delta}-r-i a \cos \theta\right] T -\left[\frac{s^2 \cos ^2 \theta}{\sin ^2 \theta}+s\right] -4 s(r-M) \partial_r.
\end{align*}
The Teukolsky operators can also be written as, see for example \cite[(3.3)]{MZ23},
\begin{align}
    &\teuk_s=-\frac{(r^2+a^2)^2-a^2\Delta\sin^2\theta}{\Delta}T^2+\partial_r(\Delta\partial_r)-\frac{4aMr}{\Delta}T\Phi-\frac{a^2}{\Delta}\Phi^2\nonumber\\
    &\quad\quad\quad+\drond\drond'-2ias\cos\theta T-2s[(r-M)\drin+2rT],\label{eq:opteuk}\\
    &\teukhat_s=-\frac{(r^2+a^2)^2-a^2\Delta\sin^2\theta}{\Delta}T^2+\partial_r(\Delta\partial_r)-\frac{4aMr}{\Delta}T\Phi-\frac{a^2}{\Delta}\Phi^2\nonumber\\
    &\quad\quad\quad+\drond'\drond-2ias\cos\theta T+2s[(r-M)\drout-2rT].\label{eq:opteukhat}
\end{align}
These expressions directly prove that $T$, $\Phi$, and the Carter operator $\mcq_s$ commute with the Teukolsky equation, i.e.
\begin{align}\label{eq:commutateurs}
    [T,\teukhat_s]=[T,\teuk_s]=[\mcq_s,\teukhat_s]=[\mcq_s,\teuk_s]=[\Phi,\teukhat_s]=[\Phi,\teuk_s]=0.
\end{align}
The following expressions of the Teukolsky operators will be convenient in the energy estimates:
\begin{prop}\label{prop:}
We have:
    \begin{align}\label{eq:teuke3e4}
    \teuk_s=&-4(r^2+a^2)e_3e_4+ \mcu^2+\dfrac{1}{\sin\theta}\partial_\theta(\sin\theta\partial_\theta)-4 ias\cos\theta T+2is\dfrac{a^2\sin\theta\cos\theta}{(r^2+a^2)}\mcu\nonumber\\
    &+2 r (e_4-\mu e_3)+4 s[(r-M)e_3- rT]+\dfrac{2ar }{r^2+a^2}\Phi- s-s^2\dfrac{a^4\sin^2\theta\cos^2\theta}{(r^2+a^2)^2},
\end{align}
    \begin{align}\label{eq:teukhate3e4}
    \teukhat_s=&-4(r^2+a^2)e_3e_4+ \mcu^2+\frac{1}{\sin\theta}\partial_\theta(\sin\theta\partial_\theta)-4 ias\cos\theta T+2is\frac{a^2\sin\theta\cos\theta}{(r^2+a^2)}\mcu\nonumber\\
    &+2 r (e_4-\mu e_3)+4 s[(r-M)\mu^{-1}e_4- rT]+\frac{2ar }{r^2+a^2}\Phi+ s-s^2\frac{a^4\sin^2\theta\cos^2\theta}{(r^2+a^2)^2}.
\end{align}
\end{prop}
\begin{proof}
See \cite[Prop. 2.17]{G24}.
\end{proof}
The spin $+2$ equivalent of the following identity was crucial in \cite{G24}.
\begin{prop}\label{prop:voila}
    For any spin $-2$ scalar $\Psihat$, we have
    \begin{align*}
\ethat((r^2+a^2)\Delta^{-2}e_4\Psihat)=\frac{1}{4}\mu\Delta^{-2}\left[\teukhat_{-2}-\carterm-2aT\Phi-6rT\right]\Psihat.
\end{align*}
\end{prop}
\begin{proof}
    We have 
    \begin{align*}
        \ethat((r^2+a^2)\Delta^{-2}e_4\psihat)&=r\mu\Delta^{-2} e_4\psihat-2\Delta^{-2}(r-M)e_4\psihat+(r^2+a^2)\Delta^{-2} \ethat e_4\psihat\\
        &=\frac{1}{4}\mu\Delta^{-2}(-4(r^2+a^2)e_3e_4-8(r-M)\mu^{-1}e_4+4re_4)\psihat.
    \end{align*}
Using $e_4-\mu e_3=\mu\partial_r=2e_4-T-a/(r^2+a^2)\Phi$ we get
    \begin{align}
        \ethat((r^2+a^2)&\Delta^{-2} e_4\psihat)=\label{eq:tocomb}\\
        &\frac{1}{4}\mu\Delta^{-2}\Big(-4(r^2+a^2)e_3e_4-8(r-M)\mu^{-1}e_4+2r(e_4-\mu e_3)+2r T+\frac{2ra}{r^2+a^2}\Phi\Big)\psihat.\nn
    \end{align}
Next, using \eqref{eq:teukhate3e4} we get
    \begin{align*}
    \teukhat_{-2}=&-4(r^2+a^2)e_3e_4+ \widetilde{U}^2+\frac{1}{\sin\theta}\partial_\theta(\sin\theta\partial_\theta)+8ia\cos\theta T\\
    &+2 r (e_4-\mu e_3)-8[(r-M)\mu^{-1}e_4- rT]+\frac{2ar }{r^2+a^2}\Phi -2,
\end{align*}
where 
$$\widetilde{U}^2:=\frac{1}{\sin^2\theta}\Phi^2+a^2\sin^2 T^2+2a T\Phi-4ai\cos\theta T-\frac{4i\cos\theta}{\sin^2\theta}\Phi-4\cot^2\theta.$$
We infer, using the definition of the Carter operator \eqref{eq:carter},
    \begin{align*}
    \teukhat_{-2}=&-4(r^2+a^2)e_3e_4+2 r (e_4-\mu e_3)-8(r-M)\mu^{-1}e_4+\frac{2ar }{r^2+a^2}\Phi\\
    &\quad\quad+\drond'\drond+a^2\sin^2 T^2+2aT\Phi+4ai\cos\theta T+8rT\\
    &=-4(r^2+a^2)e_3e_4+2 r (e_4-\mu e_3)-8(r-M)\mu^{-1}e_4+\frac{2ar }{r^2+a^2}\Phi\\
    &\quad\quad+\mcq_{-2}+2aT\Phi+8rT.
\end{align*}
Combining this identity with \eqref{eq:tocomb} concludes the proof.
\end{proof}
\subsubsection{Decomposition $(\partial_t,\partial_r)$ of the Teukolsky operators}\label{section:decdtdr} The following $(\partial_t,\partial_r)$ decompositions of the Teukolsky operators will be crucial in this paper. We will use them to reduce the Teukolsky equation to an ODE along hypersurfaces of constant $u$ and $\ubar$.
\begin{prop}
    We have the decomposition of the Teukolsky operator
    \begin{align}
        \label{eq:decteuk}\teuk_s=\teuk_s^{[\partial_r]}+\teuk_s^{[\partial_t]},
    \end{align}
    where, recalling the definitions of the coordinate vector fields $\drin$, $\drout$ in Section \ref{section:coordsss},
\begin{align}\label{eq:teukdr}
    \teuk_s^{[\partial_r]}&:=\Delta\drin^2+2a\Phi\drin-2(r-M)(s-1)\drin+\drond\drond',\\
    \teuk_s^{[\partial_t]}&:=a^2\sin^2\theta T^2+2(r^2+a^2)T\drin+2aT\Phi+(2r(1-2s)-2ias\cos\theta)T.\label{eq:teukdt}
\end{align}
Similarly we have 
\begin{align}\label{eq:decteukhat}
    \teukhat_s=\teukhat_s^{[\partial_r]}+\teukhat_s^{[\partial_t]},
\end{align}
where
\begin{align}\label{eq:teukhatdr}
    \teukhat_s^{[\partial_r]}&:=\Delta\drout^2-2a\Phi\drout+2(r-M)(s+1)\drout+\drond'\drond,\\
    \teukhat_s^{[\partial_t]}&:=a^2\sin^2\theta T^2-2(r^2+a^2)T\drout+2aT\Phi-(2r(1+2s)+2ias\cos\theta)T.\label{eq:teukhatdt}
\end{align}
\end{prop}
\begin{proof}
    
Using \eqref{eq:opteuk} and
$$\partial_r=\drin+\frac{a}{\Delta}\Phi+\frac{r^2+a^2}{\Delta}T,$$
we get 
\begin{align*}
    \teuk_s&=-\frac{(r^2+a^2)^2-a^2\Delta\sin^2\theta}{\Delta}T^2+\left(\drin+\frac{a}{\Delta}\Phi+\frac{r^2+a^2}{\Delta}T\right)(\Delta\drin+a\Phi+(r^2+a^2)T)\\
    &\quad\quad\quad-\frac{4aMr}{\Delta}T\Phi-\frac{a^2}{\Delta}\Phi^2+\drond\drond'-2ias\cos\theta T-2s[(r-M)\drin+2rT]\\
    &=\Delta \drin^2+a^2\sin^2\theta T^2+2(r-M)\drin+2a\Phi\drin+2(r^2+a^2)T\drin+\frac{a^2}{\Delta}\Phi^2+\frac{2a(r^2+a^2)}{\Delta}T\Phi\\
    &\quad\quad\quad+2r T-\frac{4aMr}{\Delta}T\Phi-\frac{a^2}{\Delta}\Phi^2+\drond\drond'-2ias\cos\theta T-2s[(r-M)\drin+2rT]\\
    &=\teuk_s^{[\partial_r]}+\teuk_s^{[\partial_t]},
\end{align*}
as stated in \eqref{eq:decteuk}. The proof of \eqref{eq:decteukhat} is similar.
\end{proof}

\begin{prop}\label{prop:commdrout}
    We have the commutator identities
    $$[\teuk_s,\drin]=[4(r-M)e_3-4rT+2(s-1)]\drin-2(1-2s)T,$$
    $$[\teukhat_s,\drout]=-[4(r-M)\mu^{-1}e_4-4rT+2(s+1)]\drout+2(1+2s)T.$$
\end{prop}
\begin{proof}
    We simply use the above expressions \eqref{eq:teukdt}, \eqref{eq:teukdr}, \eqref{eq:teukhatdr}, \eqref{eq:teukhatdt}, as well as $\drin=-2e_3$ and $\drout=2\mu^{-1}e_4$.
\end{proof}
\begin{lem}\label{lem:commproj}
    We have, for any $\ell\geq 2$:$$[\drin,\mathbf{P}_{\ell}]=[\drout,\mathbf{P}_{\ell}]=[\teuk_{-2}^{[\partial_r]},\mathbf{P}_{\ell}]=[\teukhat_{-2}^{[\partial_r]},\mathbf{P}_{\ell}]=0.$$
\end{lem}
\begin{proof}
Using the two expressions
$$\mathbf{P}_{\ell}(\psi)=\sum_{|m|\leq\ell}(\psi)^{+}_{m,\ell}Y_{m,\ell}^{-2}(\cos\theta)e^{im\phi_+}=\sum_{|m|\leq\ell}(\psi)^-_{m,\ell}Y_{m,\ell}^{-2}(\cos\theta)e^{im\phi_-},$$
as well as $\drin(Y_{m,\ell}^{-2}(\cos\theta)e^{im\phi_+})=\drout(Y_{m,\ell}^{-2}(\cos\theta)e^{im\phi_-})=0$, we deduce respectively $[\drin,\mathbf{P}_{\ell}]=0$ and $[\drout,\mathbf{P}_{\ell}]=0$. The two remaining identities come from \eqref{eq:teukdr} and \eqref{eq:teukhatdr}, where we also use the fact that $\drond\drond'$ and $\drond'\drond$ are self-adjoint on $\mathbb{S}^2$.
\end{proof}

\subsection{Subregions of the Kerr black hole interior}

We will divide the Kerr black hole interior $\mathbb{R}_\ubar\times\mathbb{R}_u\times\mathbb{S}^2$ into different regions. We fix $\rb>r_-$ close to $r_-$ and $\gamma>0$ small. The final values of $\gamma$ and $\rb$ will be chosen in the energy estimates, respectively in Section \ref{section:energymethod} and in Appendix \ref{appendix:bulkII}. We define the hypersurface\footnote{The hypersurface $\Gamma$ is spacelike for $r$ is sufficiently close to $r_-$.}
$$\Gamma:=\{2r^*=\ubar^\gamma\}\cap\{\ubar\geq 1\},$$
and the following subregions of the Kerr black hole interior: 
\begin{align*}
    &\mathbf{I}:=\{ \rb\leq r\leq r_+\}\cap\{\ubar\geq 1\},\\
    &\mathbf{II}:=\{r_-\leq r\leq \rb\}\cap\{2r^*\leq\ubar^\gamma\}\cap\{w\leq w_{\rb,\gamma}\},\\
    &\mathbf{III}:=\{r_-\leq r\leq \rb\}\cap\{2r^*\geq\ubar^\gamma\}\cap\{w\leq w_{\rb,\gamma}\},
\end{align*}
where $w_{\rb,\gamma}:=2\rb^*-(2\rb^*)^{1/\gamma}-\rb+r_-$ is such that $\{w=w_{\rb,\gamma}\}$ intersects $\{r=\rb\}$ at $\Gamma\cap\{r=\rb\}$ (recall the definition of $w$ in Section \ref{section:coordsss}). See Figure \ref{fig:regions-2} for an illustration of regions $\un$, $\deux$, $\trois$.
\begin{figure}[h!]
    \centering
    \includegraphics[scale=0.4]{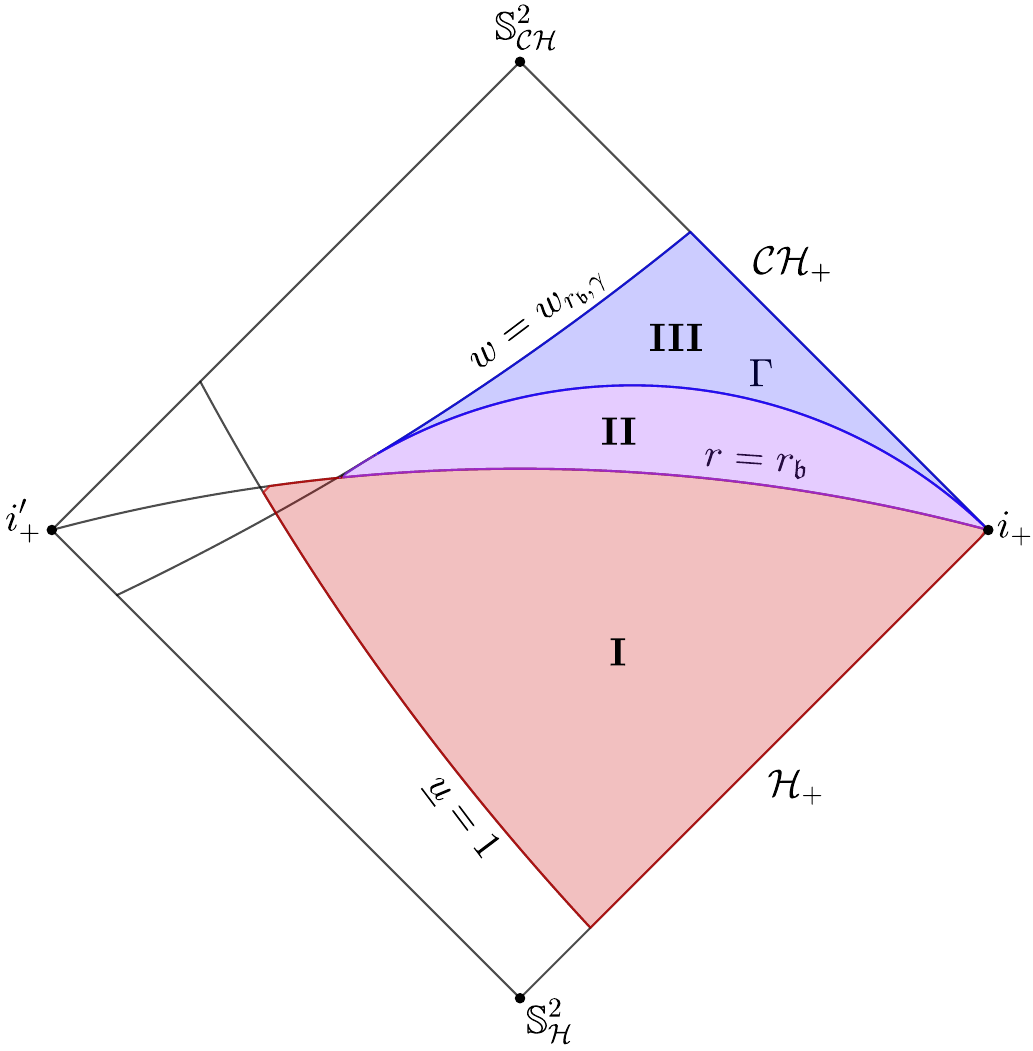}
    \caption{Subregions $\un,\deux$ and $\trois$ of the Kerr black hole interior.}
    \label{fig:regions-2}
\end{figure}

Region $\un$ is the redshift region close to $\eh$, region $\trois$ is the blueshift region close to $\ch$ and region $\deux$ is an intermediate region where the blueshift effect is already present. Notice that with our choice of $\rb,\gamma$, we get that $w_{\rb,\gamma}\ll -1$ is negative and large in absolute value. This implies $u\ll -1$ in $\deux\cup\trois$.
\begin{lem}\label{lem:usimubarII}
    We have $|u|\sim\ubar$ in $\deux$ and $\ubar\geq|u|$ in $\trois$.
\end{lem}
\begin{proof}
The first point comes from the definition of $\deux$: $u+\ubar=2r^*\leq\ubar^{\gamma}$ which implies $\ubar\leq |u|+\ubar^\gamma$ thus $\ubar\lesssim|u|$ for $\ubar\geq 1$ large enough. Moreover in $\deux$ we also have $2r^*\geq 2\rb^*$ which implies $|u|=-u\leq\ubar-2\rb^*\leq\ubar$, where we use $r^*\geq 0$ for $r$ close to $r_-$. The second fact comes from $|u|=-u=\ubar-2r^*\leq \ubar$ in $\trois$.
\end{proof}
\section{Statement of the main theorem}\label{section:mainthm}

\subsection{Assumptions on the event horizon}\label{section:assumptions}
Our main result will apply to a spin $-2$ scalar solution of the spin $-2$ Teukolsky equation, which is initially smooth and compactly supported on a spacelike hypersurface as in \cite{G24, MZ23}. In practice, we will only rely on Price's law results on the event horizon proven in \cite{MZ23} which hold under such assumptions and which we recall below.

We define the set of derivatives which are tangential to the event horizon: 
$$\nablabar:=\{\Phi,T,\drond,\drond'\}.$$
For a multi-index $k=(k_1,k_2,k_3,k_4)\in\mathbb{N}^4$, we define $|k|:=k_1+k_2+k_3+k_4$, and 
$$\nablabar^k=\Phi^{k_1}T^{k_2}\drond^{k_3}(\drond')^{k_4}.$$
Let $N_j$, $N_k$ be sufficiently large integers that we will choose later, see Remark \ref{rem:postthm}. We will consider a spin $-2$ scalar $\Psim$ solution of the spin $-2$ Teukolsky equation
$$\teukm\Psim=0$$
such that there are constants $Q_{m,2}$ for $|m|\leq 2$, $\delta>0$ small, such that for $j\leq N_j$ and $|k|\leq N_k$
\begin{align}\label{eq:hyppl}
    \left|e_3^{\leq 1}T^j\nablabar^k\left(\Psim-\ansatzm\right)\right|\lesssim\ubar^{-7-j-\delta},\quad\text{on}\:\eh\cap\{\ubar\geq 1\}.
\end{align}
\begin{rem}
    Price's law type assumption \eqref{eq:hyppl} was proven in \cite{MZ23} by Ma and Zhang\footnote{The decay of the tangential derivatives $\nablabar^k$ is not explicitely stated in \cite{MZ23}, but can be easily deduced from the fact that $\Phi$, $T$ and the Carter operator commute with the Teukolsky equation, and by applying the main theorem of \cite{MZ23}. The statement for the $e_3$ derivative comes from \cite[(4.101)]{MZ23}, recalling $e_3(\ubar)=e_3(\theta)=e_3(\phi_+)=0$.} for solutions of the spin $-2$ Teukolsky equation which are initially smooth and compactly supported on a spacelike hypersurface $\Sigma_0$ defined in \cite[Section 2.1]{MZ23}, where $Q_{m,2}=2^7\mathbb{Q}_{m,2}/5$ with $\mathbb{Q}_{m,2}$ defined in \cite[Lem. 5.7]{MZ23}. Their proof holds for $|a|\ll M$, and for $|a|<M$ conditionally on an energy-Morawetz estimate for the Teukolsky equation in subextremal Kerr. This estimate was then proven in the whole subextremal range $|a|<M$ in \cite{TDCSR2,TDCSR1}.
\end{rem}
Using the results in \cite{G24} for the spin $-2$ equation in the redshift region $\un$, in the present article we will be able to restrict most of the analysis to regions $\deux$ and $\trois$. The first example of this is the following result, which follows from \cite[Theo. 4.11]{G24}:
\begin{prop}\label{prop:lapremiere}
    Assume that $\Psim$ is a solution of the spin $-2$ Teukolsky equation in the Kerr black hole interior for $0<|a|<M$ that satisfies assumption \eqref{eq:hyppl}. Then, 
    $$\Psim=\ansatzm+\errm,$$
    where for $j\leq N_j-2$ and $2k_1+k_2\leq N_k-3$, we have in $\un$,
    \begin{align}
        |\drin^{\leq 1}T^j\mcq_{-2}^{k_1}\Phi^{k_2}\errm|\lesssim\ubar^{-7-j-\delta},\label{eq:borneI}
    \end{align}
    and, recalling the definition of the energy density in Definition \ref{defi:energy}, the energy bound on $\{r=\rb\}$,
\begin{align}
    \iint_{\left\{\underline{w}_1 \leq \underline{w} \leq w_2\right\} \cap\left\{r=r_b\right\}} \mathbf{e}\left[T^j \mathcal{Q}_{-2}^{k_1} \Phi^{k_2} \drin^{\leq 1} \errm\right] \mathrm{d} \nu \mathrm{d} \underline{u} \lesssim \underline{w}_1^{-2 (7+\delta)-2 j}+\int_{\underline{w}_1}^{\underline{w}_2} \underline{w}^{-2 (7+\delta)-2 j} \mathrm{~d} \underline{w} .\label{eq:enerI}
\end{align}
\end{prop}
\begin{proof}
    We prove that $\errm$ satisfies the assumptions of \cite[Theo. 4.11]{G24} with $\beta=7+\delta$. Assumption \cite[(4.18)]{G24} is \eqref{eq:hyppl}. Now we check \cite[(4.19)]{G24}. Using the decomposition \eqref{eq:decteuk} of the Teukolsky operator, we have the explicit computation
    \begin{align}
        \left|\drin^{\leq 1}T^j\mcq_{-2}^{k_1}\Phi^{k_2}\teuk_{-2}\errm\right|&=\left|\drin^{\leq 1}T^j\mcq_{-2}^{k_1}\Phi^{k_2}\teuk_{-2}\left(\ansatzm\right)\right|\nn\\
        &=\left|\drin^{\leq 1}T^j\mcq_{-2}^{k_1}\Phi^{k_2}\teuk_{-2}^{[\partial_t]}\left(\ansatzm\right)\right|\nn\\
        &\lesssim\ubar^{-8-j},\label{eq:teukerr}
    \end{align}
    where we used \eqref{eq:teukdt} as well as $\drin \ubar=\drin \phi_+=\drin\theta=\drond\drond'Y_{m,2}^{-2}(\cos\theta)=0$ by \eqref{eq:annulemode} for the before-to-last step. We conclude the proof by applying \cite[Theo. 4.11]{G24}.
\end{proof}
\subsection{Precise version of the main theorem}
The following theorem is the main result of this paper.
\begin{thm}[Main theorem, precise version]\label{thm:main}
    Assume that $\Psim$ is a smooth spin $-2$ scalar which satisfies the spin $-2$ Teukolsky equation $\teukm\Psim=0$ in the Kerr black hole interior with $0<|a|<M$ and such that \eqref{eq:hyppl} holds on the event horizon $\mch_+$. Then denoting $\Psihatm=\Delta^2\Psim$, we have in region $\deux\cup\trois=\{r_-\leq r\leq \rb\}\cap\{w\leq w_{\rb,\gamma}\}$ the following asymptotic behavior:
    \begin{align}
        \Psihatm(u,r,\theta,\phi_-)=\frac{1}{u^8}\sum_{|m|\leq 2}C_mY_{m,2}^{-2}(\cos\theta)e^{im\phi_-}+O(|u|^{-8-\delta/2})+O(r-r_-),\label{eq:mainasympt}
    \end{align}
    where, recalling the definition of the constants $b_{m,\ell}^{-2}$ in Proposition \ref{prop:modecoupling},
    \begin{align}\label{eq:cm}
        C_m=\frac{7Q_{m,2}}{2(r_--M)+iam}\int_{r_-}^{r_+}\Delta^2(r')e^{2imr_{mod}(r')}(5r'+iam+2iab_{m,2}^{-2})\dee r'.
    \end{align}
\end{thm}
\begin{rem}\label{rem:postthm}We remark the following:
\begin{enumerate}
    \item As explained in Section \ref{section:assumptions}, Theorem \ref{thm:main} can be applied to a solution $\Psim$ which is initially smooth and compactly supported on a spacelike hypersurface $\Sigma_0$ as in \cite{MZ23}. In particular, we do not need to assume higher order decay of the higher modes $(\psim)_{\ell\geq 3}$ on $\mch_+$.
    \item The constants $C_m$ are generically non-zero. Indeed, for generic initial data, $Q_{m,2}\neq 0$ (see \cite[Lem. 5.7]{MZ23}), and except for isolated values of black hole parameters $(a,M)$,  
    $$\int_{r_-}^{r_+}\Delta^2(r')e^{2imr_{mod}(r')}(5r'+iam+2iab_{m,2}^{-2})\dee r'\neq 0.$$
    This can be seen most easily for $m=0$, as the real part of the above integral is positive.
    \item The asymptotic behavior \eqref{eq:mainasympt} implies that for a large enough constant $C>0$, on $\ch\cap\{u\leq-C\}$ we have generically
    $$\Psihatm\sim\frac{1}{u^8}\sum_{|m|\leq 2}C_mY_{m,2}^{-2}(\cos\theta)e^{im\phi_-}\neq 0.$$
    Together with the discussion in Section \ref{section:curvatureinstab}, this result suggests that gravitational perturbations of a Kerr black hole build up to form a coordinate-independent curvature singularity at the perturbed Cauchy horizon.
   
    \item Inspecting the proof, we find that \eqref{eq:mainasympt} holds choosing $N_k\geq 12$ and $N_j\geq 10$, see Remark \ref{rem:maxderlost}. We did not try to optimize the loss of derivatives. Also assuming that \eqref{eq:hyppl} holds for a higher number of $T$ and $\nablabar$ derivatives implies that the asymptotic behavior \eqref{eq:mainasympt} holds for a corresponding higher number of $T$, $\carterm$ and $\Phi$ derivatives.    
    \item The asymptotic behavior \eqref{eq:mainasympt} confirms the heuristic prediction by Ori in \cite{ori}.
\end{enumerate}
\end{rem}
\subsection{Proof of the main theorem}\label{section:proof}
Theorem \ref{thm:main} is a consequence of the three following results, for a smooth spin $-2$ scalar $\Psim$ such that $\teukm\Psim=0$ in the Kerr black hole interior with $0<|a|<M$, with $\Psihatm=\Delta^2\Psim$.
\begin{thm}\label{thm:precise2l}
    Assume that $\Psim$ satisfies \eqref{eq:hyppl}. Then we have in $\{r_-\leq r\leq \rb\}\cap\{w\leq w_{\rb,\gamma}\}$:
    $$(\Psihatm)_{\ell=2}=\frac{1}{u^8}\sum_{|m|\leq 2}C_mY_{m,2}^{-2}(\cos\theta)e^{im\phi_-}+O(|u|^{-8-\delta/2})+O(r-r_-),$$
    where the constants $C_m$ are defined by \eqref{eq:cm}.
\end{thm}
\begin{thm}\label{thm:decayl3}
    Assume that $\Psim$ satisfies \eqref{eq:hyppl}. Then we have in region $\{r_-\leq r\leq \rb\}\cap\{w\leq w_{\rb,\gamma}\}$, 
    $$|(\Psihatm)_{\ell=3}|\lesssim |u|^{-8-\delta/2}+r-r_-.$$
\end{thm}
\begin{thm}\label{thm:decayhigher}
    Assume that $\Psim$ satisfies \eqref{eq:hyppl}. Then we have in $\{r_-\leq r\leq \rb\}\cap\{w\leq w_{\rb,\gamma}\}$,
    $$|(\Psihatm)_{\ell\geq 4}|\lesssim|u|^{-8-\delta/2}+r-r_-.$$
\end{thm}
We are now ready to prove Theorem \ref{thm:main}.
\begin{proof}[Proof of Theorem \ref{thm:main}]
We decompose $\Psihatm$ as follows:
$$\Psihatm=(\Psihatm)_{\ell=2}+(\Psihatm)_{\ell=3}+(\Psihatm)_{\ell\geq 4}.$$
The proof of Theorem \ref{thm:main} then proceeds by combining Theorems \ref{thm:precise2l}, \ref{thm:decayl3}, and \ref{thm:decayhigher}.
\end{proof}
The rest of the paper is devoted to the proofs of Theorems \ref{thm:precise2l}, \ref{thm:decayl3}, and \ref{thm:decayhigher}. In Section \ref{section:upperbounds}, we prove energy estimates and pointwise bounds for the spin $-2$ Teukolsky equation in the whole Kerr black hole interior. These bounds are then used in Section \ref{section:ccun} and \ref{section:ccdeux} where we prove Theorem \ref{thm:precise2l}, in Sections  \ref{section:l=3I} and \ref{section:l=3II} where we prove Theorem \ref{thm:decayl3}, and finally in Sections \ref{section:higherI} and \ref{section:higherII} where we prove Theorem \ref{thm:decayhigher}.

\section{Upper bound for $\errhatm$ near $\ch$}
\label{section:upperbounds}
 From now on, $\Psim$ satisfies $\teukm\Psim=0$, and we denote:
$$\errhatm=\Delta^2\errm=\Psihatm-\anshatm.$$
\subsection{Energy method for $\teukhat_{-2}$ in region $\deux$}\label{section:energymethod}

We begin this section by introducing a modified Teukolsky operator $\teukhat_{-2}^{(c,V)}$. Similarly as in \cite{G24}, this definition will be used to unify the treatment of the PDEs satisfied by $\errhatm$ and $\drout\errhatm$ after commuting the Teukolsky equation with $\drout$.

\begin{defi}
    Let $V,c$ be real numbers. We define the following spin-weighted operator
\begin{align}
    \teukhat_{s}^{(c,V)}:=&-4(r^2+a^2)e_3e_4+ \mcu^2+\frac{1}{\sin\theta}\partial_\theta(\sin\theta\partial_\theta)-4 ias\cos\theta T+2is\frac{a^2\sin\theta\cos\theta}{(r^2+a^2)}\mcu\label{eq:teukmodifdef}\\
    &+2 r (e_4-\mu e_3)+4 cs[(r-M)\mu^{-1}e_4- rT]+\frac{2ar }{r^2+a^2}\Phi+ s-s^2\frac{a^4\sin^2\theta\cos^2\theta}{(r^2+a^2)^2}+V.\nn
\end{align}
\end{defi}
\begin{rem}
    We remark the following identities, that are deduced from \eqref{eq:teukhate3e4} and \eqref{eq:teukmodifdef}:
    \begin{align}
        &\teukhat_{s}^{(c,V)}=\teukhat_s+4s(c-1)[(r-M)\mu^{-1}e_4-rT]+V,\label{eq:teukmodif}\\
        &\teukhat_{s}=\teukhat_s^{(1,0)}.\label{eq:teukavecteukmodif}
    \end{align}
In this paper, we will only use a finite number of constants $c,V$, so that we do not make explicit the dependance in $c,V$ in the notations $\lesssim$ and $O$. Notice that \eqref{eq:teukmodif} above implies, together with \eqref{eq:commutateurs},
\begin{align}\label{eq:commutateursmodif}
[\teukhat_s^{(c,V)},T]=[\teukhat_s^{(c,V)},\Phi]=[\teukhat_s^{(c,V)},\mcq_s]=0.
\end{align}
\end{rem}
In this section, we will consider a spin $-2$ scalar $\Psihat$ satisfying, for $\beta\in\mathbb{R}$, $c>0$, $V\in\mathbb{R}$ and $0\leq j\leq N_j'$, $0\leq 2k_1+k_2\leq N_k'$:
\begin{align}
    \bullet\:&T^j\mcq_{-2}^{k_1}\Phi^{k_2}\Psihat=O(\ubar^{-\beta-j})\text{   on   }\{r=r_\mathfrak{b}\},\label{eq:hypborneII}\\
    \bullet\:&        \iint_{\left\{\underline{w}_1 \leq \underline{w} \leq \underline{w}_2\right\} \cap \{r=\rb\}} \mathbf{e}[T^j\carterm^{k_1}\Phi^{k_2}\psihat]\dee\nu\mathrm{d} \underline{u}\lesssim\wbar_1^{-2\beta-2j}+\int_{\underline{w}_1}^{\underline{w}_2} \underline{w}^{-2\beta-2j} \mathrm{d} \underline{w},\label{eq:hypenerII}\\
    \bullet\:&\teukhat_{-2}^{(c,V)}T^j\mcq_{-2}^{k_1}\Phi^{k_2}\Psihat=O(\ubar^{-\beta-j})\text{   in   }\deux.\label{eq:hypteukII}
\end{align}
Note that \eqref{eq:hypteukII} will be satisfied by $\errm$ with $\beta=7+\delta$, see Proposition \ref{prop:teukdec}. The rest of this section is dedicated to proving the almost-propagation of the polynomial decay \eqref{eq:hypborneII} from the spacelike hypersurface $\{r=\rb\}$ to the whole region $\deux$, see Proposition \ref{prop:propagII}.  We will follow the strategy of \cite{scalarMZ} which introduces a log-degenerate energy control in the interior of the Kerr black hole for the scalar wave equation. This section is an extension of this method to the Teukolsky equation. When following the energy method of \cite{scalarMZ}, we will see that the negative $-2$ spin for the modified Teukolsky operator $\teukhat_{-2}^{(c,V)}$ will give the positivity of the bulk term in the energy estimate. Contrary to what was done in \cite{G24} for the spin $+2$ Teukolsky operator $\teuk_{+2}^{(c,V)}$, we cannot directly control the non-degenerate energy of the solution, but a degenerate one, and there is a loss of precision of order $\gamma$ in the final polynomial bound.
\begin{defi}
For $\Psihat$ a spin-weighted scalar, and $\alpha\geq 0$, we define 
    $$\ener_\alpha[\Psihat]:=(-\mu)|{\log(-\mu)}|^{-\alpha}|e_3\Psihat|^2+|e_4\Psihat|^2+|\mcu\Psihat|^2+|\partial_\theta\Psihat|^2.$$
\end{defi}
\begin{rem} Let us note the following concerning $\ener_\alpha$: 
\begin{itemize}
    \item The weight $(-\mu)|{\log(-\mu)}|^{-\alpha}$ degenerates at the event and Cauchy horizons.
    \item Although the proof of the main theorem works with any fixed value of $\alpha>1$, in practice we will only use the value $\alpha=3/2$, so we still denote the $\alpha$-dependent bounds with $O$ and $\lesssim$. For example, $\mu^2\lesssim(-\mu)|{\log(-\mu)}|^{-\alpha}$ so $\ener_{deg}[\Psihat]\lesssim\ener_\alpha[\Psihat]$ in $\deux\cup\trois$.
\end{itemize}
     
\end{rem}

We first prove the following energy estimate:

\begin{prop}\label{prop:enerestdeuxx}
    Let $\psihat$ be a spin $-2$ scalar satisfying \eqref{eq:hypteukII} with $\beta,V\in\mathbb{R}$, $c>0$. Then for $p=p(a,M)>0$ sufficiently large, and $\rb=\rb(a,M)$ sufficiently close to $r_-$, we have, for $0\leq j\leq N_j'$, $0\leq 2k_1+k_1\leq N_k'$, $\alpha>1$, and $1\leq\wbar_1\leq\wbar_2$:
    \begin{align*}
 &\iint_{\left\{\underline{w}=\underline{w}_2\right\} \cap \mathbf{II}} \mathbf{e}_\alpha[T^j\carterm^{k_1}\Phi^{k_2}\Psihat](-\mu) \dee\nu \mathrm{d} u+\iiint_{\left\{\underline{w}_1 \leq \underline{w} \leq \underline{w}_2\right\} \cap \mathbf{II}} \wbar^{-\gamma}\mathbf{e}_{\alpha}[T^j\carterm^{k_1}\Phi^{k_2}\Psihat](-\mu) \dee\nu \mathrm{d} u \mathrm{d} \underline{u}  \\
&+\iint_{\Gamma\cap\left\{\underline{w}_1 \leq \underline{w} \leq \underline{w}_2\right\}} \mathcal{T}_\Gamma[T^j\carterm^{k_1}\Phi^{k_2}\Psihat] \dee\nu \mathrm{d} \underline{u}\lesssim  \iint_{\left\{\underline{w}=\underline{w}_1\right\} \cap \mathbf{II}} \mathbf{e}_\alpha[T^j\carterm^{k_1}\Phi^{k_2}\Psihat](-\mu) \dee\nu \mathrm{d} u\\
&+\iint_{\left\{\underline{w}_1 \leq \underline{w} \leq \underline{w}_2\right\} \cap \{r=\rb\}} \mathbf{e}[T^j\carterm^{k_1}\Phi^{k_2}\Psihat]\dee\nu\mathrm{d} \underline{u}+\int_{\underline{w}_1}^{\underline{w}_2} \underline{w}^{-2\beta-2j} \mathrm{d} \underline{u},
\end{align*}
where for any $\psihat$, denoting $f(r)=|\log(-\mu)|^{-\alpha}$ and $g(r)=(r^2+a^2)^p$,
\begin{align*}
    \mathcal{T}_\Gamma[\psihat]=&2(r^2+a^2)g(r)|\partial_\ubar\psihat|^2-\frac{1}{2}\mu f(r)(|\partial_\theta\psihat|^2+|\mcu\psihat|^2)-a\sin\theta \mu \mathfrak{R}\left(\overline{X(\psihat)}\mcu\psihat)\right)\\
    &+(1-\gamma\ubar^{\gamma-1})\left(2(r^2+a^2)f(r)|\ethat\psihat|^2-\frac{1}{2}\mu g(r)(|\partial_\theta\psihat|^2+|\mcu\psihat|^2)+a\sin\theta \mu \mathfrak{R}\left(\overline{X(\psihat)}\mcu\psihat)\right)\right),
\end{align*}
and $X(\psihat)=|{\log(-\mu)}|^{-\alpha}\ethat\psihat+(r^2+a^2)^p\partial_\ubar\psihat$.
\end{prop}
\begin{proof}
It is enough to prove the case $j=k_1=k_2=0$. We fix the spin $s=-2$ but we do the computations for a general spin $s$. Recall that defining $\ethat=-\mu e_3$ we have in outgoing Eddington-Finkelstein coordinates $\partial_u=\ethat+a/(r^2+a^2)\Phi$ and $\partial_\ubar=e_4$. Thus using \eqref{eq:teukmodifdef} we get in $\deux$:
\begin{align}
    \mu\teukhat_s^{(c,V)}&=4(r^2+a^2)\ethat\partial_\ubar+ \mu \mcu^2+\frac{\mu}{\sin\theta}\partial_\theta(\sin\theta\partial_\theta)-4 ias\mu\cos\theta T+2is\mu \frac{a^2\sin\theta\cos\theta}{(r^2+a^2)}\mcu\label{eq:laprem}\\
    &\quad+2 r\mu (\ethat+\partial_\ubar)+4 cs[(r-M)\partial_\ubar- r\mu T]+\frac{2ar\mu  }{r^2+a^2}\Phi+ s\mu -s^2\mu\frac{a^4\sin^2\theta\cos^2\theta}{(r^2+a^2)^2}+V,\nonumber\\
    &=4(r^2+a^2)\partial_\ubar \ethat+ \mu \mcu^2+\frac{\mu}{\sin\theta}\partial_\theta(\sin\theta\partial_\theta)-4 ias\mu\cos\theta T+2is\mu \frac{a^2\sin\theta\cos\theta}{(r^2+a^2)}\mcu\label{eq:ladeuz}\\
    &\quad+2 r\mu (\ethat+\partial_\ubar)+4 cs[(r-M)\partial_\ubar- r\mu T]-\frac{2ar\mu  }{r^2+a^2}\Phi+ s\mu -s^2\mu\frac{a^4\sin^2\theta\cos^2\theta}{(r^2+a^2)^2}+V.\nonumber
\end{align}
Note that we used here
$$4(r^2+a^2)[\ethat,\partial_\ubar]=-\frac{4ar\mu}{r^2+a^2}\Phi.$$
Next, similarly as in \cite{scalarMZ} for the scalar wave equation, we multiply the modified inhomogeneous Teukolsky equation \eqref{eq:hypteukII} by $\mu$ and by the complex conjugate of $X(\Psihat):=f(r)\ethat\Psihat+g(r)\partial_\ubar\Psihat,$ where $$f(r):=|{\log(-\mu)}|^{-\alpha},\quad g(r)=(r^2+a^2)^p,$$
with $p=p(a,M)\gg 1$ and $\alpha>1$ that will be chosen later, and $\rb=\rb(a,M)$ that will be chosen sufficiently close to $r_-$. We take the real part, and integrate on $S(u,\ubar)$ with respect to $\dee\nu$ to get
\begin{align}
    \intS \Real\left(f(r)\overline{\ethat\Psihat}\mu\teukhat_{s}^{(c,V)}\Psihat\right)+\Real\left(g(r)\overline{\partial_\ubar\Psihat}\mu\teukhat_{s}^{(c,V)}\Psihat\right)\dee\nu=\intS\mu\Real\left(\overline{X(\Psihat)}O(\ubar^{-\beta})\right)\dee\nu.\label{eq:unefois}
\end{align}
By Lemma \ref{lem:intS}, this can be rewritten 
\begin{align}\label{eq:dudubar}    \partial_\ubar\left(\int_{S(u,\ubar)}\mathbf{F}_\ubar[\Psihat]\dee\nu\right)+\partial_u\left(\int_{S(u,\ubar)}\mathbf{F}_u[\Psihat]\dee\nu\right)+\int_{S(u,\ubar)}\mathbf{B}[\Psihat]\dee\nu=\intS\mu\Real\left(\overline{X(\Psihat)}O(\ubar^{-\beta})\right)\dee\nu,
\end{align}
where $\mathbf{F}_\ubar[\Psihat]$, $\mathbf{F}_u[\Psihat]$ and $\mathbf{B}[\Psihat]$ are defined in Lemma \ref{lem:intS}. As in \cite{scalarMZ}, integrating \eqref{eq:dudubar} on $\{\wbar_1\leq\wbar\leq\wbar_2\}\cap\deux$ with respect to $\dee u\dee\ubar$ gives 
\begin{align}
 \iint_{\left\{\underline{w}=\underline{w}_2\right\} \cap \mathbf{II}} \mathcal{T}_{\underline{w}}[\Psihat] \dee\nu \mathrm{d} u&+\iint_{\Gamma\cap\left\{\underline{w}_1 \leq \underline{w} \leq \underline{w}_2\right\}} \mathcal{T}_\Gamma[\Psihat] \dee\nu \mathrm{d} \underline{u}+\iiint_{\left\{\underline{w}_1 \leq \underline{w} \leq \underline{w}_2\right\} \cap \mathbf{II}} \mathbf{B}[\Psihat] \dee\nu \mathrm{d} u \mathrm{d} \underline{u} \nonumber\\
& =\iint_{\left\{\underline{w}=\underline{w}_1\right\} \cap \mathbf{II}} \mathcal{T}_{\underline{w}}[\Psihat] \dee\nu \mathrm{d} u+\iint_{\{r=r_{\mathfrak{b}}\}  \cap\left\{\underline{w}_1 \leq \underline{w} \leq \underline{w}_2\right\}} \mathcal{T}_r[\Psihat] \dee\nu \mathrm{d} \underline{u}\nn\\
&\quad+\iiint_{\left\{\underline{w}_1 \leq \underline{w} \leq \underline{w}_2\right\} \cap \mathbf{II}} \mu\mathfrak{R}\left(\overline{X(\Psihat)}O(\ubar^{-\beta})\right) \dee\nu \mathrm{d} u \mathrm{d} \underline{u},\label{eq:justeavant}
\end{align}
where $\mathcal{T}_{\underline{w}}[\Psihat]=\left(1-\frac{1}{2} \mu\right) \mathbf{F}_{\underline{u}}[\Psihat]-\frac{1}{2} \mu \mathbf{F}_u[\Psihat]$, $\mathcal{T}_r[\Psihat]=-\frac{1}{2} \mu\left(\mathbf{F}_{\underline{u}}[\Psihat]+\mathbf{F}_u[\Psihat]\right)$, and $\mathcal{T}_\Gamma[\Psihat]=\mathbf{F}_u[\Psihat]+(1-\gamma\ubar^{\gamma-1})\mathbf{F}_{{\ubar}}[\Psihat]$.
Now we estimate the different quantities in \eqref{eq:justeavant} taking into account the choice of $f,g$.

\noindent\textbf{Step 1: control of the bulk terms.} We begin by proving that $\mathbf{B}[\Psihat]$ controls $\ener_{\alpha+1}[\Psihat]$ in $\{r_-\leq r\leq\rb\}$ using a blueshift energy estimate. This estimate is present for the spin $-2$ rescaled Teukolsky operator $\teukhat_{-2}^{(c,V)}$ near $\ch$. It is also present for the spin $+2$ Teukolsky operator near $\ch$, where it is stronger in the sense that in that case, the bulk term controls the non-degenerate energy $\ener[\psi]$, see \cite{G24}. By Lemma \ref{lem:bulkpos}, we get that for $\alpha>1$, $s=-2$, $c>0$, $V\in\mathbb{R}$, for $p(a,M)>0$ large enough, and $r_\mathfrak{b}(p,a,M)$ close enough to $r_-$, we have in $\{r_-\leq r\leq\rb\}$,
\begin{align}\label{eq:bulkpos}
    \intS\mathbf{B}[\Psihat]\dee\nu\gtrsim(-\mu)\intS\ener_{\alpha+1}[\Psihat]\dee\nu.
\end{align}
Next we write, for $\varepsilon>0$,
\begin{align}
    \Bigg|\iiint&_{\left\{\underline{w}_1 \leq \underline{w} \leq \underline{w}_2\right\} \cap \mathbf{II}} \mu\Real(\overline{X(\Psihat)}O(\ubar^{-\beta})) \dee\nu \mathrm{d} u \mathrm{d} \underline{u}\Bigg|\lesssim \label{eq:abso}\\
    &\varepsilon\iiint_{\left\{\underline{w}_1 \leq \underline{w} \leq \underline{w}_2\right\} \cap \mathbf{II}} |X(\Psihat)|^2(-\mu)\dee\nu \mathrm{d} u \mathrm{d} \underline{u}+\varepsilon^{-1}\iiint_{\left\{\underline{w}_1 \leq \underline{w} \leq \underline{w}_2\right\} \cap \mathbf{II}} \ubar^{-2\beta} (-\mu)\dee\nu \mathrm{d} u \mathrm{d} \underline{u}.\nonumber
\end{align}
Changing variables\footnote{On $\wbar=cst$, we have $\mu\dee u=(2-\mu)\dee r$.} from $u$ to $r$, and using compactness in $r$, we get 
$$\iiint_{\left\{\underline{w}_1 \leq \underline{w} \leq \underline{w}_2\right\} \cap \mathbf{II}} \ubar^{-2\beta} (-\mu)\dee\nu \mathrm{d} u \mathrm{d} \underline{u}\lesssim\iint_{\left\{\underline{w}_1 \leq \underline{w} \leq \underline{w}_2\right\} \cap \mathbf{II}} \ubar^{-2\beta} \mathrm{d} r \mathrm{d} \underline{u}\lesssim\int_{\wbar_1}^{\wbar_2}\wbar^{-2\beta}\dee\ubar.$$
As $\mu^2|{\log(-\mu)}|^{-2\alpha}\lesssim(-\mu)|{\log(-\mu)}|^{-\alpha-1}$ in $\{r_-\leq r\leq\rb\}$, we have
$$|X(\Psihat)|^2\lesssim \mu^2|{\log(-\mu)}|^{-2\alpha}|e_3\Psihat|^2+|\partial_\ubar\Psihat|^2\lesssim\ener_{\alpha+1}[\Psihat].$$ Combining this fact with \eqref{eq:bulkpos} and \eqref{eq:justeavant}, taking $\varepsilon>0$ small enough such that the first term on the RHS of \eqref{eq:abso} is absorbed in the LHS of \eqref{eq:justeavant} using Lemma \ref{lem:bulkpos}, we get 
\begin{align}
 &\iint_{\left\{\underline{w}=\underline{w}_2\right\} \cap \mathbf{II}} \mathcal{T}_{\underline{w}}[\Psihat] \dee\nu \mathrm{d} u+\iint_{\Gamma\cap\left\{\underline{w}_1 \leq \underline{w} \leq \underline{w}_2\right\}} \mathcal{T}_\Gamma[\Psihat] \dee\nu \mathrm{d} \underline{u}+\iiint_{\left\{\underline{w}_1 \leq \underline{w} \leq \underline{w}_2\right\} \cap \mathbf{II}} \ener_{\alpha+1}[\Psihat] \dee\nu \mathrm{d} u \mathrm{d} \underline{u} \nonumber\\
& \lesssim\iint_{\left\{\underline{w}=\underline{w}_1\right\} \cap \mathbf{II}} \mathcal{T}_{\underline{w}}[\Psihat] \dee\nu \mathrm{d} u+\iint_{\{r=r_{\mathfrak{b}}\}  \cap\left\{\underline{w}_1 \leq \underline{w} \leq \underline{w}_2\right\}} \mathcal{T}_r[\Psihat] \dee\nu \mathrm{d} \underline{u}+\int_{\wbar_1}^{\wbar_2}\wbar^{-2\beta}\dee\ubar.\label{eq:justeavant23}
\end{align}
As in \cite[(4.19)]{scalarMZ}, using the definition of region $\deux$ we get 
$$\ener_{\alpha+1}[\Psihat]\gtrsim|{\log(-\mu)}|^{-1}\ener_{\alpha}[\Psihat]\sim(r^*)^{-1}\ener_{\alpha}[\Psihat]\gtrsim\wbar^{-\gamma}\ener_{\alpha}[\Psihat]$$
in $\deux$ which implies, re-injecting in \eqref{eq:justeavant23},
\begin{align}
 &\iint_{\left\{\underline{w}=\underline{w}_2\right\} \cap \mathbf{II}} \mathcal{T}_{\underline{w}}[\Psihat] \dee\nu \mathrm{d} u+\iint_{\Gamma\cap\left\{\underline{w}_1 \leq \underline{w} \leq \underline{w}_2\right\}} \mathcal{T}_\Gamma[\Psihat] \dee\nu \mathrm{d} \underline{u}+\iiint_{\left\{\underline{w}_1 \leq \underline{w} \leq \underline{w}_2\right\} \cap \mathbf{II}} \wbar^{-\gamma}\ener_{\alpha}[\Psihat] \dee\nu \mathrm{d} u \mathrm{d} \underline{u} \nonumber\\
& \lesssim\iint_{\left\{\underline{w}=\underline{w}_1\right\} \cap \mathbf{II}} \mathcal{T}_{\underline{w}}[\Psihat] \dee\nu \mathrm{d} u+\iint_{\{r=r_{\mathfrak{b}}\}  \cap\left\{\underline{w}_1 \leq \underline{w} \leq \underline{w}_2\right\}} \mathcal{T}_r[\Psihat] \dee\nu \mathrm{d} \underline{u}+\int_{\wbar_1}^{\wbar_2}\wbar^{-2\beta}\dee\ubar.\label{eq:justeavant2}
\end{align}
\noindent\textbf{Step 2: control of the boundary terms.} 
The integrated term on $\wbar=cst$ is:
\begin{align*}
     \mathcal{T}_{\underline{w}}[\Psihat]&=\left(1-\mu/2\right)\mathbf{F}_{\underline{u}}[\Psihat]-\dfrac{1}{2} \mu \mathbf{F}_u[\Psihat]\\
     &=2(r^2+a^2)f(r)|\ethat\Psihat|^2-\dfrac{1}{2}\mu g(r)(|\partial_\theta\Psihat|^2+|\mcu\Psihat|^2)+a\sin\theta\mu \mathfrak{R}(X(\Psihat)\mcu\Psihat)\\
     &\quad\quad-\frac{\mu}{2}\Bigg[2(r^2+a^2)f(r)|\ethat\Psihat|^2+2(r^2+a^2)g(r)|\partial_\ubar\Psihat|^2-\dfrac{1}{2}\mu (f(r)+g(r))(|\partial_\theta\Psihat|^2+|\mcu\Psihat|^2)\Bigg]\\
     &\sim (-\mu)\ener_\alpha[\Psihat],
\end{align*}
where we absorbed the term $a\sin\theta\mu \mathfrak{R}(X(\Psihat)\mcu\Psihat)$ as in \cite{scalarMZ} using 
\begin{align*}
    |a\sin\theta \mu f(r) \mathfrak{R}(\overline{\ethat\Psihat}\mcu\Psihat)|&\leq a^2f(r)|\ethat\Psihat|^2+\frac{1}{4}\mu^2f(r)|\mcu\Psihat|^2,\\
    |a\sin\theta\mu g(r) \mathfrak{R}(\overline{\partial_\ubar\Psihat}\mcu\Psihat)|&\leq -\frac{3}{4}\mu a^2g(r)|\partial_\ubar\Psihat|^2-\frac{1}{3}\mu g(r)|\mcu\Psihat|^2.
\end{align*}
Then, we have 
$$\mathcal{T}_r[\Psihat]=-\dfrac{1}{2} \mu\left(\mathbf{F}_{\underline{u}}[\Psihat]+\mathbf{F}_u[\Psihat]\right)\sim \ener[\Psihat],\text{   on   }\{r=\rb\}.$$
 Combining this control of the boundary terms with \eqref{eq:justeavant2} concludes the proof of Proposition \ref{prop:enerestdeuxx}.
\end{proof}
\begin{prop}\label{prop:enerdeuxx}
    Assume that $\psihat$ is a spin $-2$ scalar satisfying \eqref{eq:hypenerII} and \eqref{eq:hypteukII}. Then for $\alpha>1$ and $\gamma>0$ small enough we have, for $0\leq j\leq N_j'$, $0\leq 2k_1+k_2\leq N_k'$:
    $$\iint_{\left\{\wbar={\wbar}_1\right\} \cap \mathbf{II}} \ener_\alpha[T^j\mcq_{-2}^{k_1}\Phi^{k_2}\Psihat](-\mu) \dee\nu \mathrm{d} u\lesssim \wbar_1^{-2\beta-2j+\gamma}.$$
\end{prop}
\begin{proof}
    Once again, it suffices to prove the case $j=k=0$. Proposition \ref{prop:enerestdeuxx} and the energy assumption \eqref{eq:hypenerII} on $\{r=\rb\}$ gives
\begin{align}
 \iint&_{\left\{\underline{w}=\underline{w}_2\right\} \cap \mathbf{II}} \mathbf{e}_\alpha[\Psihat](-\mu) \dee\nu \mathrm{d} u+\iiint_{\left\{\underline{w}_1 \leq \underline{w} \leq \underline{w}_2\right\} \cap \mathbf{II}} \wbar^{-\gamma}\mathbf{e}_{\alpha}[\Psihat](-\mu) \dee\nu \mathrm{d} u \mathrm{d} \underline{u} \label{eq:pourapres} \\
&\lesssim  \iint_{\left\{\underline{w}=\underline{w}_1\right\} \cap \mathbf{II}} \mathbf{e}_\alpha[\Psihat](-\mu) \dee\nu \mathrm{d} u+\wbar_1^{-2\beta}+\int_{\underline{w}_1}^{\underline{w}_2} \underline{w}^{-2\beta} \mathrm{d} \underline{u},\nonumber
\end{align}
where we used the fact that the term on $\Gamma$ can be chosen non negative as
$$\mathcal{T}_\Gamma[\Psihat]=\widehat{\mathbf{F}}_u[\Psihat]+(1-\gamma\ubar^{\gamma-1})\widehat{\mathbf{F}}_{{\ubar}}[\Psihat]\geq f(r)|\ethat\Psihat|^2+g(r)|e_4\Psihat|^2-\mu(f(r)+g(r))(|\partial_\theta\Psihat|^2+|\mcu\Psihat|^2)\geq 0$$
for $\gamma>0$ small enough. Using the fact that for any $\wbar_1\geq 1$, the energy 
$$\iint_{\left\{\underline{w}=\underline{w}_1\right\} \cap \mathbf{II}} \mathbf{e}_\alpha[\Psihat](-\mu) \dee\nu \mathrm{d} u\lesssim 1$$
is finite\footnote{In the notation of Lemma \ref{lem:decay}, this means that $f(1)$ is finite. This holds because $\{\wbar = \wbar_1\}\cap\mathbf{II}$ is a compact region inside a globally hyperbolic spacetime and by standard existence results for linear wave equations with smooth initial data.}, we conclude the proof using Lemma \ref{lem:decay} with $p=2\beta$.
\end{proof}

The following result will be used to deduce decay in $L^2(S(u,\ubar))$ from energy decay.
\begin{lem}\label{lem:integ}
    Let $\psi$ be a spin $\pm 2$ scalar. We define the scalar function $$f(\wbar,r):=\|\psi\|_{L^2(S(u(r,\wbar),\ubar(r,\wbar)))}.$$
    Then for $\wbar\geq 1$ and $r_2\leq r_1$ such that $(\wbar,r_1),(\wbar,r_2)\in\deux$,
    $$f(\wbar,r_2)\lesssim f(\wbar,r_1)+\wbar^{\frac{1}{2}\gamma(\alpha+1)}\left(\int_{r_2}^{r_1}\int_{S(\wbar,r')}\ener_\alpha[\psi]\dee\nu\dee r'\right)^{1/2}.$$
\end{lem}
\begin{proof}
    In coordinates $(\ubar,r,\theta,\phi_+)$ we have $\partial_{r_\mathrm{in}}=-2e_3$, $\partial_{\ubar}=T$.
    Thus for a fixed $\wbar$,
    \begin{align}\label{eq:intapres}
        \frac{\partial}{\partial r}(f(\wbar,r))=\left[\frac{\partial}{\partial r}(\wbar+r-r_+)\partial_{\ubar_\mathrm{in}}+\frac{\partial r}{\partial r}\partial_{r_\mathrm{in}}\right]f(\ubar,r)=(\partial_{\ubar_\mathrm{in}}+\partial_{r_\mathrm{in}})f(\ubar,r).
    \end{align}
    Moreover, by a Cauchy-Schwarz inequality,
    \begin{align*}
        |(\partial_{r_\mathrm{in}}+\partial_{\ubar})(f^2)|&=\left|(\partial_{r_\mathrm{in}}+\partial_{\ubar})\left(\int_0^\pi\int_0^{2\pi}|\psi|^2(\ubar,r,\theta,\phi_+)\sin\theta\dee\theta\dee\phi_+\right)\right|\\
        &=\left|2\mathfrak{R}\left(\int_0^\pi\int_0^{2\pi}(\overline{\psi} (\partial_{r_\mathrm{in}}+\partial_{\ubar})\psi)(\ubar,r,\theta,\phi_+)\sin\theta\dee\theta\dee\phi_+\right)\right|\\
        &\lesssim 2f\|(\partial_{r_\mathrm{in}}+\partial_{\ubar})\psi\|_{L^2(S(\wbar,r))}=2f\|(T-2e_3)\psi\|_{L^2(S(\wbar,r))},
    \end{align*}
    which yields
    \begin{align}\label{eq:pourint1}
        |(\partial_{r_\mathrm{in}}+\partial_{\ubar})f|=\frac{1}{2f}|(\partial_{r_\mathrm{in}}+\partial_{\ubar})(f^2)|\lesssim\|(T-2e_3)\psi\|_{L^2(S(\wbar,r))}.
    \end{align}
    Integrating \eqref{eq:intapres}, together with the bound \eqref{eq:pourint1} gives
    \begin{align}
        f(\wbar,r_2)&\lesssim f(\wbar,r_1)+\int_{r_2}^{r_1}\left(\int_{S(\wbar,r')}|T\psi-2e_3\psi|^2\dee\nu\right)^{1/2}\dee r'\nn\\
        &\lesssim f(\wbar,r_1)+\int_{r_2}^{r_1}\left(\int_{S(\wbar,r')}|T\psi|^2\dee\nu\right)^{1/2}\dee r'+\int_{r_2}^{r_1}\left(\int_{S(\wbar,r')}|e_3\psi|^2\dee\nu\right)^{1/2}\dee r'.\label{eq:otherterm}
    \end{align}
    Next, we use the expression \eqref{eq:expreT} which gives $T=O(\mu)e_3+O(1)e_4+O(1)\mcu+O(1)$. Using $\mu^2\lesssim -\mu|{\log(-\mu)}|^{-\alpha}$ and the Poincaré inequality \eqref{eq:poincare} we get
    \begin{align}\label{eq:boboT}
        \int_{S(\wbar,r')}|T\psi|^2\dee\nu\lesssim \int_{S(\wbar,r')}\ener_\alpha[\psi]\dee\nu.
    \end{align}
    Now we treat the last term on the RHS of \eqref{eq:otherterm} which is more difficult to bound. Using a Cauchy-Schwarz inequality we get 
\begin{align}
    \int_{r_2}^{r_1}\Bigg(\int_{S(\wbar_1,r)}&|e_3\psi|^2\dee\nu\Bigg)^{1/2}\nn\\
    &\lesssim \left(\int_{r_2}^{r_1}(-\mu)^{-1}|{\log(-\mu)}|^{\alpha}\dee r\right)^{1/2}\left(\int_{r_2}^{r_1}\int_{S(\wbar,r')}(-\mu)|{\log(-\mu)}|^{-\alpha}|e_3\psi|^2\dee\nu\dee r'\right)^{1/2}\nn\\
    &\lesssim\left(\int_{r_2}^{r_1}(-\mu)^{-1}|{\log(-\mu)}|^{\alpha}\dee r\right)^{1/2}\left(\int_{r_2}^{r_1}\int_{S(\wbar,r')}\ener_\alpha[\psi]\dee\nu\dee r'\right)^{1/2}.\label{eq:elle}
\end{align}
All is left to do is to estimate the first integral on the RHS of \eqref{eq:elle}. We have in $\deux$, 
$$|{\log(-\mu)}|\sim r^*\lesssim\ubar^\gamma\lesssim\wbar^\gamma.$$
We also have 
$$|u_{r_2}(\wbar)-u_{r_1}(\wbar)|=|2r^*-r-2\rb^*+\rb|\lesssim\wbar^\gamma$$
in $\deux$. This gives 
$$\int_r^{\rb}(-\mu)^{-1}|{\log(-\mu)}|^{\alpha}\dee r\lesssim\int_{u_{\rb}(\wbar)}^{u_r(\wbar)}|{\log(-\mu)}|^{\alpha}\dee u\lesssim\wbar^{\gamma(\alpha+1)}$$
and concludes the proof of Lemma \ref{lem:integ}.
\end{proof}
\begin{prop}\label{prop:propagII}
    Let $\Psihat$ be a spin $-2$ scalar satisfying \eqref{eq:hypborneII}, \eqref{eq:hypenerII} and \eqref{eq:hypteukII}. Then for $\gamma>0$ small enough we have in $\deux$, for $0\leq j\leq N_j'-2$ and $0\leq 2k_1+k_2\leq N_k'-2$:
    $$|T^j\carterm^{k_1}\Phi^{k_2}\Psihat|\lesssim\ubar^{-\beta-j+C_\alpha\gamma},$$
    where $C_\alpha=\alpha/2+1$.
\end{prop}
\begin{proof}
    It is enough to treat the case $j=k_1=k_2=0$. The loss of two derivatives will be due to the use of the Sobolev embedding \eqref{eq:sobolev}. Using Lemma \ref{lem:integ} between $r(u,\ubar)$ and $\rb$, and Proposition \ref{prop:enerdeuxx} gives
\begin{align*}
    \|\Psihat\|_{L^2(S(u,\ubar))}\lesssim\|\Psihat\|_{L^2(S(u_{\rb}(\wbar),\ubar_{\rb}(\wbar))}+\wbar^{\frac{1}{2}\gamma(\alpha+1)}\left(\iint_{\left\{\wbar={\wbar(u,\ubar)}\right\} \cap \mathbf{II}} \ener_\alpha[T^j\mcq_s^k\Psihat](-\mu) \dee\nu \mathrm{d} u\right)^{1/2}.
\end{align*}

Thus denoting $C_\alpha=\alpha/2+1$ we have shown, using also \eqref{eq:hypborneII} to control the $L^2(S(u,\ubar))$ norm on $\{r=\rb\}$, 
\begin{align}\label{eq:elledeux}
    \|\Psihat\|_{L^2(S(u,\ubar))}\lesssim\wbar^{-\beta+C_\alpha\gamma}.
\end{align}
Together with the Sobolev embedding \eqref{eq:sobolev} and the $j,k_1\neq 0 $ version of \eqref{eq:elledeux}, we infer
$$|\Psihat|^2\lesssim\|\carterm\Psihat\|^2_{L^2(S(u,\ubar))}+\|T^{\leq 2}\Psihat\|^2_{L^2(S(u,\ubar))}\lesssim\wbar_1^{-2\beta+2C_\alpha\gamma},$$
as stated. This concludes the proof of Proposition \ref{prop:propagII}.
\end{proof}
\subsection{Upper bounds for $\errhatm$ in intermediate region $\deux$}
In this section, we prove a non-sharp upper bound for $T^j\Psihatm$ in region $\deux$. To this end, we use the estimates from Section \ref{section:energymethod} and the bounds \eqref{eq:borneI}, \eqref{eq:enerI} in $\un$ that were deduced from assumption \eqref{eq:hyppl} on the event horizon. The following result is a simple computation:
\begin{prop}\label{prop:teukdec}
    We have in $\deux\cup\trois$, for $j,k_1,k_2\geq 0$, 
    $$\left|\drout^{\leq 1}\teukhat_{-2}T^j\carterm^{k_1}\Phi^{k_2}\left(\anshatm\right)\right|\lesssim\ubar^{-7-\delta-j}.$$
\end{prop}
\begin{proof}
    We have, using the definition of the rescaled operator $\teukhat_{-2}$, 
    \begin{align*}
        &\left|\drout^{\leq 1}\teukhat_{-2}T^j\carterm^{k_1}\Phi^{k_2}\left(\anshatm\right)\right|\\
        =&\left|\drout^{\leq 1}\Delta^2\teuk_{-2}T^j\carterm^{k_1}\Phi^{k_2}\left(\ansatzm\right)\right|\\
        =&\left|\drout^{\leq 1}\Delta^2T^j\carterm^{k_1}\Phi^{k_2}\teuk_{-2}^{[\partial_t]}\left(\ansatzm\right)\right|\\
        \lesssim& -\Delta\ubar^{-8-j}\lesssim\ubar^{-7-j-\delta},
    \end{align*}
    where we used the computation in \eqref{eq:teukerr} in the third line.
\end{proof}
\begin{rem}
     Note that to get the bound in Proposition \ref{prop:teukdec}, we used the estimate $|\Delta|\lesssim 1$ in $\deux$, which is too imprecise to derive a sharp lower bound for $\Psihatm$. This section is in fact only devoted to proving a somewhat precise upper bound for the $T$ derivatives of $\errhatm$ (see Proposition \ref{prop:upperboundII} below), and the precise lower bound will be obtained in Section \ref{section:precise}. 
\end{rem}
\begin{prop}\label{prop:upperboundII}
    Assume that $\Psim$ satisfies \eqref{eq:hyppl}. For $\gamma>0$ small enough, we have in $\deux$, 
    \begin{align*}
        \Psihatm=\anshatm+\errhatm,
    \end{align*}
    where for $0\leq j\leq N_j-4$ and $0\leq 2k_1+k_2\leq N_k-5$, 
    $$|T^j\carterm^{k_1}\Phi^{k_2}\drout^{\leq 1}\errhatm|\lesssim\ubar^{-7-j-\delta/2}.$$
\end{prop}
\begin{rem}
    The result of Proposition \ref{prop:upperboundII} only gives a sharp lower bound for $\Psihatm$ in any subregion of $\deux$ where $r$ is bounded away from $r_-$. This lower bound degenerates at $\Gamma$ where $\Delta^2/\ubar^7$ is much smaller than $\ubar^{-7-\delta/2}$ ($\Delta$ decays exponentially in $\ubar$ on $\Gamma$).
\end{rem}
\begin{proof}[Proof of Proposition \ref{prop:upperboundII}] We recall 
$$\errhatm=\Psihatm-\anshatm.$$
We have, combining the Teukolsky equation $\teukhat_{-2}\Psihatm=0$, \eqref{eq:commutateurs} and Proposition \ref{prop:teukdec}:    $$\teukhat_{-2}T^j\carterm^{k_1}\Phi^{k_2}\errhatm=O(\ubar^{-7-j-\delta}).$$
    Thus using the pointwise bound \eqref{eq:borneI}, the energy bound \eqref{eq:enerI} on $\{r=\rb\}$, \eqref{eq:teukavecteukmodif} and Proposition \ref{prop:propagII} with $c=1>0$, $V=0$, $\beta=7+\delta$, gives the result without the $\drout$ derivative, i.e. 
\begin{align}\label{eq:sansdr}
    |T^j\carterm^{k_1}\Phi^{k_2}\errhatm|\lesssim\ubar^{-7-j-\delta+C_\alpha\gamma}\lesssim\ubar^{-7-j-\delta/2}
\end{align}
where $\alpha=3/2$, for $\gamma>0$ small enough. Now we prove 
\begin{align}
    |T^j\carterm^{k_1}\Phi^{k_2}\drout\errhatm|\lesssim\ubar^{-7-j-\delta/2}.
\end{align}
Using Proposition \ref{prop:commdrout} for the expression of $[\teukhat_{-2},\drout]$, and Proposition \ref{prop:teukdec}
we find that $\drout\errhatm$ satisfies, in $\deux\cup\trois$: 
\begin{align}\label{eq:pde}
    \Big(\teukhat_{-2}+4[(r-M)\mu^{-1}e_4-rT]-2\Big)&T^j\carterm^{k_1}\Phi^{k_2}\drout\errhatm=\nn\\
    &-6T^j\carterm^{k_1}\Phi^{k_2}T\errhatm+O(\ubar^{-7-j-\delta}).
\end{align}
Thus using \eqref{eq:sansdr}, as well as the expression \eqref{eq:teukmodif} of the modified Teukolsky operator, we get 
\begin{align}\label{eq:LHS}
   \teukhat_{-2}^{(1/2,-2)}T^j\carterm^{k_1}\Phi^{k_2}\drout\errhatm=O(\ubar^{-7-j-\delta}).
\end{align}
We now prove that $\drout\errhatm$ satisfies \eqref{eq:hypborneII} and \eqref{eq:hypenerII}. This comes from the identity 
$$\drout=-2e_3+O(1)T+O(1)\Phi,\quad\text{ on }\{r=\rb\},$$
as well as the bounds \eqref{eq:borneI}, \eqref{eq:enerI}. We conclude the proof using Proposition \ref{prop:propagII} for $\drout\errhatm$ with $c=1/2>0$, $V=-2$, $\beta=7+\delta$ and $\gamma>0$ small enough.
\end{proof}
\subsection{Energy method for $\teukhat_{-2}$ in region $\trois$}\label{section:energymethodIII}
Sections \ref{section:energymethodIII} and \ref{section:upperboundIII} are dedicated to propagating the bounds of Proposition \ref{prop:upperboundII} from region $\deux$ to region $\trois$. This will be achieved using, in region $\trois$, a version of the energy method which is more naive than the one used in $\deux$. We now implement this energy method. In this section, $\alpha>1$ is fixed.

In this section we consider a spin $-2$ scalar $\Psihat$ such that there is $B\in\mathbb{R}$, $c>0$, $V\in\mathbb{R}$ such that:
\begin{align}
    \bullet\:&\iint_{\Gamma\cap\{\wbar_1\leq\wbar\leq\wbar_2\}}\mathcal{T}_\Gamma[\Psihatm]\dee\nu\dee\ubar\lesssim 1\text{  for any  }1\leq \wbar_1\leq\wbar_2,\label{eq:hypenerIII}\\
    \bullet\:&\teukhat_{-2}^{(c,V)}\Psihat=O(\ubar^{B})\text{   in   }\trois,\label{eq:hypteukIII}
\end{align}
recalling the expression of $\mathcal{T}_\Gamma[\Psihatm]$ from Proposition \ref{prop:enerestdeuxx}.
\begin{prop}\label{prop:enerborneIII}
    Let $\Psihat$ be a spin $-2$ scalar satisfying \eqref{eq:hypenerIII} and \eqref{eq:hypteukIII}. Then for $\wbar_1\geq 1$:
    \begin{align*}
        \iint_{\{\wbar=\wbar_1\}\cap\trois}\ener_\alpha[\Psihat]\dee\nu\dee r\lesssim 1+\int_{\wbar_{\rb}\!(w_{\rb,\gamma})}^{\wbar_1}\wbar^{2B}\dee\wbar;
\end{align*}
\end{prop}
\begin{proof}
We use the same energy method as in region $\deux$, but this time we integrate \eqref{eq:dudubar} with $\beta=-B$ on $\{w\leq w_{\rb,\gamma},\:\wbar\leq\wbar_1\}\cap\trois$. Using Lemma \ref{lem:bulkpos} and the bound equivalent to \eqref{eq:abso} in $\{w\leq w_{\rb,\gamma},\:\wbar\leq\wbar_1\}\cap\trois$ for $\varepsilon>0$ small enough, we get similarly as in \eqref{eq:justeavant23}:
\begin{align}
        &\iint_{\{w\leq w_{\rb,\gamma},\:\wbar=\wbar_1\}\cap\trois}{\ener_\alpha}[\Psihat](-\mu)\dee\nu\dee u+\iint_{\{\wbar\leq \wbar_1,\:w=w_{\rb,\gamma}\}\cap\trois}\mathcal{T}_{w}[\Psihat](-\mu)\dee\nu\dee \ubar\nonumber\\
        &+\iiint_{\{w\leq w_{\rb,\gamma},\:\wbar\leq\wbar_1\}\cap\trois}\ener_{\alpha+1}[\Psihat](-\mu)\dee\nu\dee u\dee\ubar\lesssim \nonumber\\
        &\quad\quad\quad\quad\quad\quad\quad\iint_{\Gamma\cap\{\wbar_{\rb}\!(w_{\rb,\gamma})\leq\wbar\leq\wbar_1\}}\mathcal{T}_\Gamma[\Psihat]\dee\nu\dee\ubar+\int_{\wbar_{\rb}\!(w_{\rb,\gamma})}^{\underline{w}_1} \underline{w}^{2B} \mathrm{d} \underline{w},\label{eq:similarly}
\end{align}
where, as in \cite{scalarMZ}, the boundary term on $\{w=w_{\rb,\gamma}\}$ is 
$$\mathcal{T}_{w}[\Psihat]:=(1-\mu/2)\mathbf{F}_u[\Psihat]-(\mu/2)\mathbf{F}_\ubar[\Psihat]\sim|\partial_\ubar\Psihat|^2-\mu|{\log(-\mu)}|^{-\alpha}(|\ethat\Psihat|^2+|\partial_\theta\Psihat|^2+|\mcu\Psihat|^2)\geq 0.$$
Thus from \eqref{eq:hypenerIII} we deduce, dropping two non-negative terms on the LHS of \eqref{eq:similarly}, 
$$\iint_{\{\wbar=\wbar_1\}\cap\trois}{\ener_\alpha}[\Psihat](-\mu)\dee\nu\dee u\lesssim 1+\int_{\wbar_{\rb}\!(w_{\rb,\gamma})}^{\underline{w}_1} \underline{w}^{2B} \mathrm{d} \underline{w},$$
and we conclude the proof changing variables from $u$ to $r$ in the integral on the left.
\end{proof}
\subsection{Upper bound for $\errhatm$ in region $\trois$}\label{section:upperboundIII}

In this section, using the naive non-decaying energy estimate of Section \ref{section:energymethodIII}, we first find a polynomial upper bound for $\errhatm$ in $\trois$. Commuting the Teukolsky equation with $\drout$ and applying the energy method again, this will give a polynomial upper bound for $\drout\errhatm$ in $\trois$, which, integrated from $\Gamma$, will propagate the $O(|u|^{-7-j-\delta/2})$ decay of $T^j\errhatm$ from $\Gamma$ to $\trois$, using the exponential decay of $\Delta$ in $\ubar$ in $\trois$. Finally, we integrate a 1+1 wave equation extracted from the Teukolsky equation to prove the $O(|u|^{-7-j-\delta/2})$ decay of $T^j\drout\errhatm$ in $\trois$.

\begin{prop}\label{prop:result}Assume that $\Psim$ satisfies \eqref{eq:hyppl}. We have in $\trois$, for $j,k_1,k_2\geq 0$, and $\alpha>1$,
\begin{align*}
        \iint_{\{\wbar=\wbar_1\}\cap\trois}\ener_\alpha[T^j\carterm^{k_1}\Phi^{k_2}\errhatm]\dee\nu\dee r\lesssim 1.
\end{align*}
\end{prop}
\begin{proof}
Using Proposition \ref{prop:teukdec} and Proposition \ref{prop:enerestdeuxx} with $\Psihat=\errhatm$ and $\beta=7+\delta$, as well as Proposition \ref{prop:enerdeuxx}, we get that $T^j\carterm^{k_1}\Phi^{k_2}\errhatm$ satisfies \eqref{eq:hypenerIII} and \eqref{eq:hypteukIII} with $B=-7-\delta$. Thus, using Proposition \ref{prop:enerborneIII}, we get 
    \begin{align*}
        \iint_{\{\wbar=\wbar_1\}\cap\trois}\ener_\alpha[T^j\carterm^{k_1}\Phi^{k_2}\errhatm]\dee\nu\dee r\lesssim 1+\int_{\wbar_{\rb}\!(w_{\rb,\gamma})}^{\wbar_1}\wbar^{-2(7+\delta)}\dee\wbar\lesssim 1,
\end{align*}
which concludes the proof.
\end{proof}
The following result will be used to deduce polynomial upper bounds from the energy estimate of Proposition \ref{prop:result}.
\begin{lem}\label{lem:integIII}
    Let $\psi$ be a spin $\pm 2$ scalar. We define the scalar function $$f(\wbar,r):=\|\psi\|_{L^2(S(u(r,\wbar),\ubar(r,\wbar)))}.$$
    Then for $\wbar\geq 1$ and $r_2\leq r_1$ such that $(\wbar,r_1),(\wbar,r_2)\in\trois$,
    $$f(\wbar,r_2)\lesssim f(\wbar,r_1)+\wbar^{\frac{1}{2}(\alpha+1)}\left(\int_{r_2}^{r_1}\int_{S(\wbar,r')}\ener_\alpha[\psi]\dee\nu\dee r'\right)^{1/2}.$$
\end{lem}
\begin{proof}
    Recall \eqref{eq:otherterm} and \eqref{eq:elle} which yield
    \begin{align}
        f(\wbar,r_2)\lesssim&\:f(\wbar,r_1)+\int_{r_2}^{r_1}\left(\int_{S(\wbar,r')}|T\psi|^2\right)^{1/2}\dee r'\nn\\
        &+\left(\int_{r_2}^{r_1}(-\mu)^{-1}|{\log(-\mu)}|^{\alpha}\dee r\right)^{1/2}\left(\int_{r_2}^{r_1}\int_{S(\wbar,r')}\ener_\alpha[\psi]\dee\nu\dee r'\right)^{1/2}.\label{eq:othertermIII}
    \end{align}
By \eqref{eq:boboT}, it is only left to estimate the first integral on the second line of \eqref{eq:othertermIII}. We have in $\trois$, 
$$|{\log(-\mu)}|\sim r^*\sim u+\ubar\lesssim \wbar$$
for $\wbar\gtrsim 1$ large enough. We also have 
$$|u_{r_2}(\wbar)-u_{r_1}(\wbar)|=|2r^*-r-2\rb^*+\rb|\lesssim\wbar$$
in $\trois$ for $\wbar\gtrsim 1$ large enough. This gives 
$$\int_r^{\rb}(-\mu)^{-1}|{\log(-\mu)}|^{\alpha}\dee r\lesssim\int_{u_{\rb}(\wbar)}^{u_r(\wbar)}|{\log(-\mu)}|^{\alpha}\dee u\lesssim\wbar^{\alpha+1}$$
and concludes the proof.
\end{proof}
\begin{prop}\label{prop:Linf}
Assume that $\Psim$ satisfies \eqref{eq:hyppl}. Then, in $\trois$, for $0\leq j\leq N_j-4$ and $0\leq 2k_1+k_2\leq N_k-5$, and $\alpha>1$,
$$|T^j\carterm^{k_1}\Phi^{k_2}\errhatm|\lesssim \ubar^{\frac{1}{2}(\alpha+1)}.$$
\end{prop}
\begin{proof}
We begin by proving the bound in $L^2(S(u,\ubar))$ norm. We use Lemma \ref{lem:integIII} to get in $\trois$:
\begin{align*}
    \|&T^j\carterm^{k_1}\Phi^{k_2}\errhatm\|_{L^2(S(u,\ubar))} \\
    &\lesssim\|T^j\carterm^{k_1}\Phi^{k_2}\errhatm\|_{L^2(S(u_\Gamma(\wbar),\ubar_\Gamma(\wbar)))}+\wbar^{\frac{1}{2}(\alpha+1)}\left(\iint_{\{\wbar=\wbar_1\}\cap\trois}\ener_\alpha[T^j\carterm^{k_1}\Phi^{k_2}\errhatm]\dee\nu\dee r\right)^{1/2}\\
    &\lesssim\wbar^{\frac{1}{2}(\alpha+1)}\lesssim\ubar^{\frac{1}{2}(\alpha+1)},
\end{align*}
where we used Proposition \ref{prop:upperboundII} to bound the term on $\Gamma$ and Proposition \ref{prop:result} to bound the energy term. Next, using the Sobolev embedding \eqref{eq:sobolev}, we get 
$$|T^j\carterm^{k_1}\Phi^{k_2}\errhatm|^2\lesssim \int_{S(u,\ubar)}|T^{\leq j+2}\carterm^{k_1}\Phi^{k_2}\errhatm|^2\dee\nu+\int_{S(u,\ubar)}|T^j\carterm^{k_1+1}\Phi^{k_2}\errhatm|^2\dee\nu\lesssim \ubar^{\alpha+1},$$
which concludes the proof. The loss of derivatives is the same as the one in Proposition \ref{prop:upperboundII} because, although we use the Sobolev embedding \eqref{eq:sobolev}, we only use the decay of $\errm$ on $\Gamma$ in $L^2(S(u,\ubar))$, which saves two angular and $T$ derivatives.
\end{proof}
We now turn to getting the polynomial upper bound for $\drout\errhatm$ in $\trois$.
\begin{prop}\label{prop:enerdroutIII}Assume that $\Psim$ satisfies \eqref{eq:hyppl}. Then, we have in $\trois$, for $0\leq j\leq N_j-2$ and $0\leq 2k_1+k_2\leq N_k-3$, and $\alpha>1$,
\begin{align*}
        \iint_{\{\wbar=\wbar_1\}\cap\trois}\ener_\alpha[T^j\carterm^{k_1}\Phi^{k_2}\drout\errhatm]\dee\nu\dee r\lesssim\wbar_1^{\alpha+2}.
\end{align*}
\end{prop}
\begin{proof}
    Recall the equation \eqref{eq:pde} satisfied by $\drout\errhatm$. Using Proposition \ref{prop:Linf}, it rewrites\footnote{We save two angular and $T$ derivatives as in the energy method we only use the bound for the RHS of \eqref{eq:LHS2} in $L^2(S(u,\ubar))$ norm.}
\begin{align}\label{eq:LHS2}
    \teukhat_{-2}^{(1/2,-2)}T^j\carterm^{k_1}\Phi^{k_2}\drout\errhatm=O(\ubar^{\frac{1}{2}(\alpha+1)}).
\end{align}
    This proves that $T^j\carterm^{k_1}\Phi^{k_2}\drout\errhatm$ satisfies \eqref{eq:hypteukIII} with $B=(\alpha+1)/2$, $c=1/2>0$, $V=-2$. Moreover, $T^j\carterm^{k_1}\Phi^{k_2}\drout\errhatm$ satisfies \eqref{eq:hypenerIII} by Propositions \ref{prop:enerestdeuxx} and \ref{prop:enerdeuxx} with $\beta=7+j+\delta$. Thus, by Proposition \ref{prop:enerborneIII} we have
    \begin{align*}
        \iint_{\{\wbar=\wbar_1\}\cap\trois}\ener_\alpha[T^j\carterm^{k_1}\Phi^{k_2}\drout\errhatm]\dee\nu\dee r&\lesssim 1+\int_{\wbar_{\rb}\!(w_{\rb,\gamma})}^{\wbar_1}\wbar^{\alpha+1}\dee\wbar\\
        &\lesssim\wbar_1^{\alpha+2}
\end{align*}
for $\wbar_1\geq 1$ large enough.
\end{proof}
\begin{prop}\label{prop:bornepoldrout}
    Assume that $\Psim$ satisfies \eqref{eq:hyppl}. We have in $\trois$, for $0\leq j\leq N_j-4$ and $0\leq 2k_1+k_2\leq N_k-5$, and $\alpha>1$,
$$|T^j\carterm^{k_1}\Phi^{k_2}\drout\errhatm|\lesssim \ubar^{3/2+\alpha}.$$
\end{prop}
\begin{proof}
We use Lemma \ref{lem:integIII} and Proposition \ref{prop:enerdroutIII} to get
\begin{align*}
    \|T^j\carterm^{k_1}\Phi^{k_2}\drout\errhatm\|_{L^2(S(u,\ubar))}&\lesssim\|T^j\carterm^{k_1}\Phi^{k_2}\drout\errhatm\|_{L^2(S(u_\Gamma(\wbar),\ubar_\Gamma(\wbar)))}\\
    &\quad+\wbar^{\frac{1}{2}(\alpha+1)}\left(\iint_{\{\wbar=\wbar_1\}\cap\trois}\ener_\alpha[T^j\carterm^{k_1}\Phi^{k_2}\drout\errhatm]\dee\nu\dee r\right)^{1/2}\\
    &\lesssim\wbar^{3/2+\alpha}\lesssim\ubar^{3/2+\alpha},
\end{align*}
where we used Proposition \ref{prop:upperboundII} to bound the term on $\Gamma$. As before, we conclude with the Sobolev embedding \eqref{eq:sobolev}.
\end{proof}
We now deduce actual decay from the estimate in Proposition \ref{prop:bornepoldrout} and the identity $\partial_\ubar=(\mu/2)\drout$, as well as the fact that $\mu$ decays exponentially in $\ubar$ in $\trois$. Indeed, recall the identity $-\Delta\sim\exp(-2|\kappa_-|r^*)$ in $\{r_-\leq r\leq\rb\}$. The definition of region $\trois$ then implies
\begin{align}\label{eq:decayexpdelta}
    -\Delta\lesssim\exp(-|\kappa_-|\ubar^\gamma)\text{   in   }\trois.
\end{align}
\begin{prop}\label{prop:almostsharp}
    Assume that $\Psim$ satisfies \eqref{eq:hyppl}. We have in $\trois$, for $0\leq j\leq N_j-4$ and $0\leq 2k_1+k_2\leq N_k-5$,
    $$|T^j\carterm^{k_1}\Phi^{k_2}\errhatm|\lesssim |u|^{-7-j-\delta/2}.$$
\end{prop}
\begin{proof}
    Using Proposition \ref{prop:bornepoldrout} and $\partial_\ubar=e_4=(\mu/2)\drout$, defining $K=3/2+\alpha$, we get in $\trois$:
    $$|\partial_\ubar T^j\carterm^{k_1}\Phi^{k_2}\errhatm|\lesssim -\Delta\ubar^K\lesssim\exp(-|\kappa_-|\ubar^\gamma)\ubar^K.$$
    This yields, integrating from $\Gamma$ on $u=cst$, $\phi_-=cst$, $\theta=cst$: 
    \begin{align}
        |T^j\carterm^{k_1}\Phi^{k_2}\errhatm(u,\ubar,\theta,\phi_-)&-T^j\carterm^{k_1}\Phi^{k_2}\errhatm(u,\ubar_\Gamma(u),\theta,\phi_-)|\label{eq:impliestog}\\
        &\lesssim\int_{\ubar_\Gamma(u)}^{\ubar}\exp(-|\kappa_-|x^\gamma)x^K\dee x\lesssim \int_{\ubar_\Gamma(u)}^{+\infty}\exp(-|\kappa_-|x^\gamma)x^K\dee x.\nn
    \end{align}
Integrating by parts, we have for any $A\geq 1$, 
\begin{align*}
    \int_A^{+\infty}&\exp(-|\kappa_-|x^\gamma)x^K\dee x\\
    &=\frac{-1}{\gamma|\kappa_-|}\int_A^{+\infty}\frac{\dee}{\dee x}[\exp(-|\kappa_-|x^\gamma)]x^{K-\gamma+1}\dee x\\
    &=\frac{-1}{\gamma|\kappa_-|}\Bigg[-\exp(-|\kappa_-|A^\gamma)A^{K-\gamma+1}-(K-\gamma+1)\int_A^{+\infty}\exp(-|\kappa_-|x^\gamma)x^{K-\gamma}\dee x\Bigg].
\end{align*}
As $\gamma>0$, $\exp(-|\kappa_-|x^\gamma)x^{K-\gamma}=o(\exp(-|\kappa_-|x^\gamma)x^{K})$ as $x\to\infty$ thus 
$$\int_A^{+\infty}\exp(-|\kappa_-|x^\gamma)x^{K-\gamma}\dee x=o\left(\int_A^{+\infty}\exp(-|\kappa_-|x^\gamma)x^{K}\dee x\right)$$
as $A\to+\infty$, which proves 
$$\int_A^{+\infty}\exp(-|\kappa_-|x^\gamma)x^K\dee x\underset{A\to+\infty}{\sim}\frac{\exp(-|\kappa_-|A^\gamma)A^{K-\gamma+1}}{\gamma|\kappa_-|}.$$
This implies, together with  \eqref{eq:impliestog}, 
\begin{align*}
    |T^j\carterm^{k_1}\Phi^{k_2}\errhatm(u,\ubar,\theta,\phi_-)&-T^j\carterm^{k_1}\Phi^{k_2}\errhatm(u,\ubar_\Gamma(u),\theta,\phi_-)|\\
    &\lesssim\exp(-|\kappa_-|\ubar_\Gamma(u)^\gamma)\ubar_\Gamma(u)^{K-\gamma+1}\lesssim |u|^{-7-j-\delta/2},
\end{align*}
where we used Lemma \ref{lem:usimubarII} which gives $\ubar_\Gamma(u)\sim|u|$. Combining this bound with Proposition \ref{prop:upperboundII} and Lemma \ref{lem:usimubarII}, which proves that the term on $\Gamma$ is bounded by $|u|^{-7-j-\delta/2}$, concludes the proof.
\end{proof}
\begin{rem}
    The previous proof also implies that $\errhatm$ has a limit on $\ch$.
\end{rem}
The bound in Proposition \ref{prop:almostsharp} is the desired 
goal of this section, for $\errhatm$. More precisely, we will use Proposition \ref{prop:almostsharp} with $j\geq 1$, i.e. the fact that $T$ derivatives of $\errhatm$ decay at least like $|u|^{-8-\delta/2}$. The rest of this section is dedicated to showing that this bound also holds for $T$ derivatives of $\drout\errhatm$. This relies on the integration of a 1+1 wave equation extracted from the Teukolsky equation, similarly as what was done for $\Psihatp$ in \cite{G24}.
\begin{prop}\label{prop:almostsharpdrout}
Choosing $\gamma>0$ small enough, we have in $\trois$, 
    \begin{align*}
        \Psihatm=\anshatm+\errhatm,
    \end{align*}
    where for $0\leq j\leq N_j-5$ and $0\leq 2k_1+k_2\leq N_k-7$,
    $$|T^j\carterm^{k_1}\Phi^{k_2}\drout^{\leq 1}\errhatm|\lesssim|u|^{-7-j-\delta/2}.$$
\end{prop}
\begin{proof}
    The bound without the $\drout$ derivative has been proven in Proposition \ref{prop:almostsharp}, so we prove the bound with the $\drout$. We begin by rewriting the Teukolsky equation $\teukhat_{-2}\errhatm=O(\ubar^{-8})$ as a 1+1 wave equation with a right-hand-side that we will control using Proposition \ref{prop:almostsharp}. Using Propositions \ref{prop:voila} and \ref{prop:teukdec} we get 
\begin{align*}
    \ethat((r^2+a^2)\Delta^{-2}e_4T^j\carterm^{k_1}\Phi^{k_2}\errhatm)=\frac{1}{4}&(-\mu)\Delta^{-2}(\carterm+2aT\Phi+6rT)[T^j\carterm^{k_1}\Phi^{k_2}\errhatm]\\
    &+O(-\Delta^{-1}\ubar^{-7-j-\delta}).
\end{align*}
We can rewrite this as 
\begin{align*}
    \ethat((r^2+a^2)&\Delta^{-2}e_4T^j\carterm^{k_1}\Phi^{k_2}\errhatm)\\
    &=-\frac{\Delta^{-1}}{4(r^2+a^2)}(\carterm+2aT\Phi+6rT)[T^j\carterm^{k_1}\Phi^{k_2}\errhatm]+O(-\Delta^{-1}\ubar^{-7-j-\delta}).
\end{align*}
Thus using $\drout=2\mu^{-1}e_4$ and Propositions \ref{prop:almostsharp} and \ref{prop:upperboundII}, as well as $\ubar\sim|u|$ in $\deux$ and $\ubar^{-7-j-\delta}\lesssim|u|^{-7-j-\delta}$ in $\trois$ by Lemma \ref{lem:usimubarII}, this implies in $\deux\cup\trois$:
\begin{align}\label{eq:1+1aint}
    \ethat(\Delta^{-1}\drout T^j\carterm^{k_1}\Phi^{k_2}\errhatm)=O(-\Delta^{-1}|u|^{-7-j-\delta/2}).
\end{align}
We integrate \eqref{eq:1+1aint} on integral curves of $\ethat$ from $\{r=\rb\}$ to $(u,\ubar,\theta,\phi_-)$ where $\ubar,\theta,\phi_+$ are constant. Using $\ethat=\partial_u$ in these coordinates, we obtain
\begin{align}\label{eq:revenirr}
    \Delta^{-1}(u,\ubar)\drout T^j\carterm^{k_1}\Phi^{k_2}\errhatm&(u,\ubar,\theta,\phi_-)=\nn\\
    &\Delta^{-1}(\rb)\drout T^j\carterm^{k_1}\Phi^{k_2}\errhatm(u_{\rb}(\ubar),\ubar,\theta,(\phi_-)_{\rb}(u,\ubar,\phi_-))\nn\\
    &\quad\quad+\int_{u_{\rb}(\ubar)}^u O(-\Delta^{-1}(u',\ubar)|u'|^{-7-j-\delta/2})\dee u'.
\end{align}
We now estimate the last term on the RHS. We have, on $\ubar=cst$,
$$\partial_u\Delta^{-1}=\ethat\Delta^{-1}=(\mu/2)\partial_r\Delta^{-1}=-\frac{r-M}{r^2+a^2}\Delta^{-1}.$$
This gives 
\begin{align*}
    \int_{u_{\rb}(\ubar)}^u\Delta^{-1}(u',\ubar)|u'|^{-7-j-\delta/2}\dee u'&\sim\int_{u_{\rb}(\ubar)}^u \partial_u\Delta^{-1}(u',\ubar)|u'|^{-7-j-\delta/2}\dee u'\\
    &\sim \Delta^{-1}(u,\ubar)|u|^{-7-j-\delta/2}-\Delta^{-1}(\rb)|u_{\rb}(\ubar)|^{-7-j-\delta/2}\\
    &\quad\quad+(7+j+\delta)\int_{u_{\rb}(\ubar)}^u\Delta^{-1}(u',\ubar)|u'|^{-8-j-\delta/2}\dee u',
\end{align*}
and thus for $|u|$ large enough we get 
$$\int_{u_{\rb}(\ubar)}^uO(-\Delta^{-1}(u',\ubar)|u'|^{-7-j-\delta/2})\dee u'\lesssim -\Delta^{-1}(u,\ubar)|u|^{-7-j-\delta/2}.$$
Going back to \eqref{eq:revenirr} and using \eqref{eq:borneI} to bound the term on $\{r=\rb\}$, this gives in $\trois$, 
$$\drout T^j\carterm^{k_1}\Phi^{k_2}\errhatm=O(-\Delta\ubar^{-7-j-\delta})+O(|u|^{-7-j-\delta/2})=O(|u|^{-7-j-\delta/2}),$$
which concludes the proof.
\end{proof}

\section{Precise asymptotics of $\Psihatm$ near $\ch$}\label{section:precise}

In this section, we begin by finding the precise asymptotics of the $\ell=2$ modes of $\Psihatm$, proving Theorem \ref{thm:precise2l}. Then, we continue this section by proving a $O(|u|^{-8-\delta}+r-r_-)$ upper bound for the $\ell\geq 3$ modes of $\Psihatm$ near $\ch$, separating the analysis between the modes $\ell=3$ and $\ell\geq 4$, which proves Theorems \ref{thm:decayl3} and \ref{thm:decayhigher}. The analysis is done by using the $(\partial_t,\partial_r)$-decomposition \eqref{eq:decteuk} of the Teukolsky operator to integrate simple ODEs in $r$ with source terms that can be controlled using the estimates of Section \ref{section:upperbounds}.

In Sections \ref{section:ccun}, \ref{section:ccdeux}, \ref{section:l=3I}, \ref{section:l=3II}, we write as follows the ingoing spin-weighted spherical harmonics decomposition, for $\ell=2,3$:
\begin{align}\label{eq:harmonicsdec}
    (\errm)_{\ell}=\sum_{|m|\leq \ell}F_{m,\ell}(r,\ubar)Y_{m,\ell}^{-2}(\cos\theta)e^{im\phi_+},
\end{align}
where the functions $F_{m,\ell}(r,\ubar)$ are smooth. Note that this implies, for $j\geq 0$,
\begin{align}\label{eq:Tharmonicsdec}
    (T^j\errm)_{\ell}=\sum_{|m|\leq \ell}\parubarj{F_{m,\ell}}(r,\ubar)Y_{m,\ell}^{-2}(\cos\theta)e^{im\phi_+}.
\end{align}
We also define the functions 
\begin{align}\label{eq:lien}
G_{m,\ell}(r,u):=\Delta^2e^{2imr_{mod}}F_{m,\ell}(r,2r^*-u),
\end{align}
such that we have the outgoing spin-weighted spherical harmonics decomposition
\begin{align}\label{eq:outgoingdec}
    (\errhatm)_{\ell}=\sum_{|m|\leq 2}G_{m,\ell}(r,u)Y_{m,2}^{-2}(\cos\theta)e^{im\phi_-}.
\end{align}
We also define for convenience the functions 
\begin{align}\label{eq:deffm}
    {}^{(j)}\!f_{m,\ell,\ubar}(r):=\parubarj{F_{m,\ell}}(r,\ubar),\quad g_{m,\ell,u}:=G_{m,\ell}(r,u).
\end{align}
\subsection{Precise asymptotics for $(\Psim)_{\ell=2}$ in region $\un$}\label{section:ccun}

We begin the proof of Theorem \ref{thm:precise2l} with the following result.
\begin{prop}\label{prop:odeI}
    Assume that $\Psim$ satisfies \eqref{eq:hyppl}. Then the functions $F_{m,2}(r,\ubar)$ satisfy in $\un$ the following equations, for $|m|\leq 2$:
\begin{align}\label{eq:odeIl}
    \Delta\parrdeux{}\left(\parubarj{F_{m,2}}\right)(r,\ubar)+(2iam+6(r-M))\parr{}\left(\parubarj{F_{m,2}}\right)(r,\ubar)=P_m(r)\parubarj{}\left(\frac{1}{\ubar^8}\right)+O(\ubar^{-8-j-\delta}),
\end{align}
where, recalling the definition of the constants $b_{m,\ell}^s$ in Proposition \ref{prop:modecoupling},
\begin{align}\label{eq:Pm(r)}
    P_m(r):=14Q_{m,2}(5r+iam+2iab_{m,2}^{-2}).
\end{align}
\end{prop}
\begin{rem}
    The notations $\frac{\partial}{\partial r}$ and $\frac{\partial}{\partial \ubar}$ applied to $F_{m,2}(r,\ubar)$ stand respectively for the derivative with respect to the first and second variables of $F_{m,2}(r,\ubar)$, and not to some Boyer-Lindquist or Eddington-Finkelstein coordinates vector fields.
\end{rem}
\begin{proof}[Proof of Proposition \ref{prop:odeI}]

Using \eqref{eq:borneI} we have in $\un$:
\begin{align}\label{eq:prendreteuk}
    \Psim=\ansatzm+\errm,
\end{align}
where, in particular:
\begin{align*}
    &|T^2(T^j\errm)|\lesssim\ubar^{-9-j-\delta},\quad |T(T^j\errm)|\lesssim\ubar^{-8-j-\delta},\\
    &|T\drin(T^j\errm)|\lesssim\ubar^{-8-j-\delta},\quad |T\Phi(T^j\errm)|\lesssim\ubar^{-8-j-\delta}.
\end{align*}

By the expression \eqref{eq:teukdt} of $\teuk_{-2}^{[\partial_t]}$ this implies in $\un$
\begin{align}\label{eq:teukdtborne}
    |\teuk_{-2}^{[\partial_t]}T^j\errm|\lesssim\ubar^{-8-j-\delta}.
\end{align}
Moreover, using \eqref{eq:prendreteuk}, the Teukolsky equation writes
$$\teuk_{-2}T^j\errm=-T^j\teuk_{-2}\left(\ansatzm\right),$$
thus using \eqref{eq:teukdtborne} and the $(\partial_t,\partial_r)$-decomposition \eqref{eq:decteuk} we get 
\begin{align}\label{eq:cabahaproj}
    \teuk^{[\partial_r]}_{-2}T^j\errm=-T^j\teuk_{-2}\left(\ansatzm\right)+O(\ubar^{-8-j-\delta}).
\end{align}
Now we compute using \eqref{eq:teukdr}, \eqref{eq:teukdt}, as well as $$\drin\ubar=\drin\phi_+=\drin\theta=\drond\drond'(Y_{m,2}^{-2}(\cos\theta)e^{im\phi_+})=0,$$
the source term:
\begin{align}
    -T^j\teuk_{-2}&\left(\ansatzm\right)=\nn\\
    &T^j\left(\frac{14}{\ubar^8}\right)\sum_{|m|\leq 2}Q_{m,2}(5r+2ia\cos\theta+iam)Y_{m,2}^{-2}(\cos\theta)e^{im\phi_+}+O(\ubar^{-9-j}).\label{eq:unedessources}
\end{align}
Using Lemma \ref{lem:commproj} and \eqref{eq:projcostheta}, projecting \eqref{eq:cabahaproj} onto the ingoing $\ell=2$ mode thus gives, in view of \eqref{eq:Pm(r)},
\begin{align}\label{eq:aprojl}
    \teuk_{-2}^{[\partial_r]}T^j(\errm)_{\ell=2}=T^j\left(\frac{1}{\ubar^8}\right)\sum_{|m|\leq 2}P_m(r)Y_{m,2}^{-2}(\cos\theta)e^{im\phi_+}+O(\ubar^{-8-j-\delta}).
\end{align}
Then using the expression \eqref{eq:teukdr} of the operator $\teuk_{-2}^{[\partial_r]}$, as well as \eqref{eq:annulemode}, \eqref{eq:aprojl} rewrites
\begin{align*}
    \sum_{|m|\leq 2}\Big[\Delta\parrdeux{}\left(\parubarj{F_{m,2}}\right)(r,\ubar)+(2iam+&6(r-M))\parr{}\left(\parubarj{F_{m,2}}\right)(r,\ubar)\Big]Y_{m,2}^{-2}(\cos\theta)e^{im\phi_+}\nonumber\\
    &=\parubarj{}\left(\frac{1}{\ubar^8}\right)\sum_{|m|\leq 2}P_m(r)Y_{m,2}^{-2}(\cos\theta)e^{im\phi_+}+O(\ubar^{-8-\delta})
\end{align*}
which, after projecting onto each $(\ell=2,m)$ mode, concludes the proof of the proposition.
\end{proof}
\subsubsection{Resolution of the ODE \eqref{eq:odeIl}}
In the setting of Proposition \ref{prop:odeI}, let us fix $\ubar\geq 1$, $|m|\leq 2$, $j\geq 0$. By Proposition \ref{prop:odeI}, the ODE satisfied by ${}^{(j)}\!f_{m,2,\ubar}$ (which is defined in \eqref{eq:deffm}) is 
\begin{align}\label{eq:odeodel}
    \Delta {}^{(j)}\!f_{m,2,\ubar}''(r)+(2iam+6(r-M)){}^{(j)}\!f_{m,2,\ubar}'(r)=S_m^{(j)}(r,\ubar),
\end{align}
where the source $S_m^{(j)}(r,\ubar)$ satisfies 
\begin{align}\label{eq:sourceIl}
    S_m^{(j)}(r,\ubar)=P_m(r)\parubarj{}\left(\frac{1}{\ubar^8}\right)+O(\ubar^{-8-j-\delta}).
\end{align}
We will determine ${}^{(j)}\!f_{m,2,\ubar}$ using the variation of constants method. To this end, we look for a basis of solutions to the homogeneous problem 
\begin{align}\label{eq:odeHl}
    \Delta v''(r)+(2iam+6(r-M))v'(r)=0.
\end{align}
An obvious solution to \eqref{eq:odeHl} is 
\begin{align}\label{eq:v1def}
    {}^{(2)}\!v_1(r):=1.
\end{align}
Another solution ${}^{(2)}\!v_2$ satisfies 
$${}^{(2)}\!v_2''=-\frac{2iam+6(r-M)}{\Delta}{}^{(2)}\!v_2'$$
thus up to a multiplicative constant, 
$${}^{(2)}\!v_2'=\exp\left(\int^r-\frac{2iam+6(r'-M)}{\Delta(r')}\dee r'\right)=\Delta(r)^{-3}e^{-2imr_{mod}},$$
so that $({}^{(2)}\!v_1,{}^{(2)}\!v_2)$ is a basis of solutions to \eqref{eq:odeHl}, with
\begin{align}
    {}^{(2)}\!v_2(r):=\int_{\rb}^r\Delta(r')^{-3}e^{-2imr_{mod}(r')}\dee r'.\label{eq:v2def}
\end{align}
We now state a result on the behavior of ${}^{(2)}\!v_2(r)$ near $r=r_\pm$.
\begin{lem}\label{lem:v2dvl}
    Let ${}^{(2)}\!v_2(r)$ be given by \eqref{eq:v2def}. We have the following asymptotic behaviors: 
    \begin{itemize}
        \item For $\rb\leq r\leq r_+$:
    $${}^{(2)}\!v_2(r)=\frac{-\Delta^{-2}(r)e^{-2imr_{mod}}}{4(r_+-M)+2iam}+O((r_+-r)^{-1}).$$
    \item For $r_-\leq r\leq\rb$:
    $${}^{(2)}\!v_2(r)=\frac{-\Delta^{-2}(r)e^{-2imr_{mod}}}{4(r_--M)+2iam}+O((r-r_-)^{-1}).$$
    \end{itemize}
\end{lem}
\begin{proof}
    We begin with the first statement with the case $m\neq 0$. Recalling
    $$\Delta'(r)=2(r-M),\quad r_{mod}'(r)=\frac{a}{\Delta},$$
    we have the following integration by parts
\begin{align*}
    \int_{\rb}^r\Delta^{-3}(r')e^{-2imr_{mod}(r')}\dee r'&=-\frac{1}{2iam}\Bigg(\Bigg[\Delta^{-2}e^{-2imr_{mod}}\Bigg]_{\rb}^r+\int_{\rb}^r4(r'-M)\Delta^{-3}(r')e^{-2imr_{mod}(r')}\dee r'\Bigg)\\
    &=\frac{-\Delta^{-2}(r)e^{-2imr_{mod}}}{2iam}-\frac{2(r-M)}{iam}\int_{\rb}^r\Delta^{-3}(r')e^{-2imr_{mod}(r')}\dee r'\\
    &\quad+O((r_+-r)^{-1}).
\end{align*}
This yields, as stated,
$$\int_{\rb}^r\Delta^{-3}(r')e^{-2imr_{mod}(r')}\dee r'=\frac{-\Delta^{-2}(r)e^{-2imr_{mod}}}{4(r_+-M)+2iam}+O((r_+-r)^{-1}).$$
The case $m=0$ can be treated in an easier way, using $\Delta^{-3}\sim (r-r_+)^{-3}$ in $\un$ and directly integrating. The second statement is proven by exchanging the roles of $r_+$ and $r_-$.
\end{proof}
Using the basis $({}^{(2)}\!v_1,{}^{(2)}\!v_2)$ of solutions of the homogeneous problem \eqref{eq:odeHl}, we now use the variation of constants method to find a particular solution ${}^{(j)}\!p_{m,2,\ubar}$ of \eqref{eq:odeodel} as 
$${}^{(j)}\!p_{m,2,\ubar}(r)={}^{(j)}\!M_{m,2,\ubar}(r){}^{(2)}\!v_1(r)+{}^{(j)}\!N_{m,2,\ubar}(r){}^{(2)}\!v_2(r),$$
where ${}^{(j)}\!M_{m,2,\ubar}(r),{}^{(j)}\!N_{m,2,\ubar}(r)$ are functions such that
\begin{align}\label{eq:systemIl}
    \begin{cases}
        {}^{(j)}\!M_{m,2,\ubar}'{}^{(2)}\!v_1+{}^{(j)}\!N_{m,2,\ubar}'{}^{(2)}\!v_2=0,\\
        {}^{(j)}\!M_{m,2,\ubar}'{}^{(2)}\!v_1'+{}^{(j)}\!N_{m,2,\ubar}'{}^{(2)}\!v_2'=\Delta^{-1}S_m^{(j)}(r,\ubar).
    \end{cases}
\end{align}
Solving \eqref{eq:systemIl} gives, in view of \eqref{eq:v1def}, \eqref{eq:v2def},
$${}^{(j)}\!N_{m,2,\ubar}'(r)=\Delta^2S_m^{(j)}(r,\ubar)e^{2imr_{mod}},\quad {}^{(j)}\!M_{m,2,\ubar}'(r)=-\Delta^2S_m^{(j)}(r,\ubar)e^{2imr_{mod}}{}^{(2)}\!v_2(r).$$
Using Lemma \ref{lem:v2dvl} we get that ${}^{(j)}\!M_{m,2,\ubar}'(r)$ converges as $r\to r_+$. Thus our particular solution of \eqref{eq:odeodel} can be chosen as
$${}^{(j)}\!p_{m,2,\ubar}(r)=-\int_{r_+}^r\Delta^2(r'){}^{(2)}\!v_2(r')S_m^{(j)}(r',\ubar)e^{2imr_{mod}(r')}\dee r'+{}^{(2)}\!v_2(r)\int_{r_+}^r\Delta^2(r')S_m^{(j)}(r',\ubar)e^{2imr_{mod}(r')}\dee r'.$$
Using the definition of ${}^{(2)}\!v_2(r)$ we get more precisely
\begin{align}\label{eq:luilal}
    {}^{(j)}\!p_{m,2,\ubar}(r)=\int_{r_+}^r\Delta^2(r')e^{2imr_{mod}(r')}S_m^{(j)}(r',\ubar)\left(\int_{r'}^r\Delta^{-3}(r'')e^{-2imr_{mod}(r'')}\dee r''\right)\dee r'.
\end{align}
Now, since ${}^{(j)}\!f_{m,2,\ubar}-{}^{(j)}\!p_{m,2,\ubar}$ satisfies the homogeneous equation \eqref{eq:odeHl}, there are constants $A_{m,2}^{(j)}(\ubar)$ and $B_{m,2}^{(j)}(\ubar)$ that depend only on $\ubar$, $m$ and $j$, such that for any $r\in[\rb,r_+]$, $\ubar\geq 1$, $|m|\leq 2$, and $j\geq 0$,
\begin{align}\label{eq:FMresolul}
    {}^{(j)}\!f_{m,2,\ubar}(r)=A_{m,2}^{(j)}(\ubar)+B_{m,2}^{(j)}(\ubar){}^{(2)}\!v_2(r)+{}^{(j)}\!p_{m,2,\ubar}(r).
\end{align}
\subsubsection{Computation of $A_{m,2}^{(j)}(\ubar)$, $B_{m,2}^{(j)}(\ubar)$}
Notice that ${}^{(j)}\!p_{m,2,\ubar}(r_+)=0$ and using Lemma \ref{lem:v2dvl} we have that ${}^{(2)}\!v_2$ diverges as $r\to r_+$, thus using the smoothness of $\Psim$ at $\mch_+$ yields necessarily $$B_{m,2}^{(j)}(\ubar)=0.$$
Next, evaluating \eqref{eq:FMresolul} at $r=r_+$ and using \eqref{eq:deffm} and the initial assumption \eqref{eq:hyppl} on $\mch_+$ implies
$$A_{m,2}^{(j)}(\ubar)={\parubarj{F_{m,2}}(r_+,\ubar)}=O(\ubar^{-7-j-\delta}).$$
Putting back into \eqref{eq:FMresolul} the expression \eqref{eq:sourceIl} of the source $S_m^{(j)}(r,\ubar)$, we find in $\un$
\begin{align}
    \parubarj{F_{m,2}}(r,\ubar)&=\parubarj{}\left(\frac{1}{\ubar^8}\right)\int_{r_+}^r\Delta^2(r')e^{2imr_{mod}(r')}P_m(r')\left(\int_{r'}^r\Delta^{-3}(r'')e^{-2imr_{mod}(r'')}\dee r''\right)\dee r'\nonumber\\
    &\quad+O(\ubar^{-8-j-\delta})+O(\ubar^{-7-j-\delta})\nn\\
    &=O(\ubar^{-7-j-\delta}).\label{eq:asympto}
\end{align}
\begin{rem}
    Recall that we expect an asymptotic in $u^{-8}$ for $\Psihatm$ near $\ch$. The asymptotic behavior \eqref{eq:asympto} in $\un$ might seem too imprecise at first, but this imprecision will actually be canceled when passing from $\Psim$ close to $\mch_+$ to $\Psihatm$ close to $\ch$. What will really influence the final asymptotic at $\ch$ is actually the precise asymptotic in $\ubar^{-8}$ of the radial derivative $\partial_rF_{m,2}(r,\ubar)$ in $\un$, that we obtain below.
\end{rem}
 Note that \eqref{eq:FMresolul}, together with $B_{m,2}^{(j)}(\ubar)=0$ and \eqref{eq:luilal}, yields  
\begin{align*}
    \parubarj{F_{m,2}}(r,\ubar)=A_{m,2}^{(j)}(\ubar)+\int_{r_+}^r\Delta^2(r')e^{2imr_{mod}(r')}S_m^{(j)}(r',\ubar)\left(\int_{r'}^r\Delta^{-3}(r'')e^{-2imr_{mod}(r'')}\dee r''\right)\dee r',
\end{align*}
thus we have in $\un$
\begin{align*}
    \parr{}\left(\parubarj{F_{m,2}}\right)(r,\ubar)=\frac{1}{\Delta^3(r)e^{2imr_{mod}}}\int_{r_+}^r\Delta^2(r')e^{2imr_{mod}(r')}S_m^{(j)}(r',\ubar)\dee r'.
\end{align*}
Re-injecting the definition \eqref{eq:sourceIl} of $S_m^{(j)}(r',\ubar)$ gives in $\un$ the following precise asymptotic behavior
\begin{align}\label{eq:drbienla}
    \parr{}\left(\parubarj{F_{m,2}}\right)(r,\ubar)=\frac{1}{\Delta^3(r)e^{2imr_{mod}}}\parubarj{}\left(\frac{1}{\ubar^8}\right)\int_{r_+}^r\Delta^2(r')e^{2imr_{mod}(r')}P_m(r')\dee r'+O(\ubar^{-8-j-\delta}).
\end{align}
We now propagate these asymptotics to $\deux\cup\trois=\{r_-\leq r\leq \rb\}\cap\{w\leq w_{\rb,\gamma}\}$ by a similar method. 
\subsection{Precise asymptotics of $(\Psihatm)_{\ell=2}$ in region $\deux\cup\trois$}\label{section:ccdeux}
We follow in region $\deux\cup\trois=\{r_-\leq r\leq \rb\}\cap\{w\leq w_{\rb,\gamma}\}$ a procedure analogous to the one in region $\un$ of Section \ref{section:ccun}. 
\begin{prop}\label{prop:odeII} Assume that $\Psim$ satisfies \eqref{eq:hyppl}. Then the functions $G_{m,2}(r,u)$, defined in \eqref{eq:outgoingdec}, satisfy in $\deux\cup\trois$ the following equations, for $|m|\leq 2$: 
    \begin{align}\label{eq:odeIIl}
    \Delta\parrdeux{G_{m,2}}(r,u)-(2iam+2(r-M))\parr{G_{m,2}}(r,u)-4G_{m,2}(r,u)=\frac{\Delta^2e^{2imr_{mod}}P_m(r)}{(2r^*-u)^8}+O(|u|^{-8-\delta/2}).
\end{align}
\end{prop}
\begin{proof}[Proof of Proposition \ref{prop:odeII}]
Recall that in $\{r_-\leq r\leq \rb\}\cap\{w\leq w_{\rb,\gamma}\}$ we have 
$$\Psihatm=\anshatm+\errhatm,$$
where $|T^j\carterm^{k_1}\Phi^{k_2}\drout^{\leq 1}\errhatm|\lesssim|u|^{-7-j-\delta/2}$, see Propositions \ref{prop:upperboundII} and \ref{prop:almostsharpdrout}, as well as Lemma \ref{lem:usimubarII}. Using \eqref{eq:teukhatdt}, this gives the estimate in region $\{r_-\leq r\leq \rb\}\cap\{w\leq w_{\rb,\gamma}\}$: 
\begin{align}\label{eq:premiereIIl}
    |\teukhat_{-2}^{[\partial_t]}\errhatm|\lesssim|u|^{-8-\delta/2}.
\end{align}
As before, the Teukolsky equation writes, using the $(\partial_t,\partial_r)$ decomposition of $\teukhat_{-2}$ and \eqref{eq:premiereIIl}, 
$$\teukhat_{-2}^{[\partial_r]}\errhatm=-\teukhat_{-2}\left(\anshatm\right)+O(|u|^{-8-\delta/2}).$$
Using the identity $\teukhat_{-2}(\Delta^2\:\cdot\:)=\Delta^2\teuk_{-2}(\:\cdot\:)$, see \eqref{eq:rescaledequation}, and \eqref{eq:unedessources}, we obtain 
\begin{align}
    \teukhat_{-2}^{[\partial_r]}\errhatm=\frac{14\Delta^2}{\ubar^8}\sum_{|m|\leq 2}Q_{m,2}(5r+2ia\cos\theta+iam)Y_{m,2}^{-2}(\cos\theta)e^{im\phi_+}+O(|u|^{-8-\delta/2}).
\end{align}
Just like in \eqref{eq:aprojl}, projecting onto the outgoing $\ell=2$ mode thus gives, using $\phi_+=\phi_-+2r_{mod}$,
\begin{align}\label{eq:aprojIIl}
    \teukhat_{-2}^{[\partial_r]}(\errhatm)_{\ell=2}=\frac{\Delta^2}{\ubar^8}\sum_{|m|\leq 2}P_m(r)e^{2imr_{mod}}Y_{m,2}^{-2}(\cos\theta)e^{im\phi_-}+O(|u|^{-8-\delta/2}),
\end{align}
where we used Lemma \ref{lem:commproj}, \eqref{eq:projcostheta}, and \eqref{eq:Pm(r)}. Then using the expression \eqref{eq:teukhatdr} of the operator $\teukhat_{-2}^{[\partial_r]}$, as well as \eqref{eq:dronddrond'spin}, \eqref{eq:aprojIIl} rewrites
\begin{align*}
    \sum_{|m|\leq 2}\Big[\Delta\parrdeux{G_{m,2}}(r,u)-(2iam+&2(r-M))\parr{G_{m,2}}(r,u)-4G_{m,2}(r,u)\Big]Y_{m,2}^{-2}(\cos\theta)e^{im\phi_-}\nonumber\\
    &=\frac{\Delta^2}{\ubar^8}\sum_{|m|\leq 2}P_m(r)e^{2imr_{mod}}Y_{m,2}^{-2}(\cos\theta)e^{im\phi_-}+O(|u|^{-8-\delta/2}),
\end{align*}
which concludes the proof after projecting on each $(\ell=2,m)$ mode and using $\ubar=2r^*-u$.
\end{proof}
\subsubsection{Resolution of the ODE \eqref{eq:odeIIl}}
In the setting of Proposition \ref{prop:odeII}, let us consider values of $u\ll -1$ such that the corresponding level sets are in $\{r_-\leq r\leq \rb\}\cap\{w\leq w_{\rb,\gamma}\}$. Recall the definition of $G_{m,2}$ in \eqref{eq:deffm}, which satisfies the ODE 
\begin{align}\label{eq:odeodeIIl}
    \Delta G_{m,2}''(r)-(2iam+2(r-M))G_{m,2}'(r)-4G_{m,2}(r)=\widehat{S}_m(r,u),
\end{align}
where the source $\widehat{S}_m(r,u)$ satisfies 
\begin{align}\label{eq:sourceIIl}
    \widehat{S}_m(r,u)=\frac{\Delta^2e^{2imr_{mod}}P_m(r)}{(2r^*-u)^8}+O(|u|^{-8-\delta/2}).
\end{align}
As before, we will determine $G_{m,2}$ using the variation of constants method. To this end, we look for a basis of solutions to the homogeneous problem 
\begin{align}\label{eq:odeHIIl}
    \Delta \widehat{v}''(r)-(2iam+2(r-M))\widehat{v}'(r)-4\widehat{v}(r)=0,
\end{align}
using the variation of constants method, and the initial data on $\{r=\rb\}$. To solve \eqref{eq:odeHIIl}, notice that if $v$ satisfies \eqref{eq:odeHl} then $\widehat{v}:=v\Delta^2e^{2imr_{mod}}$ satisfies \eqref{eq:odeHIIl}. Indeed, under these assumptions we have 
\begin{align*}
    \widehat{v}'&=[4(r-M)\Delta v+\Delta^2v'+2iam\Delta v]e^{2imr_{mod}}\\
    &=[\Delta(4(r-M)+2iam)v+\Delta^2v']e^{2imr_{mod}},
\end{align*}
and 
\begin{align*}
    \widehat{v}''&=[2(r-M)(4(r-M)+2iam)v+4\Delta v+\Delta(4(r-M)+2iam)v'+4(r-M)\Delta v'+\Delta^2v''\\
    &\quad\quad +2iam(4(r-M)+2iam)v+2iam\Delta v']e^{2imr_{mod}},\\
    &=[2(r-M)+2iam]\Delta v'e^{2imr_{mod}}+[(2(r-M)+2iam)(4(r-M)+2iam)+4\Delta]ve^{2imr_{mod}}
\end{align*}
by \eqref{eq:odeHl}, which gives as stated
\begin{align*}
    \Delta\widehat{v}''=(2(r-M)+2iam)\widehat{v}'+4\Delta^2ve^{2imr_{mod}}=(2(r-M)+2iam)\widehat{v}'+4\widehat{v}.
\end{align*}
Recalling the basis $({}^{(2)}v_1, {}^{(2)}v_2)$ of solutions of \eqref{eq:odeHl}, this yields a basis of solutions of \eqref{eq:odeHIIl} given by
\begin{align}
    &{}^{(2)}\!\widehat{v_1}(r):=\Delta^2(r)e^{2imr_{mod}(r)},\\
    &{}^{(2)}\!\widehat{v_2}(r):=\Delta^2(r)e^{2imr_{mod}(r)}\int_{\rb}^r\Delta^{-3}(r')e^{-2imr_{mod}(r')}\dee r'.
\end{align}
We now use the variation of constants method to find a particular solution 
$$\widehat{p}_{m,2,u}(r):=\widehat{M}_{m,2,u}(r){}^{(2)}\!\widehat{v_1}(r)+\widehat{N}_{m,2,u}(r){}^{(2)}\!\widehat{v_1}(r)$$
of \eqref{eq:odeodeIIl}, where $\widehat{M}_{m,2,u}(r)$ and $\widehat{N}_{m,2,u}(r)$ are functions such that
\begin{align}\label{eq:systemIIl}
    \begin{cases}
        \widehat{M}_{m,2,u}'{}^{(2)}\!\widehat{v_1}+\widehat{N}_{m,2,u}'{}^{(2)}\!\widehat{v_2}=0,\\
        \widehat{M}_{m,2,u}'{}^{(2)}\!\widehat{v_1}'+\widehat{N}_{m,2,u}'{}^{(2)}\!\widehat{v_2}'={\Delta^{-1}}{\widehat{S}_m(r,u)}.
    \end{cases}
\end{align}
To solve \eqref{eq:systemIIl}, let us compute the Wronskian $\widehat{W}_2$ of $({}^{(2)}\!\widehat{v_1}, {}^{(2)}\!\widehat{v_2})$. We have the identity
$${}^{(2)}\!\widehat{v_2}={}^{(2)}\!\widehat{v_1}\int_{\rb}^r\frac{\Delta(r')e^{2imr_{mod}(r')}}{{}^{(2)}\!\widehat{v_1}(r')^2}\dee r',$$
thus the Wronskian is 
$$\widehat{W}_2={}^{(2)}\!\widehat{v_1}'{}^{(2)}\!\widehat{v_2}-{}^{(2)}\!\widehat{v_2}'{}^{(2)}\!\widehat{v_1}=-\Delta e^{2imr_{mod}}.$$
This gives $\widehat{M}_{m,2,u}'\widehat{W}_2=\widehat{S}_m(r,u){}^{(2)}\!\widehat{v_2}/\Delta$ and $\widehat{N}_{m,2,u}'\widehat{W}_2=-\widehat{S}_m(r,u){}^{(2)}\!\widehat{v_1}/\Delta$, hence 
$$\widehat{M}_{m,2,u}'(r)=\frac{- {}^{(2)}\!\widehat{v_2}(r)\widehat{S}_m(r,u)}{\Delta^2e^{2imr_{mod}}},\quad \widehat{N}_{m,2,u}'(r)=\frac{{}^{(2)}\!\widehat{v_1}(r)\widehat{S}_m(r,u)}{\Delta^2e^{2imr_{mod}}}.$$
Re-injecting the definitions of ${}^{(2)}\!\widehat{v_2}(r)$ and ${}^{(2)}\!\widehat{v_1}(r)$, this gives 
\begin{align*}
    &\widehat{M}_{m,2,u}'(r)=-{}^{(2)}\!v_2(r)\widehat{S}_m(r,u)=-\widehat{S}_m(r,u)\int_{\rb}^r\Delta^{-3}(r')e^{-2imr_{mod}(r')}\dee r',\\
    &\widehat{N}_{m,2,u}'(r)=\widehat{S}_m(r,u).
\end{align*}
Our particular solution can thus be chosen as
\begin{align*}
    \widehat{p}_{m,2,u}(r)=-{}^{(2)}\!\widehat{v_1}(r)\int_{\rb}^r\widehat{S}_m(r',u)\left(\int_{\rb}^{r'}\Delta^{-3}(r'')e^{-2imr_{mod}(r'')}\dee r''\right)\dee r'+{}^{(2)}\!\widehat{v_2}(r)\int_{\rb}^{r}\widehat{S}_m(r',u)\dee r'.
\end{align*}
We conclude that there are constants $\widehat{A}_{m,2}(u)$, $\widehat{B}_{m,2}(u)$ such that in $\{r_-\leq r\leq \rb\}\cap\{w\leq w_{\rb,\gamma}\}$,
\begin{align}
    G_{m,2}(r,u)=&-{}^{(2)}\!\widehat{v_1}(r)\int_{\rb}^r\widehat{S}_m(r',u)\left(\int_{\rb}^{r'}\Delta^{-3}(r'')e^{-2imr_{mod}(r'')}\dee r''\right)\dee r'+{}^{(2)}\!\widehat{v_2}(r)\int_{\rb}^{r}\widehat{S}_m(r',u)\dee r'\nn\\
    &+\widehat{A}_{m,2}(u){}^{(2)}\!\widehat{v_1}(r)+\widehat{B}_{m,2}(u){}^{(2)}\!\widehat{v_2}(r).\label{eq:troizl}
\end{align}
\subsubsection{Computation of the constants $\widehat{A}_{m,2}(u)$, $\widehat{B}_{m,2}(u)$}
Evaluating \eqref{eq:troizl} on $r=\rb$ and using \eqref{eq:lien}, we get
$$\widehat{A}_{m,2}(u){}^{(2)}\!\widehat{v_1}(\rb)=G_{m,2}(\rb,u)=\Delta^2(\rb)e^{2imr_{mod}(\rb)}F_{m,2}(\rb,\ubar_{\rb}(u)),$$
where $\ubar_{\rb}(u)=2\rb^*-u$, which yields 
\begin{align}\label{eq:permanul}
    \widehat{A}_{m,2}(u)=F_{m,2}(\rb,\ubar_{\rb}(u))=O(|u|^{-7-\delta}).
\end{align}
We also have
\begin{align*}
    \parr{G_{m,2}}(\rb,u)&=\parr{}\:\Big|_{r=\rb}(\Delta^2e^{2imr_{mod}}F_{m,2}(r,2r^*-u))\\
    &=(4(\rb-M)+2iam)\Delta(\rb)e^{2imr_{mod}(\rb)}F_{m,2}(\rb,\ubar_{\rb}(u))\\
    &\quad+2(\rb^2+a^2)\Delta(\rb)e^{2imr_{mod}(\rb)}\parubar{F_{m,2}}(\rb,\ubar_{\rb}(u))+\Delta^2(\rb)e^{2imr_{mod}(\rb)}\parr{F_{m,2}}(\rb,\ubar_{\rb}(u)).
\end{align*}
Moreover, by \eqref{eq:troizl} and using ${}^{(2)}\!\widehat{v_2}(\rb)=0$, we obtain $$\parr{G_{m,2}}(\rb,u)=\widehat{A}_{m,2}(u){}^{(2)}\!\widehat{v_1}'(\rb)+\widehat{B}_{m,2}(u){}^{(2)}\!\widehat{v_2}'(\rb).$$
Thus, using \eqref{eq:permanul}, \eqref{eq:asympto} with $j=1$, as well as ${}^{(2)}\!\widehat{v_1}'(\rb)=(4(\rb-M)+2iam)\Delta(\rb)e^{2imr_{mod}(\rb)}$, yields
\begin{align*}
    \widehat{B}_{m,2}(u){}^{(2)}\!\widehat{v_2}'(\rb)=\Delta^2(\rb)e^{2imr_{mod}(\rb)}\parr{F_{m,2}}(\rb,\ubar_{\rb}(u))+O(|u|^{-8-\delta/2}).
\end{align*}
Using ${}^{(2)}\!\widehat{v_2}'(\rb)=\Delta^{-1}(\rb)$, as well as \eqref{eq:drbienla} we get 
\begin{align*}
    \widehat{B}_{m,2}(u)=\frac{1}{\ubar_{\rb}(u)^8}\int_{r_+}^{\rb}\Delta^2(r')e^{2imr_{mod}(r')}P_m(r')\dee r'+O(\ubar_{\rb}(u)^{-8-\delta/2}).
\end{align*}
Recalling $\ubar_{\rb}(u)=2\rb^*-u$, we obtain
\begin{align}\label{eq:K22(u)}
    \widehat{B}_{m,2}(u)=\frac{1}{u^8}\int_{r_+}^{\rb}\Delta^2(r')e^{2imr_{mod}(r')}P_m(r')\dee r'+O(|u|^{-8-\delta/2}).
\end{align}
\subsubsection{Precise asymptotic near $\ch$}
We now turn to estimating the terms in the RHS of \eqref{eq:troizl} approaching $\ch$. A simple double integration bound gives that the first term of the first line of the RHS of \eqref{eq:troizl} is bounded by $-\Delta(r)^{-1}|{}^{(2)}\!\widehat{v_1}(r)|\lesssim r-r_-$. We also have $|\widehat{A}_{m,2}(u){}^{(2)}\!\widehat{v_1}(r)|\lesssim (r-r_-)^2$ by \eqref{eq:permanul}. For the second term on the first line of the RHS of \eqref{eq:troizl}, we first write using \eqref{eq:sourceIIl},
\begin{align*}
    \Bigg|\int_{\rb}^{r}\widehat{S}_m(r',u)\dee r'&-\frac{1}{u^8}\int_{\rb}^r\Delta^2(r')e^{2imr_{mod}(r')}P_m(r')\dee r'\Bigg|\\
    &\lesssim\left|\int_{\rb}^r\left(\frac{1}{(u-2(r')^*)^8}-\frac{1}{u^8}\right)\Delta^2(r')e^{2imr_{mod}(r')}P_m(r')\dee r'\right|+|u|^{-8-\delta/2}\\
    &\lesssim \int_r^{\rb} \frac{2(r')^*\Delta(r')^2}{|u|^9}\dee r'+|u|^{-8-\delta/2}\\
    &\lesssim |u|^{-8-\delta/2}.
\end{align*}
Using this, as well as the value \eqref{eq:K22(u)} of $\widehat{B}_{m,2}(u)$ gives
\begin{align}
    G_{m,2}(r,u)=\frac{{}^{(2)}\!\widehat{v_2}(r)}{u^8}\int_{r_+}^r\Delta^2(r')e^{2imr_{mod}(r')}P_m(r')\dee r'+O(|u|^{-8-\delta/2})+O(r-r_-).
\end{align}
Next, notice that ${}^{(2)}\!\widehat{v_2}(r)$ converges on $r=r_-$. Indeed, using Lemma \ref{lem:v2dvl} we get, as $r\to r_-$, 
$$\int_{\rb}^r\Delta^{-3}(r')e^{-2imr_{mod}(r')}\dee r'=\frac{-\Delta^{-2}(r)e^{-2imr_{mod}}}{4(r_--M)+2iam}+O((r-r_-)^{-1}),$$
which proves 
$${}^{(2)}\!\widehat{v_2}(r_-)=-\frac{1}{4(r_--M)+2iam},$$
as well as $|{}^{(2)}\!\widehat{v_2}(r)-{}^{(2)}\!\widehat{v_2}(r_-)|\lesssim r-r_-$. Combining this with the definition \eqref{eq:Pm(r)} of $P_m(r)$ finally gives, for $r\in[r_-,\rb]$ and $w\leq w_{\rb,\gamma}$,
\begin{align*}
    G_{m,2}(r,u)=\frac{1}{u^8}&\frac{7Q_{m,2}}{2(r_--M)+iam}\int_{r_-}^{r_+}\Delta^2(r')e^{2imr_{mod}(r')}(5r'+iam+2iab_{m,2}^{-2})\dee r'\\
    &+O(|u|^{-8-\delta/2})+O(r-r_-).
\end{align*}
Finally, combining this with the identity
$$(\Psihatm)_{\ell=2}=\anshatm+(\errhatm)_{\ell=2}=\sum_{|m|\leq 2}G_{m,2}(r,u)Y_{m,2}^{-2}(\cos\theta)e^{im\phi_-}+O(r-r_-)$$
concludes the proof of Theorem \ref{thm:precise2l}.
\begin{rem}\label{rem:importante}
    The $O(\ubar^{-7-\delta})$ imprecise upper bound for $\errm$ in $\un$ does not appear anymore in the asymptotics of $\errhatm$ near $\ch$, because it gets canceled by factors $r-r_-$ when rescaling from $\Psim$ to $\Psihatm$. This allows us to deduce a precise $u^{-8}$ asymptotic behavior for $\errhatm$ near $\ch$ \textit{even though the precise asymptotic behavior of $\errm$ is unknown on $\eh$.}
\end{rem}
\subsection{Upper bound for $(\Psim)_{\ell=3}$ in region $\un$}\label{section:l=3I}
We now prove Theorem \ref{thm:decayl3}. We proceed exactly like for the modes $\ell=2$. The only difference will be the vanishing of the coefficient in front of $1/u^8$ in the leading-order term of $(\Psihatm)_{\ell=3}$. Recall the definition \eqref{eq:harmonicsdec} of the functions $F_{m,3}(r,\ubar)$.
\begin{prop}\label{prop:edoI3}
    Assume that $\Psim$ satisfies \eqref{eq:hyppl}. Then, defining by convention $Q_{\pm 3,2}=0$, the functions ${F_{m,3}}(r,\ubar)$ satisfy in $\un$ the equations, for $|m|\leq 3$,
    \begin{align}\label{eq:odeI}
    \Delta\parrdeux{}\left(\parubarj{F_{m,3}}\right)(r,\ubar)+(2iam+6(r-M))\parr{}&\left(\parubarj{F_{m,3}}\right)(r,\ubar)-6\parubarj{F_{m,3}}(r,\ubar)\nn\\
    &=28ia\parubarj{}\left(\frac{1}{\ubar^8}\right)b_{m,2}^{-2}Q_{m,2}+O(\ubar^{-8-j-\delta}).
\end{align}
\end{prop}
\begin{proof}
Using \eqref{eq:borneI} we have in $\un$, $|T^j\carterm^{k_1}\Phi^{k_2}\drin^{\leq 1}\errm|\lesssim\ubar^{-7-j-\delta}$. As $(\Psim)_{\ell\geq 3}=(\errm)_{\ell\geq 3}$, we deduce the bounds:
\begin{align}\label{eq:premiere}
    &|(\sin^2\theta T^{j+2}\Psim)_{\ell=3}|\lesssim\ubar^{-9-j},\\
    &|(T^{j+1}\Psim)_{\ell= 3}|=|(T^{j+1}\errm)_{\ell= 3}|\lesssim\ubar^{-8-j-\delta},\\
    &|(T^{j+1}\drin\Psim)_{\ell= 3}|=|(T^{j+1}\drin\errm)_{\ell= 3}|\lesssim\ubar^{-8-j-\delta},\\
    &|(T^{j+1}\Phi\Psim)_{\ell=3}|=|(T^{j+1}\Phi\errm)_{\ell=3}|\lesssim\ubar^{-8-j-\delta}.
\end{align}
Using Proposition \ref{prop:modecoupling} we also have 
\begin{align}\label{eq:costhetamode}
    (\cos\theta T^{j+1}\Psim)_{\ell=3}=-\parubarj{}\left(\frac{7}{\ubar^8}\right)\sum_{|m|\leq 2}b_{m,2}^{-2}Q_{m,2}Y_{m,3}^{-2}(\cos\theta)e^{im\phi_+}+O(\ubar^{-8-j-\delta}).
\end{align}
Moreover, the Teukolsky equation writes, using the $(\partial_t,\partial_r)$ decomposition of the Teukolsky operator, 
\begin{align*}
    \teuk_{-2}^{[\partial_r]}T^j\Psim&=-\teuk_{-2}^{[\partial_t]}T^j\Psim\\
    &=-a^2\sin^2\theta T^{j+2}\Psim-2(r^2+a^2)T^{j+1}\drin\Psim-2aT^{j+1}\Phi\Psim\\
    &\quad-(10r+4ia\cos\theta)T^{j+1}\Psim.
\end{align*}
Projecting onto the $\ell=3$ modes thus gives, using \eqref{eq:premiere}--\eqref{eq:costhetamode}, \eqref{eq:projcostheta} and Lemma \ref{lem:commproj},
\begin{align}\label{eq:aproj}
    \teuk_{-2}^{[\partial_r]}(T^j\Psim)_{\ell=3}=28ia\parubarj{}\left(\frac{1}{\ubar^8}\right)\sum_{|m|\leq 2}b_{m,2}^{-2}Q_{m,2}Y_{m,3}^{-2}(\cos\theta)e^{im\phi_+}+O(\ubar^{-8-j-\delta}).
\end{align}
Then using the expression \eqref{eq:teukdr} of the operator $\teuk_{-2}^{[\partial_r]}$, as well as \eqref{eq:annulemode}, \eqref{eq:aproj} rewrites
\begin{align*}
    \sum_{|m|\leq 3}\Big[\Delta\parrdeux{}\left(\parubarj{F_{m,3}}\right)(r,\ubar)+(2iam+&6(r-M))\parr{}\left(\parubarj{F_{m,3}}\right)(r,\ubar)-6\parubarj{F_{m,3}}(r,\ubar)\Big]Y_{m,3}^{-2}(\cos\theta)e^{im\phi_+}\\
    &=28ia\parubarj{}\left(\frac{1}{\ubar^8}\right)\sum_{|m|\leq 2}b_{m,2}^{-2}Q_{m,2}Y_{m,3}^{-2}(\cos\theta)e^{im\phi_+}+O(\ubar^{-8-j-\delta}),
\end{align*}
which concludes the proof by projecting onto an ingoing $(\ell=3,m)$ mode with $|m|\leq 3$.
\end{proof}
\subsubsection{Resolution of the ODE \eqref{eq:odeI}}
In the setting of Proposition \ref{prop:edoI3}, let us fix $\ubar\geq 1$ and $j\geq 0$. Recall the functions ${}^{(j)}f_{m,3,\ubar}$ defined in \eqref{eq:deffm}. By \eqref{eq:odeI}, the ODE satisfied by ${}^{(j)}\!f_{m,3,\ubar}$ is 
\begin{align}\label{eq:odeode}
    \Delta {}^{(j)}\!f_{m,3,\ubar}''(r)+(2iam+6(r-M)){}^{(j)}\!f_{m,3,\ubar}'(r)-6{}^{(j)}\!f_{m,3,\ubar}(r)=R^{(j)}_m(r,\ubar),
\end{align}
where the source $R^{(j)}_m(r,\ubar)$ satisfies 
\begin{align}\label{eq:sourceI}
    R^{(j)}_m(r,\ubar)=28ia\parubarj{}\left(\frac{1}{\ubar^8}\right)b_{m,2}^{-2}Q_{m,2}+O(\ubar^{-8-j-\delta}).
\end{align}
As before, we will recover ${}^{(j)}\!f_{m,3,\ubar}$ using a basis of solutions of the homogeneous problem 
\begin{align}\label{eq:odeH}
    \Delta v''(r)+(2iam+6(r-M))v'(r)-6v(r)=0,
\end{align}
using the variation of constants method. An obvious solution of \eqref{eq:odeH} is 
\begin{align}
    {}^{(3)}\!v_1(r):=r-M+\frac{iam}{3}.\label{eq:h1def}
\end{align}
Let ${}^{(3)}\!v_2$ be another solution of \eqref{eq:odeH}, and define the Wronskian
$$W_3:={}^{(3)}\!v_1'{}^{(3)}\!v_2-{}^{(3)}\!v_2'{}^{(3)}\!v_1={}^{(3)}\!v_2-(r-M+aim/3){}^{(3)}\!v_2'.$$
We obtain, differentiating $W_1$,
\begin{align*}
    W'_1&=-(r-M+aim/3){}^{(3)}\!v_2''\\
    &=\frac{r-M+aim/3}{\Delta}((6(r-M)+2iam){}^{(3)}\!v_2'-6{}^{(3)}\!v_2)\\
    &=-\frac{6(r-M)+2iam}{\Delta}W_1.
\end{align*}
Thus, making an appropriate choice of free constant, we have
$$W_1=\exp\left(-\int^r\frac{6(r'-M)+2iam}{\Delta(r')}\dee r'\right)=-\Delta^{-3}e^{-2imr_{mod}}.$$
This gives 
$$\left(\frac{{}^{(3)}\!v_2}{{}^{(3)}\!v_1}\right)'=-\frac{W_1}{{}^{(3)}\!v_1^2}=\frac{\Delta^{-3}e^{-2imr_{mod}}}{{}^{(3)}\!v_1(r)^2}.$$
Thus we can choose the other member of the basis of solutions of \eqref{eq:odeH} to be 
\begin{align}\label{eq:h2def}
    {}^{(3)}\!v_2(r):=(r-M+aim/3)\int_{\rb}^r\frac{\Delta^{-3}(r')e^{-2imr_{mod}(r')}}{(r'-M+aim/3)^2}\dee r'.
\end{align}
\begin{rem}
    Note that even for $m=0$ where ${}^{(3)}\!v_1(r)$ vanishes at $r=M$, the RHS of \eqref{eq:h2def} admits a limit as $r\to M$, and is thus well defined on $(r_-,r_+)$.
\end{rem}
Now we use the variation of constants method to find a particular solution ${}^{(j)}\!p_{m,3,\ubar}$ of \eqref{eq:odeode} of the form
$${}^{(j)}\!p_{m,3,\ubar}(r):={}^{(j)}\!M_{m,3,\ubar}(r){}^{(3)}\!v_1(r)+{}^{(j)}\!N_{m,3,\ubar}(r){}^{(3)}\!v_2(r),$$
where ${}^{(j)}\!M_{m,3,\ubar}(r)$ and ${}^{(j)}\!N_{m,3,\ubar}(r)$ satisfy
\begin{align}\label{eq:systemI}
    \begin{cases}
        {}^{(j)}\!M_{m,3,\ubar}'{}^{(3)}\!v_1+{}^{(j)}\!N_{m,3,\ubar}'{}^{(3)}\!v_2=0,\\
        {}^{(j)}\!M_{m,3,\ubar}'{}^{(3)}\!v_1'+{}^{(j)}\!N_{m,3,\ubar}'{}^{(3)}\!v_2'={\Delta}^{-1}{R^{(j)}_m(r,\ubar)},
    \end{cases}
\end{align}
Solving \eqref{eq:systemI} gives ${}^{(j)}\!M_{m,3,\ubar}'W_1=R^{(j)}_m(r,\ubar){}^{(3)}\!v_2/\Delta$ and ${}^{(j)}\!N_{m,3,\ubar}'W_1=-R^{(j)}_m(r,\ubar){}^{(3)}\!v_1/\Delta$, hence 
$${}^{(j)}\!M_{m,3,\ubar}'(r)=-\Delta^2 {}^{(3)}\!v_2(r)R^{(j)}_m(r,\ubar)e^{2imr_{mod}},\quad {}^{(j)}\!N_{m,3,\ubar}'(r)=\Delta^2 {}^{(3)}\!v_1(r)R^{(j)}_m(r,\ubar)e^{2imr_{mod}}.$$
Our particular solution can thus be chosen as
\begin{align}\label{eq:luila}
    {}^{(j)}\!p_{m,3,\ubar}(r)=-&{}^{(3)}\!v_1(r)\int_{r_+}^r\Delta^2(r'){}^{(3)}\!v_2(r')R^{(j)}_m(r',\ubar)e^{2imr_{mod}(r')}\dee r'\nn\\
    &+{}^{(3)}\!v_2(r)\int_{r_+}^r\Delta^2(r'){}^{(3)}\!v_1(r')R^{(j)}_m(r',\ubar)e^{2imr_{mod}(r')}\dee r'.
\end{align}
Notice that a naive upper bound for ${}^{(3)}\!v_2$ is $(r_+-r)^{-2}$ as $r\to r_+$, thus the first integral on the RHS of \eqref{eq:luila} is well-defined. Note that using the expression \eqref{eq:h2def} of ${}^{(3)}\!v_2$ yields 
\begin{align*}
    {}^{(j)}\!p_{m,3,\ubar}(r)=&-{}^{(3)}\!v_1(r)\int_r^{r_+}\Delta^2(r')R^{(j)}_m(r',\ubar)e^{2imr_{mod}(r')}{}^{(3)}\!v_1(r')\int_{\rb}^{r'}\frac{\Delta^{-3}(r'')e^{-2imr_{mod}(r'')}}{{}^{(3)}\!v_1(r'')^2}\dee r''\dee r'\\
    &+{}^{(3)}\!v_1(r)\left(\int_{\rb}^{r}\frac{\Delta^{-3}(r')e^{-2imr_{mod}(r')}}{{}^{(3)}\!v_1(r')^2}\dee r'\right)\int_r^{r_+}\Delta^2(r')R^{(j)}_m(r',\ubar)e^{2imr_{mod}(r')}{}^{(3)}\!v_1(r')\dee r'.
\end{align*}
Thus defining the function 
\begin{align}\label{eq:Bmdef}
    K_m(r,r'):=\Delta^2(r')e^{2imr_{mod}(r')}(r'-M+aim/3)\int_{r'}^r\frac{\Delta^{-3}(r'')e^{-2imr_{mod}(r'')}}{(r''-M+aim/3)^2}\dee r'',
\end{align}
we obtain the expression
\begin{align}
    {}^{(j)}\!p_{m,3,\ubar}(r)=(r-M+aim/3)\int_{r_+}^rR^{(j)}_m(r',\ubar)K_m(r,r')\dee r'.
\end{align}
As ${}^{(j)}\!f_{m,3,\ubar}-{}^{(j)}\!p_{m,3,\ubar}$ satisfies the homogeneous equation \eqref{eq:odeH}, there are constants $A_{m,3}^{(j)}(\ubar)$ and $B_{m,3}^{(j)}(\ubar)$ that depend only on $\ubar$, $m$ and $j$ such that for any $r\in[\rb,r_+]$, any $|m|\leq 2$, and any $\ubar\geq 1$,
\begin{align}\label{eq:FMresolu}
    {}^{(j)}\!f_{m,3,\ubar}(r)=A_{m,3}^{(j)}(\ubar){}^{(3)}\!v_1(r)+B_{m,3}^{(j)}(\ubar){}^{(3)}\!v_2(r)+(r-M+aim/3)\int_{r_+}^rR^{(j)}_m(r',\ubar)K_m(r,r')\dee r'.
\end{align}
\subsubsection{Computation of $A_{m,3}^{(j)}(\ubar)$, $B_{m,3}^{(j)}(\ubar)$}
To compute the constants $A_{m,3}^{(j)}(\ubar)$, $B_{m,3}^{(j)}(\ubar)$ in terms of the initial data on $\mch_+$, we first discuss the regularity at $r=r_+$ of the different functions at play. We begin with the following result
\begin{lem}\label{lem:v2dv}
    The function ${}^{(3)}\!v_2(r)$ defined by \eqref{eq:h2def} diverges as $r\to r_\pm$. More precisely,
    \begin{itemize}
    \item For $\rb\leq r\leq r_+$:
    $${}^{(3)}\!v_2(r)=\frac{-\Delta^{-2}(r)e^{-2imr_{mod}}}{(4(r_+-M)+2iam)(r_+-M+aim/3)}+O((r_+-r)^{-1}).$$
    \item For $r_-\leq r\leq\rb$:
    $${}^{(3)}\!v_2(r)=\frac{-\Delta^{-2}(r)e^{-2imr_{mod}}}{(4(r_--M)+2iam)(r_--M+aim/3)}+O((r-r_-)^{-1}).$$
    \end{itemize}
\end{lem}
\begin{proof}
We begin with the first point. We have the estimate
\begin{align*}
    \Bigg|\int_{\rb}^r\frac{\Delta^{-3}(r')e^{-2imr_{mod}(r')}}{(r'-M+aim/3)^2}\dee r'-\frac{1}{(r-M+aim/3)^2}&\int_{\rb}^r\Delta^{-3}(r')e^{-2imr_{mod}(r')}\dee r'\Bigg|\\
    &\lesssim \int_r^{\rb}(r_+-r')^{-2}\dee r'\lesssim (r_+-r)^{-1}.
\end{align*}
We conclude using Lemma \ref{lem:v2dvl}. The second point is proved in a similar way.
\end{proof}
Now, we prove ${}^{(j)}\!p_{m,3,\ubar}(r)\to 0$ as $r\to r_+$. Indeed, we have already seen that the first integral on the RHS of \eqref{eq:luila} is well-defined and thus goes to zero as $r\to r_+$. We can also easily bound the other term:
$$\left|{}^{(3)}\!v_2(r)\int_{r_+}^r\Delta^2(r'){}^{(3)}\!v_1(r')R^{(j)}_m(r',\ubar)e^{2imr_{mod}(r')}\dee r'\right|\lesssim(r_+-r)^{-2}\int_{r}^{r_+}(r_+-r')^2\dee r'\lesssim r_+-r,$$
which proves that indeed ${}^{(j)}\!p_{m,3,\ubar}(r)\to 0$ as $r\to r_+$. This implies, by \eqref{eq:FMresolu},
$$\left|{}^{(j)}\!f_{m,\ubar}(r)-A_{m,3}^{(j)}(\ubar){}^{(3)}\!v_1(r)-B_{m,3}^{(j)}(\ubar){}^{(3)}\!v_2(r)\right|\lesssim|{}^{(j)}\!p_{m,3,\ubar}(r)|\underset{r\to r_+}{\longrightarrow}0.$$
As $T^j\Psim$ is regular at $r=r_+$, $\parubarj{}F_{m,3}(r,\ubar)={}^{(j)}\!f_{m,3,\ubar}(r)$ is also regular at $r=r_+$. By Lemma \ref{lem:v2dv}, this implies necessarily $B_{m,3}^{(j)}(\ubar)=0$, and thus by \eqref{eq:FMresolu},
\begin{align}\label{eq:yeyeye}
    \parubarj{F_{m,3}}(r,\ubar)=A_{m,3}^{(j)}(\ubar){}^{(3)}\!v_1(r)+(r-M+aim/3)\int_{r_+}^rR^{(j)}_m(r',\ubar)K_m(r,r')\dee r'.
\end{align}
Evaluating at $r=r_+$ and using the initial assumption \eqref{eq:hyppl} on $\mch_+$, we find 
$$A_{m,3}^{(j)}(\ubar)=(r_+-M+aim/3)^{-1}{\parubarj{F_{m,3}}(r_+,\ubar)}=O(\ubar^{-7-j-\delta}).$$
Putting back into \eqref{eq:FMresolu} the expression \eqref{eq:sourceI} of the source $R^{(j)}_m(r,\ubar)$, we get in $\un$:
\begin{align}\label{eqn:bornesupj}
    \parubarj{F_{m,3}}(r,\ubar)=O(\ubar^{-7-j-\delta}).
\end{align}
Note that this bound is imprecise. Instead what we will use is the precise asymptotic of a linear combination of $\parr{}\left(\parubarj{F_{m,3}}\right)(r,\ubar)$ and $\parubarj{F_{m,3}}(r,\ubar)$. Differentiating \eqref{eq:yeyeye} with respect to $r$ we find in $\un$
\begin{align}
    \parr{}\left(\parubarj{F_{m,3}}\right)(r,\ubar)=&\:A_{m,3}^{(j)}(\ubar)+\int_{r_+}^rR^{(j)}_m(r',\ubar)K_m(r,r')\dee r'\label{eq:drFmI}\\
    &+\frac{\Delta^{-3}(r)e^{-2imr_{mod}(r)}}{r-M+aim/3}\int_{r_+}^rR^{(j)}_m(r',\ubar)\Delta^2(r')e^{2imr_{mod}(r')}(r'-M+aim/3)\dee r'.\nonumber
\end{align}
\subsection{Upper bound for $(\Psihatm)_{\ell=3}$ in region $\deux\cup\trois$}\label{section:l=3II}
We continue the analysis in region $\deux\cup\trois=\{r_-\leq r\leq \rb\}\cap\{w\leq w_{\rb,\gamma}\}$.
\begin{prop}\label{prop:edoII3} Assume that $\Psim$ satisfies \eqref{eq:hyppl}. Then, defining by convention $Q_{\pm3,2}=0$, the functions ${G_{m,3}(r,u)}$ satisfy the following equations in $\deux\cup\trois$, for $|m|\leq 3$:
    \begin{align}\label{eq:odeII}
    \Delta\parrdeux{G_{m,3}}(r,u)-(2iam+2(r-M))\parr{G_{m,3}}&(r,u)-10{G_{m,3}(r,u)}=\nn\\
    &\frac{28ia\Delta^2e^{2imr_{mod}}}{\ubar^8}b_{m,2}^{-2}Q_{m,2}+O(|u|^{-8-\delta/2}).
\end{align}
\end{prop}
\begin{proof}
    Recall that in $\deux\cup\trois=\{r_-\leq r\leq \rb\}\cap\{w\leq w_{\rb,\gamma}\}$, we have 
$$\Psihatm=\anshatm+\errhatm,$$
where $|T^j\carterm^{k_1}\Phi^{k_2}\drout^{\leq 1}\errhatm|\lesssim|u|^{-7-j-\delta/2}$, using Propositions \ref{prop:upperboundII} and \ref{prop:almostsharpdrout}, as well as Lemma \ref{lem:usimubarII}. This yields the estimates in region $\{r_-\leq r\leq \rb\}\cap\{w\leq w_{\rb,\gamma}\}$: 
\begin{align}\label{eq:premiereII}
    &|(\sin^2\theta T^2\Psihatm)_{\ell=3}|\lesssim|u|^{-9},\\
    &|(T\Psihatm)_{\ell= 3}|=|(T\errhatm)_{\ell= 3}|\lesssim|u|^{-8-\delta/2},\\
    &|(T\drout\Psihatm)_{\ell= 3}|=|(T\drout\errhatm)_{\ell= 3}|\lesssim|u|^{-8-\delta/2},\\
    &|(T\Phi\Psihatm)_{\ell=3}|=|(T\Phi\errhatm)_{\ell=3}|\lesssim|u|^{-8-\delta/2}.
\end{align}
Moreover, by \eqref{eq:costhetamode} with $j=0$ we get 
\begin{align}\label{eq:costhetamodeII}
    (\cos\theta T\Psihatm)_{\ell=3}=-\frac{7\Delta^2}{\ubar^8}\sum_{|m|\leq 2}b_{m,2}^{-2}Q_{m,2}Y_{m,3}^{-2}(\cos\theta)e^{im\phi_+}+O(|u|^{-8-\delta/2}).
\end{align}
The Teukolsky equation writes, using the $(\partial_t,\partial_r)$ decomposition \eqref{eq:decteuk} of $\teukhat_s$, 
$$\teukhat_{-2}^{[\partial_r]}\Psihatm=-\teukhat_{-2}^{[\partial_t]}\Psihatm=-a^2\sin^2\theta T^2\Psihatm+2(r^2+a^2)T\drout\Psihatm-2aT\Phi\Psihatm-(6r+4ia\cos\theta)T\Psim.$$
Projecting onto the modes $\ell=3$ thus gives, using \eqref{eq:premiereII}--\eqref{eq:costhetamodeII} and \eqref{eq:projcostheta},
\begin{align}\label{eq:aprojII}
    \teukhat_{-2}^{[\partial_r]}(\Psihatm)_{\ell=3}=\frac{28ia\Delta^2}{\ubar^8}\sum_{|m|\leq 2}b_{m,2}^{-2}Q_{m,2}Y_{m,3}^{-2}(\cos\theta)e^{im\phi_+}+O(|u|^{-8-\delta/2}).
\end{align}
Then, using the expression \eqref{eq:teukhatdr} of the operator $\teukhat_{-2}^{[\partial_r]}$, as well as \eqref{eq:dronddrond'spin}, \eqref{eq:aprojII} rewrites
\begin{align*}
    \sum_{|m|\leq 3}\Big[\Delta\partial_r^2{G_{m,3}(r,u)}-(2iam+&2(r-M))\partial_r{G_{m,3}(r,u)}-10{G_{m,3}(r,u)}\Big]Y_{m,3}^{-2}(\cos\theta)e^{im\phi_-}\nonumber\\
    &=\frac{28ia\Delta^2}{\ubar^8}\sum_{|m|\leq 2}b_{m,2}^{-2}Q_{m,2}Y_{m,3}^{-2}(\cos\theta)e^{im\phi_+}+O(|u|^{-8-\delta/2}),
\end{align*}
which concludes the proof, projecting on each outgoing $(\ell=3,m)$ mode with $|m|\leq 3$.
\end{proof}
\subsubsection{Resolution of the ODE \eqref{eq:odeII}}
In the setting of Proposition \ref{prop:edoII3}, let us consider values of $u\ll -1$ such that the corresponding level sets are in $\deux\cup\trois$. Recall the function $g_{m,3,u}(r)={G_{m,3}(r,u)}$. By \eqref{eq:odeII}, the ODE satisfied by $g_{m,3,u}$ is 
\begin{align}\label{eq:odeodeII}
    \Delta g_{m,3,u}''(r)-(2iam+2(r-M))g_{m,3,u}'(r)-10g_{m,3,u}(r)=\widehat{R}_m(r,u),
\end{align}
where the source $\widehat{R}_m(r,u)$ satisfies 
\begin{align}\label{eq:sourceII}
    \widehat{R}_m(r,u)=\frac{28ia\Delta^2e^{2imr_{mod}}}{(2r^*-u)^8}b_{m,2}^{-2}Q_{m,2}+O(|u|^{-8-\delta/2}).
\end{align}
As before, we use a basis of solutions to the homogeneous problem 
\begin{align}\label{eq:odeHII}
    \Delta \widehat{v}''(r)-(2iam+2(r-M))\widehat{v}'(r)-10\widehat{v}(r)=0,
\end{align}
To solve \eqref{eq:odeHII}, notice that if $v$ satisfies \eqref{eq:odeH}, then $\widehat{v}:=v\Delta^2e^{2imr_{mod}}$ satisfies \eqref{eq:odeHII}. Indeed, under these assumptions we have as before
\begin{align*}
    \widehat{v}'&=[\Delta(4(r-M)+2iam)v+\Delta^2v']e^{2imr_{mod}},\\
    \widehat{v}''&=[2(r-M)(4(r-M)+2iam)v+4\Delta v+\Delta(4(r-M)+2iam)v'+4(r-M)\Delta v'+\Delta^2v''\\
    &\quad\quad +2iam(4(r-M)+2iam)v+2iam\Delta v']e^{2imr_{mod}},\\
    &=[2(r-M)+2iam]\Delta v'e^{2imr_{mod}}+[(2(r-M)+2iam)(4(r-M)+2iam)+10\Delta]ve^{2imr_{mod}},
\end{align*}
which gives as stated
\begin{align*}
    \Delta\widehat{v}''=(2(r-M)+2iam)\widehat{v}'+10\Delta^2ve^{2imr_{mod}}=(2(r-M)+2iam)\widehat{v}'+10\widehat{v}.
\end{align*}
This proves, recalling \eqref{eq:h1def}, \eqref{eq:h2def}, that a basis $({}^{(3)}\!\widehat{v_1},{}^{(3)}\!\widehat{v_2})$ of solutions  of \eqref{eq:odeHII} can be chosen as 
\begin{align}
    &{}^{(3)}\!\widehat{v_1}(r):=(r-M+{iam}/{3})\Delta^2(r)e^{2imr_{mod}(r)},\\
    &{}^{(3)}\!\widehat{v_2}(r):=(r-M+aim/3)\Delta^2(r)e^{2imr_{mod}(r)}\int_{\rb}^r\frac{\Delta^{-3}(r')e^{-2imr_{mod}(r')}}{(r'-M+aim/3)^2}\dee r'.
\end{align}
Now we do the variation of constants method to find a particular solution $\widehat{p}_{m,3,u}(r)$ of \eqref{eq:odeodeII}. As before, let $\widehat{M}_{m,3,u}(r)$, $\widehat{N}_{m,3,u}(r)$ be functions such that
\begin{align}\label{eq:systemII}
    \begin{cases}
        \widehat{M}_{m,3,u}'{}^{(3)}\!\widehat{v_1}+\widehat{N}_{m,3,u}'{}^{(3)}\!\widehat{v_2}=0,\\
        \widehat{M}_{m,3,u}'{}^{(3)}\!\widehat{v_1}'+\widehat{N}_{m,3,u}'{}^{(3)}\!\widehat{v_2}'={\Delta}^{-1}{\widehat{R}_m(r,u)},
    \end{cases}
\end{align}
then $\widehat{p}_{m,3,u}=\widehat{M}_{m,3,u}{}^{(3)}\!\widehat{v_1}+\widehat{N}_{m,3,u}{}^{(3)}\!\widehat{v_2}$ is a solution of \eqref{eq:odeodeII}. To solve \eqref{eq:systemII}, let us compute the Wronskian $\widehat{W}_3$ of $({}^{(3)}\!\widehat{v_1}, {}^{(3)}\!\widehat{v_2})$. We have 
$${}^{(3)}\!\widehat{v_2}={}^{(3)}\!\widehat{v_1}\int_{\rb}^r\frac{\Delta(r')e^{2imr_{mod}(r')}}{{}^{(3)}\!\widehat{v_1}(r')^2}\dee r'$$
thus the Wronskian is 
$$\widehat{W}_3={}^{(3)}\!\widehat{v_1}'{}^{(3)}\!\widehat{v_2}-{}^{(3)}\!\widehat{v_2}'{}^{(3)}\!\widehat{v_1}=-\Delta e^{2imr_{mod}}.$$
This gives $\widehat{M}_{m,3,u}'\widehat{W}_3=\widehat{R}_m(r,u){}^{(3)}\!\widehat{v_2}/\Delta$ and $\widehat{N}_{m,3,u}'\widehat{W}_3=-\widehat{R}_m(r,u){}^{(3)}\!\widehat{v_1}/\Delta$, hence 
$$\widehat{M}_{m,3,u}'=\frac{- {}^{(3)}\!\widehat{v_2}(r)\widehat{R}_m(r,u)}{\Delta^2e^{2imr_{mod}}},\quad \widehat{N}_{m,3,u}'=\frac{{}^{(3)}\!\widehat{v_1}(r)\widehat{R}_m(r,u)}{\Delta^2e^{2imr_{mod}}}.$$
More precisely, this gives 
\begin{align*}
    &\widehat{M}_{m,3,u}'=-{}^{(3)}\!{v_2}(r)\widehat{R}_m(r,u)=-(r-M+aim/3)\widehat{R}_m(r,u)\int_{\rb}^r\frac{\Delta^{-3}(r')e^{-2imr_{mod}(r')}}{(r'-M+aim/3)^2}\dee r',\\
    &\widehat{N}_{m,3,u}'={}^{(3)}\!{v_1}(r)\widehat{R}_m(r,u)=(r-M+aim/3)\widehat{R}_m(r,u).
\end{align*}
Our particular solution can thus be chosen as
\begin{align*}
    \widehat{p}_{m,3,u}(r)=&-{}^{(3)}\!\widehat{v_1}(r)\int_{\rb}^r(r'-M+aim/3)\widehat{R}_m(r',u)\left(\int_{\rb}^{r'}\frac{\Delta^{-3}(r'')e^{-2imr_{mod}(r'')}}{(r''-M+aim/3)^2}\dee r''\right)\dee r'\\
    &+{}^{(3)}\!\widehat{v_2}(r)\int_{\rb}^{r}(r'-M+aim/3)\widehat{R}_m(r',u)\dee r'.
\end{align*}
Finally, we get that there are constants $\widehat{A}_{m,3}(u)$, $\widehat{B}_{m,3}(u)$ such that in $\deux\cup\trois$ we have
\begin{align}
    {G_{m,3}(r,u)}=&-{}^{(3)}\!\widehat{v_1}(r)\int_{\rb}^r(r'-M+aim/3)\widehat{R}_m(r',u)\left(\int_{\rb}^{r'}\frac{\Delta^{-3}(r'')e^{-2imr_{mod}(r'')}}{(r''-M+aim/3)^2}\dee r''\right)\dee r'\nn\\
    &+{}^{(3)}\!\widehat{v_2}(r)\int_{\rb}^{r}(r'-M+aim/3)\widehat{R}_m(r',u)\dee r'\nn\\
    &+\widehat{A}_{m,3}(u){}^{(3)}\!\widehat{v_1}(r)+\widehat{B}_{m,3}(u){}^{(3)}\!\widehat{v_2}(r).\label{eq:troiz}
\end{align}
\subsubsection{Computation of the constants $\widehat{A}_{m,3}(u)$, $\widehat{B}_{m,3}(u)$}
We have by \eqref{eq:troiz},
$$\widehat{A}_{m,3}(u){}^{(3)}\!\widehat{v_1}(\rb)=G_{m,3}(\rb,u)=\Delta^2(\rb)e^{2imr_{mod}(\rb)}{F_{m,3}}(\rb,\ubar_{\rb}(u))$$
which yields
\begin{align}\label{eq:cm1exp}
    \widehat{A}_{m,3}(u)&=\frac{{F_{m,3}}(\rb,\ubar_{\rb}(u))}{\rb-M+aim/3}.
\end{align}
We also have, using \eqref{eq:lien},
\begin{align*}
    \parr{G_{m,3}}(\rb,u)&=\parr{}\:\Big|_{r=\rb}(\Delta^2e^{2imr_{mod}}{F_{m,3}}(r,2r^*-u))\\
    &=(4(\rb-M)+2iam)\Delta(\rb)e^{2imr_{mod}(\rb)}{F_{m,3}}(\rb,\ubar_{\rb}(u))\\
    &\quad+2(\rb^2+a^2)\Delta(\rb)e^{2imr_{mod}(\rb)}\parubar{F_{m,3}}(\rb,\ubar_{\rb}(u))+\Delta^2(\rb)e^{2imr_{mod}(\rb)}\parr{F_{m,3}}(\rb,\ubar_{\rb}(u)).
\end{align*}
Moreover, by \eqref{eq:troiz}, using ${}^{(3)}\!\widehat{v_2}(\rb)=0$, $$\parr{G_{m,3}}(\rb,u)=\widehat{A}_{m,3}(u){}^{(3)}\!\widehat{v_1}'(\rb)+\widehat{B}_{m,3}(u){}^{(3)}\!\widehat{v_2}'(\rb),$$
thus using \eqref{eqn:bornesupj} with $j=1$, this gives 
\begin{align*}
    \widehat{A}_{m,3}(u){}^{(3)}\!\widehat{v_1}'(\rb)+\widehat{B}_{m,3}(u){}^{(3)}\!\widehat{v_2}'(\rb)=&(4(\rb-M)+2iam)\Delta(\rb)e^{2imr_{mod}(\rb)}{F_{m,3}}(\rb,\ubar_{\rb}(u))\\
    &\quad+\Delta^2(\rb)e^{2imr_{mod}(\rb)}\partial_r{F_{m,3}}(\rb,\ubar_{\rb}(u))+O(|u|^{-8-\delta/2}).
\end{align*}
Moreover, we have by \eqref{eq:cm1exp},
$$\widehat{A}_{m,3}(u){}^{(3)}\!\widehat{v_1}'(\rb)=\Bigg[(4(\rb-M)+2iam)\Delta(\rb)e^{2imr_{mod}(\rb)}+\frac{\Delta^2(\rb)e^{2imr_{mod}(\rb)}}{\rb-M+aim/3}\Bigg]{F_{m,3}}(\rb,\ubar_{\rb}(u)),$$
thus using \eqref{eq:drFmI} and \eqref{eq:yeyeye} (i.e. the special combination of $\partial_r{F_{m,3}}$ and ${F_{m,3}}$ which cancels the term with $A_{m,3}(\ubar)$), we get
\begin{align*}
    \widehat{B}_{m,3}(u){}^{(3)}\!\widehat{v_2}'(\rb)&=\Delta^2(\rb)e^{2imr_{mod}(\rb)}\Bigg[\parr{F_{m,3}}(\rb,\ubar_{\rb}(u))-\frac{{F_{m,3}}(\rb,\ubar_{\rb}(u))}{\rb-M+aim/3}\Bigg]+O(|u|^{-8-\delta/2})\\
    &=\frac{\Delta^{-1}(\rb)}{\rb-M+aim/3}\int_{r_+}^{\rb}R^{(0)}_m(r',\ubar_{\rb}(u))\Delta^2(r')e^{2imr_{mod}(r')}(r'-M+aim/3)\dee r'\\
    &\quad\quad+O(|u|^{-8-\delta/2}).
\end{align*}
Finally, using ${}^{(3)}\!\widehat{v_2}'(\rb)=\Delta^{-1}(\rb)/(\rb-M+aim/3)$, \eqref{eq:sourceI} with $j=0$, and $\ubar_{\rb}(u)=2\rb^*-u$, we find  
\begin{align}
    \widehat{B}_{m,3}(u)&=\int_{r_+}^{\rb}R^{(0)}_m(r',\ubar_{\rb}(u))\Delta^2(r')e^{2imr_{mod}(r')}(r'-M+aim/3)\dee r'+O(|u|^{-8-\delta/2})\nn\\
    &=\frac{28aib_{m,2}^{-2}Q_{m,2}}{u^8}\int_{r_+}^{\rb}\Delta^2(r')e^{2imr_{mod}(r')}(r'-M+aim/3)\dee r'+O(|u|^{-8-\delta/2}).\label{eq:K2(u)}
\end{align}
\subsubsection{More precise estimate near $\ch$}
    We now turn to estimating the terms on the RHS of \eqref{eq:troiz} near $\ch$. A simple double integration bound gives that the first line is bounded by $-\Delta(r)^{-1}|{}^{(3)}\!\widehat{v_1}(r)|\lesssim r-r_-$. We also have $|\widehat{A}_{m,3}(u){}^{(3)}\!\widehat{v_1}(r)|\lesssim (r-r_-)^2$. For the term on the second line of the RHS of \eqref{eq:troiz}, we have using \eqref{eq:sourceII},
\begin{align*}
    \Bigg|\int_{\rb}^{r}\widehat{R}_m(r',&u)(r'-M+aim/3)\dee r'-\frac{28aimb_{m,2}^{-2}Q_{m,2}}{u^8}\int_{\rb}^r\Delta^2(r')e^{2imr_{mod}(r')}(r'-M+aim/3)\dee r'\Bigg|\\
    &\lesssim\left|\int_{\rb}^r\left(\frac{1}{(u-2(r')^*)^8}-\frac{1}{u^8}\right)\Delta^2(r')e^{2imr_{mod}(r')}(r'-M+aim/3)\dee r'\right|+|u|^{-8-\delta/2}\\
    &\lesssim \int_r^{\rb} \frac{2(r')^*\Delta(r')^2}{|u|^9}\dee r'+|u|^{-8-\delta/2}\\
    &\lesssim |u|^{-8-\delta/2}.
\end{align*}
Next, notice that ${}^{(3)}\!\widehat{v_2}(r)$ converges on $r=r_-$. Indeed, using Lemma \ref{lem:v2dv}, we find 
$${}^{(3)}\!\widehat{v_2}(r)\underset{r\to r_-}{\longrightarrow}-\frac{1}{(4(r_--M)+2iam)(r_--M+aim/3)},$$
and $|{}^{(3)}\!\widehat{v_2}(r)-{}^{(3)}\!\widehat{v_2}(r_-)|\lesssim r-r_-$. Finally, taking into account the expression \eqref{eq:K2(u)} of $\widehat{B}_{m,3}(u)$, we get 
\begin{align*}
    {G_{m,3}(r,u)}=&\frac{28aib_{m,2}^{-2}Q_{m,2}{}^{(3)}\!\widehat{v_2}(r_-)}{u^8}\int_{r_+}^{r_-}\Delta^2(r')e^{2imr_{mod}(r')}(r'-M+aim/3)\dee r'\\
    &\quad+O(|u|^{-8-\delta/2})+O(r-r_-).
\end{align*}
Finally, notice that 
\begin{align}\label{eq:remarkable}
    \int_{r_+}^{r_-}\Delta^2(r')e^{2imr_{mod}(r')}(r'-M+aim/3)\dee r'=\frac{1}{6}\int_{r_+}^{r_-}\frac{\dee}{\dee r'}\Big[\Delta^3(r')e^{2imr_{mod}(r')}\Big]\dee r'=0,
\end{align}
which gives $|G_{m,3}(r,u)|\lesssim|u|^{-8-\delta/2}+r-r_-$ and thus, by \eqref{eq:outgoingdec}, concludes the proof of Theorem \ref{thm:decayl3}. 
\begin{rem}
     We saw that for $(\Psim)_{\ell=2}$, which decays like $\ubar^{-7}$ on $\eh$, the constant in front of $1/u^7$ in the asymptotic development of $(\Psihatm)_{\ell=2}$ near $\ch$ vanishes. Strikingly, this is also the case for $\ell=3$ where $(\Psim)_{\ell=3}$ is expected to decay like $\ubar^{-8}$ on $\eh$ \cite[Section 1.3, (2)]{MZ23}, while the constant in front of $1/u^8$ in the asymptotic development of $(\Psihatm)_{\ell=3}$ near $\ch$ vanishes by \eqref{eq:remarkable}. We expect that this cancellation occurs for every fixed $\ell$ mode of $\Psihatm$ near $\ch$.
\end{rem}
\subsection{Upper bound for $(\Psim)_{\ell\geq 4}$ in region $\un$}\label{section:higherI}
The rest of this section is dedicated to proving Theorem \ref{thm:decayhigher}. We proceed exactly like for the modes $\ell=2$ and $\ell=3$, but this time the source terms of the ODEs will be directly treated as an error term. {Here we consider $k\in\{0,1\}$, $j\in\{0,1,2\}$, and\footnote{The ODE method will give control of $\|(\psihat)_{\ell\geq 4}\|_{L^2(S(u,\ubar))}$ near $\ch$, and we will recover a pointwise bound for $|(\Psihatm)_{\ell\geq 4}|$ by the Sobolev embedding \eqref{eq:sobolev}.} 
$$\psi:=\mcq_{-2}^kT^j\Psim,\quad\psihat:=\mcq_{-2}^kT^j\Psihatm.$$
We will consider the spin-weighted spherical harmonics decomposition
\begin{align}\label{eq:definedG}
    (\psi)_{\ell\geq 4}=\sum_{\ell\geq 4}\sum_{|m|\leq\ell}F_{m,\ell}(r,\ubar)Y_{m,\ell}^{-2}(\cos\theta)e^{im\phi_+},\quad(\psihat)_{\ell\geq 4}=\sum_{\ell\geq 4}\sum_{|m|\leq\ell}G_{m,\ell}(r,u)Y_{m,\ell}^{-2}(\cos\theta)e^{im\phi_-},
\end{align}
such that $G_{m,\ell}(r,u)=F_{m,\ell}(r,2r^*-u)\Delta^2e^{2imr_{mod}}$.}
\begin{prop}\label{prop:edoI3ell}
    Let $\ell\geq 4$. Assume that $\Psim$ satisfies \eqref{eq:hyppl}. Then, the functions ${F_{m,\ell}}(r,\ubar)$ satisfy in $\un$ the equations, for $|m|\leq\ell$,
    \begin{align}\label{eq:odeIell}
    \Delta\parrdeux{F_{m,\ell}}(r,\ubar)+(2iam+6(r-M))\parr{F_{m,\ell}}(r,\ubar)+\lambda_\ell{F_{m,\ell}}(r,\ubar)=O(\ell^{-2}\ubar^{-8-j-\delta}),
\end{align}
where $\lambda_\ell=-(\ell-2)(\ell+3)$.
\end{prop}
\begin{proof}
Using \eqref{eq:borneI} we have in $\un$
$$\psi=\mcq_{-2}^kT^j\left(\ansatzm\right)+\mcq_{-2}^kT^j\errm,$$
where $|T^j\carterm^{k_1}\Phi^{k_2}\drin^{\leq 1}\errm|\lesssim\ubar^{-7-j-\delta}$. By Proposition \ref{prop:modecoupling}, we get for $k=0,1$, and for $\ell\geq 4$,
\begin{align*}
    (\drin^{\leq 1}\mcq_{-2}^kT^j\Psim)_{(m,\ell)}^+=(\drin^{\leq 1}\mcq_{-2}^kT^j\errm)_{(m,\ell)}^++O(\ubar^{-9-j}),
\end{align*}
where the term $O(\ubar^{-9-j})$ on the RHS is only non-zero for $\ell=4$, so we can express it as $O(\ell^{-2}\ubar^{-9-j})$. Thus by Proposition \ref{prop:salva} we deduce the following bounds in $\un$, for $|m|\leq \ell$:
\begin{align}\label{eq:premiereell}
    |(\sin^2\theta \mcq_{-2}^kT^{j+2}\Psim)_{\ell,m}^+|&\lesssim\ell^{-2}\|\mcq_{-2}^{\leq k+1}T^{\leq 2}(\sin^2\theta T^{j+2}\Psim)\|_{L^\infty(S(u,\ubar))}\lesssim \ell^{-2}{\ubar^{-9-j}},\\
    |(\mcq_{-2}^kT^{j+1}\Psim)_{\ell,m}^+|&=|(\mcq_{-2}^kT^{j+1}\errm)_{\ell,m}^+|+O(\ell^{-2}\ubar^{-9-j})\\
    &\lesssim\ell^{-2}\|\mcq_{-2}^{\leq k+1}T^{\leq 2}T^{j+1}\errm\|_{L^\infty(S(u,\ubar))}+\ell^{-2}\ubar^{-9-j}\lesssim\ell^{-2}\ubar^{-8-j-\delta},\nn\\
    |(\mcq_{-2}^kT^{j+1}\drin\Psim)_{\ell,m}^+|&=|(\mcq_{-2}^{\leq k+1}T^{j+1}\drin\errm)_{\ell,m}^+|+O(\ell^{-2}\ubar^{-9-j})\\\
    &\lesssim\ell^{-2}\|\mcq_{-2}^{\leq k+1}T^{\leq 2}T^{j+1}\drin\errm\|_{L^\infty(S(u,\ubar))}+\ell^{-2}\ubar^{-9-j}\lesssim\ell^{-2}\ubar^{-8-j-\delta},\nn\\
    |(\mcq_{-2}^kT^{j+1}\Phi\Psim)_{\ell,m}^+|&=|(\mcq_{-2}^kT^{j+1}\Phi\errm)_{\ell,m}^+|+O(\ell^{-2}\ubar^{-9-j})\\
    &\lesssim\ell^{-2}\|\mcq_{-2}^{\leq k+1}T^{\leq 2}T^{j+1}\Phi\errm\|_{L^\infty(S(u,\ubar))}+\ell^{-2}\ubar^{-9-j}\lesssim\ell^{-2}\ubar^{-8-j-\delta}.\nn
\end{align}
Using Proposition \ref{prop:modecoupling} again we also get 
\begin{align}\label{eq:costhetamodeell}
    |(\cos\theta \mcq_{-2}^kT^{j+1}\Psim)_{\ell,m}^+|&=|(\cos\theta T^{j+1}\errm)_{\ell,m}^+|+O(\ell^{-2}\ubar^{-9-j})\nn\\
    &\lesssim \ell^{-2}\|\mcq_{-2}^{\leq k+1}T^{\leq 2}(\cos\theta T^{j+1}\errm)\|_{L^\infty(S(u,\ubar))}+\ell^{-2}\ubar^{-9-j}\nn\\
    &\lesssim\ell^{-2}\ubar^{-8-j-\delta}.
\end{align}
Moreover, the Teukolsky equation writes, using the $(\partial_t,\partial_r)$ decomposition of the Teukolsky operator, 
\begin{align*}
    \teuk_{-2}^{[\partial_r]}\mcq_{-2}^kT^j\Psim&=-\teuk_{-2}^{[\partial_t]}\mcq_{-2}^kT^j\Psim\\
    &=-a^2\sin^2\theta \mcq_{-2}^kT^{j+2}\Psim-2(r^2+a^2)\mcq_{-2}^kT^{j+1}\drin\Psim-2a\mcq_{-2}^kT^{j+1}\Phi\Psim\\
    &\quad-(10r+4ia\cos\theta)\mcq_{-2}^kT^{j+1}\Psim.
\end{align*}
Projecting onto the ingoing $(\ell,m)$ mode thus gives, using \eqref{eq:premiereell}--\eqref{eq:costhetamodeell},
\begin{align}\label{eq:aprojell}
    (\teuk_{-2}^{[\partial_r]}\mcq_{-2}^kT^j\Psim)_{\ell,m}^+=O(\ell^{-2}\ubar^{-8-j-\delta}).
\end{align}
Then using the expression \eqref{eq:teukdr} of the operator $\teuk_{-2}^{[\partial_r]}$, as well as \eqref{eq:annulemode} and Lemma \ref{lem:commproj}, \eqref{eq:aprojell} rewrites as \eqref{eq:odeIell} for $|m|\leq\ell$, as stated.
\end{proof}
\begin{rem}
    We lose two angular derivatives here to obtain a factor $\ell^{-2}$ in RHS of \eqref{eq:odeIell}. This is done in order to avoid issues of summability with respect to $\ell$, see \eqref{eq:finalesti}.
\end{rem}
\subsubsection{Resolution of the ODE \eqref{eq:odeIell}}
 We fix $\ell\geq 4$, $|m|\leq \ell$ for the rest of the section. In the setting of Proposition \ref{prop:edoI3ell}, let us fix $\ubar\geq 1$. By \eqref{eq:odeIell}, the ODE satisfied by $f_{m,\ell,\ubar}={F_{m,\ell}}(r,\ubar)$ is 
\begin{align}\label{eq:odeodeell}
    \Delta f_{m,\ell,\ubar}''(r)+(2iam+6(r-M))f_{m,\ell,\ubar}'(r)+\lambda_\ell f_{m,\ell,\ubar}(r)=O(\ell^{-2}\ubar^{-8-j-\delta}),
\end{align}
As before, we will recover $f_{m,\ell,\ubar}$ using a basis of solutions of the homogeneous problem 
\begin{align}\label{eq:odeHell}
    \Delta v''(r)+(2iam+6(r-M))v'(r)+\lambda_\ell v(r)=0,
\end{align}
using the variation of constants method. 
\begin{prop}\label{prop:v1ell}
    Let $\ell\geq 4$ and $|m|\leq 4$. The function
    \begin{align}
        {}^{(m,\ell)}\!v_1(r):=(r-r_-)^{\ell-2}+\sum_{k=0}^{\ell-3}\left(\prod_{j=k}^{\ell-3}\frac{(j+1)((j+3)(r_+-r_-)-2iam)}{j(j+5)-(\ell-2)(\ell+3)}\right)(r-r_-)^k\label{eq:h1defell}
    \end{align}
    satisfies the following:
    \begin{enumerate}[(i)]
        \item\label{item:(i)} ${}^{(m,\ell)}\!v_1(r)$ is a solution of the ODE \eqref{eq:odeHell} on $[r_-,r_+]$,
        \item \label{item:(ii)}${}^{(m,\ell)}\!v_1(r) \neq 0$ near $r=r_\pm$,
        \item\label{item:(iii)} ${}^{(m,\ell)}\!v_1(2M-r)=(-1)^{\ell-2}\overline{{}^{(m,\ell)}\!v_1(r)}$ for $r\in [r_-,r_+]$,
        \item\label{item:(iv)} $|{}^{(m,\ell)}\!v_1(r)|\leq|{}^{(m,\ell)}\!v_1(r_-)|$ for $r\in [r_-,r_+]$.
    \end{enumerate}
    Moreover, there exists a function ${}^{(m,\ell)}\!v_2(r)$ defined on $(r_-,r_+)$ which satisfies the following:
    \begin{enumerate}[(a)]
        \item\label{item:(a)} ${}^{(m,\ell)}\!v_2(r)$ is a solution of the ODE \eqref{eq:odeHell} on $(r_-,r_+)$,
        \item\label{item:(b)} ${}^{(m,\ell)}\!v_2(r)$ is linearly independent from ${}^{(m,\ell)}\!v_1(r)$,
        \item\label{item:(c)} the Wronskian $W_{(m,\ell)}:={}^{(m,\ell)}\!v_1'{}^{(m,\ell)}\!v_2-{}^{(m,\ell)}\!v_2'{}^{(m,\ell)}\!v_1$ satisfies
        \begin{align}
            W_{(m,\ell)}=-\Delta^{-3}e^{-2imr_{mod}},
        \end{align}
        \item\label{item:(d)} the asymptotics of ${}^{(m,\ell)}\!v_2(r)$ at $r\to r_\pm$ are given as follows:
        \begin{itemize}
            \item For $\rb\leq r\leq r_+$:
            $${}^{(m,\ell)}\!v_2(r)=\frac{-\Delta^{-2}(r)e^{-2imr_{mod}}}{(4(r_+-M)+2iam){}^{(m,\ell)}\!v_1(r_+)}+O_{(m,\ell)}((r_+-r)^{-1}).$$
            \item For $r_-\leq r\leq\rb$:
            $${}^{(m,\ell)}\!v_2(r)=\frac{-\Delta^{-2}(r)e^{-2imr_{mod}}}{(4(r_--M)+2iam){}^{(m,\ell)}\!v_1(r_-)}+O_{(m,\ell)}((r-r_-)^{-1}),$$
    \end{itemize}
    where the implicit constants may depend on $(m,\ell)$.    \end{enumerate}
\end{prop}
\begin{proof}
    See Appendix \ref{section:v1ell}.
\end{proof}
\begin{rem}\label{rem:jaimebien} We remark the following : 
\begin{enumerate}
    \item The bound in \ref{item:(iv)} is crucial to ensure that the bounds for $(\Psihat)_{(m,\ell)}$ are summable with respect to $\ell\geq 4$, $|m|\leq\ell$, see \eqref{eq:conc1}. \label{item:itemun}
    \item With a bit more work, similarly as for the $\ell=3$ case we could have defined ${}^{(m,\ell)}\!v_2(r)$ by 
    $${}^{(m,\ell)}\!v_2(r):={}^{(m,\ell)}\!v_1(r)\int_{\rb}^r\frac{\Delta^{-3}(r')e^{-2imr_{mod}(r')}}{{}^{(m,\ell)}\!v_1(r')^2}\dee r',$$
    but this will not be needed here.
\end{enumerate}
\end{rem}
Now that we have a basis $({}^{(m,\ell)}\!v_1,{}^{(m,\ell)}\!v_2)$ of solutions of \eqref{eq:odeHell}, we use as before the variation of constants method to find a particular solution $p_{m,\ell,\ubar}$ of \eqref{eq:odeodeell} of the form
\begin{align}\label{eq:luilaell}
    p_{m,\ell,\ubar}(r)=-&{}^{(m,\ell)}\!v_1(r)\int_{r_+}^r\Delta^2(r'){}^{(m,\ell)}\!v_2(r')O(\ell^{-2}\ubar^{-8-j-\delta})e^{2imr_{mod}(r')}\dee r'\nn\\
    &+{}^{(m,\ell)}\!v_2(r)\int_{r_+}^r\Delta^2(r'){}^{(m,\ell)}\!v_1(r')O(\ell^{-2}\ubar^{-8-j-\delta})e^{2imr_{mod}(r')}\dee r'.
\end{align}
Notice that Proposition \ref{prop:v1ell}, \ref{item:(d)} ensures that the first integral on the RHS of \eqref{eq:luilaell} is well-defined. As $f_{m,\ell,\ubar}-p_{m,\ell,\ubar}$ satisfies the homogeneous equation \eqref{eq:odeHell}, we get $A_{m,\ell}(\ubar)$ and $B_{m,\ell}(\ubar)$ such that
\begin{align}\label{eq:FMresoluell}
    f_{m,\ell,\ubar}(r)=A_{m,\ell}(\ubar){}^{(m,\ell)}\!v_1(r)+B_{m,\ell}(\ubar){}^{(m,\ell)}\!v_2(r)+p_{m,\ell,\ubar}(r).
\end{align}
\subsubsection{Computation of $A_{m,\ell}(\ubar)$, $B_{m,\ell}(\ubar)$}
We also have, as before, $p_{m,\ell,\ubar}(r)\to 0$ as $r\to r_+$. This implies, by \eqref{eq:FMresoluell},
$$\left|f_{m,\ubar}(r)-A_{m,\ell}(\ubar){}^{(m,\ell)}\!v_1(r)-B_{m,\ell}(\ubar){}^{(m,\ell)}\!v_2(r)\right|\lesssim|p_{m,\ell,\ubar}(r)|\underset{r\to r_+}{\longrightarrow}0.$$
As $\mcq_{-2}^kT^j\Psim$ is regular at $r=r_+$, $F_{m,\ell}(r,\ubar)=f_{m,\ell,\ubar}(r)$ is also regular at $r=r_+$. By Proposition \ref{prop:v1ell}, \ref{item:(d)}, this implies necessarily $B_{m,\ell}(\ubar)=0$, and thus by \eqref{eq:FMresoluell},
\begin{align}\label{eq:yeyeyeell}
    {F_{m,\ell}}(r,\ubar)=A_{m,\ell}(\ubar){}^{(m,\ell)}\!v_1(r)+p_{m,\ell,\ubar}(r).
\end{align}
Evaluating at $r=r_+$ and using Proposition \ref{prop:salva} and the initial assumption \eqref{eq:hyppl} on $\mch_+$, we find 
$$A_{m,\ell}(\ubar)={}^{(m,\ell)}\!v_1(r_+)^{-1}{{F_{m,\ell}}(r_+,\ubar)}=O({}^{(m,\ell)}\!v_1(r_+)^{-1}\ell^{-2}\ubar^{-7-j-\delta}).$$
This yields, in $\un$, ${F_{m,\ell}}(r,\ubar)=O_{(m,\ell)}\left(\ubar^{-7-j-\delta}\right)$, near $r=\rb$. Once again, we will instead use a particular combination of $\parr{}{F_{m,\ell}}(r,\ubar)$ and ${F_{m,\ell}}(r,\ubar)$. Differentiating \eqref{eq:yeyeyeell} with respect to $r$ we find in $\un$, close to $r=\rb$,
\begin{align*}
    \parr{{F_{m,\ell}}}&=A_{m,\ell}(\ubar){}^{(m,\ell)}\!v_1'(r)+p_{m,\ell,\ubar}'(r)\\
    &=A_{m,\ell}(\ubar){}^{(m,\ell)}\!v_1'(r)+{}^{(m,\ell)}\!v_2'(r)\int_{r_+}^r\Delta^2(r'){}^{(m,\ell)}\!v_1(r')O(\ell^{-2}\ubar^{-8-j-\delta})e^{2imr_{mod}(r')}\dee r'\\
    &\quad\quad-{}^{(m,\ell)}\!v_1'(r)\int_{r_+}^r\Delta^2(r'){}^{(m,\ell)}\!v_2(r')O(\ell^{-2}\ubar^{-8-j-\delta})e^{2imr_{mod}(r')}\dee r'\\
    &=A_{m,\ell}(\ubar){}^{(m,\ell)}\!v_1'(r)+\frac{\Delta^{-3}e^{-2imr_{mod}}}{{}^{(m,\ell)}\!v_1(r)}\int_{r_+}^r\Delta^2(r'){}^{(m,\ell)}\!v_1(r')O(\ell^{-2}\ubar^{-8-j-\delta})e^{2imr_{mod}(r')}\dee r'\\
    &\quad\quad+\frac{{}^{(m,\ell)}\!v_1'(r)}{{}^{(m,\ell)}\!v_1(r)}\left({}^{(m,\ell)}\!v_2(r)\int_{r_+}^r\Delta^2(r'){}^{(m,\ell)}\!v_1(r')O(\ell^{-2}\ubar^{-8-j-\delta})e^{2imr_{mod}(r')}\dee r'-[{}^{(m,\ell)}\!v_1\leftrightarrow {}^{(m,\ell)}\!v_2]\right),
\end{align*}
where we used ${}^{(m,\ell)}\!v_1'{}^{(m,\ell)}\!v_2-{}^{(m,\ell)}\!v_2'{}^{(m,\ell)}\!v_1=-\Delta^{-3}e^{-2imr_{mod}}$. Thus, by \eqref{eq:yeyeyeell} and \eqref{eq:luilaell}, near $r=\rb$,
\begin{align}\label{eq:derrHmell}
    \parr{{F_{m,\ell}}}=\frac{{}^{(m,\ell)}\!v_1'(r)}{{}^{(m,\ell)}\!v_1(r)}{F_{m,\ell}}+\frac{\Delta^{-3}(r)e^{-2imr_{mod}(r)}}{{}^{(m,\ell)}\!v_1(r)}\int_{r_+}^r\Delta^2(r')e^{2imr_{mod}(r')}{}^{(m,\ell)}\!v_1(r')O(\ell^{-2}\ubar^{-8-j-\delta})\dee r'.
\end{align}
\subsection{Upper bound for $(\Psihatm)_{\ell\geq 4}$ in region $\deux\cup\trois$}\label{section:higherII}
We continue the analysis in region $\deux\cup\trois=\{r_-\leq r\leq \rb\}\cap\{w\leq w_{\rb,\gamma}\}$.
\begin{prop}\label{prop:edoII3ell} Assume that $\Psim$ satisfies \eqref{eq:hyppl}. Then the functions ${G_{m,\ell}(r,u)}$, defined in \eqref{eq:definedG}, satisfy in $\deux\cup\trois$ the following equations, for $|m|\leq \ell$:
    \begin{align}\label{eq:odeIIell}
    \Delta\parrdeux{G_{m,\ell}}(r,u)-(2iam+2(r-M))\parr{G_{m,\ell}}&(r,u)+\lambda_\ell'{G_{m,\ell}(r,u)}=O(\ell^{-2}|u|^{-8-\delta/2}),
\end{align}
where $\lambda'_\ell=-(\ell+2)(\ell-1)=\lambda_\ell-4$.
\end{prop}
\begin{proof}
    Recall that in $\deux\cup\trois=\{r_-\leq r\leq \rb\}\cap\{w\leq w_{\rb,\gamma}\}$, we have 
$$\mcq_{-2}^k T^j\Psihatm=\mcq_{-2}^k T^j\left(\anshatm\right)+\mcq_{-2}^k T^j\errhatm,$$
where $|T^j\carterm^{k_1}\Phi^{k_2}\drout^{\leq 1}\errhatm|\lesssim|u|^{-7-j-\delta/2}$, using Propositions \ref{prop:upperboundII} and \ref{prop:almostsharpdrout}, as well as Lemma \ref{lem:usimubarII}. By the bound
\begin{align*}
    (\drout^{\leq 1}\mcq_{-2}^kT^j\Psihatm)_{(m,\ell)}^-=(\drout^{\leq 1}\mcq_{-2}^kT^j\errhatm)_{(m,\ell)}^-+O(\ell^{-2}\ubar^{-9-j}),
\end{align*}
and Proposition \ref{prop:salva}, this yields the following estimates in region $\{r_-\leq r\leq \rb\}\cap\{w\leq w_{\rb,\gamma}\}$: 
\begin{align}\label{eq:premiereIIell}
    |(\sin^2\theta \mcq_{-2}^k T^{j+2}\Psihatm)_{\ell,m}^-|&\lesssim\ell^{-2}\|\mcq_{-2}^{\leq k+1}T^{\leq j+2}(\sin^2\theta T^2\Psihatm)\|_{L^\infty(S(u,\ubar))}\lesssim\ell^{-2}|u|^{-9-j},\\
    |(\mcq_{-2}^k T^{j+1}\Psihatm)_{\ell,m}^-|&=|(\mcq_{-2}^k T^{j+1}\errhatm)_{\ell,m}^-|+O(\ell^{-2}\ubar^{-9-j})\\
    &\lesssim\ell^{-2}\|\mcq_{-2}^{\leq k+1}T^{\leq j+2}T\errhatm\|_{L^\infty(S(u,\ubar))}+\ell^{-2}\ubar^{-9-j}\lesssim\ell^{-2}|u|^{-8-j-\delta/2},\nn\\
    |(\mcq_{-2}^k T^{j+1}\drout\Psihatm)_{\ell,m}^-|&=|(\mcq_{-2}^k T^{j+1}\drout\errhatm)_{\ell,m}^-|+O(\ell^{-2}\ubar^{-9-j})\\
    &\lesssim\ell^{-2}\|\mcq_{-2}^{\leq k+1}T^{\leq j+2}T\drout\errhatm\|_{L^\infty(S(u,\ubar))}+\ell^{-2}\ubar^{-9-j}\lesssim\ell^{-2}|u|^{-8-j-\delta/2},\nn\\
    |(\mcq_{-2}^k T^{j+1}\Phi\Psihatm)_{\ell,m}^-|&=|(\mcq_{-2}^k T^{j+1}\Phi\errhatm)_{\ell,m}^-|+O(\ell^{-2}\ubar^{-9-j})\\
    &\lesssim\ell^{-2}\|\mcq_{-2}^{\leq k+1}T^{\leq j+2}T\Phi\errhatm|+\ell^{-2}\ubar^{-9-j}\lesssim\ell^{-2}|u|^{-8-j-\delta/2}.\nn
\end{align}
Moreover, using Proposition \ref{prop:modecoupling} again we get 
\begin{align}\label{eq:costhetamodeIIell}
    |(\cos\theta \mcq_{-2}^k T^{j+1}\Psihatm)_{\ell,m}^-|&=|(\cos\theta \mcq_{-2}^k T^{j+1}\errhatm)_{\ell,m}^-|+O(\ell^{-2}\ubar^{-9-j})\\
    &\lesssim\ell^{-2}\|\mcq_{-2}^{\leq k+1}T^{\leq j+2}(\cos\theta T\errhatm)\|_{L^\infty(S(u,\ubar))}+\ell^{-2}\ubar^{-9-j}\lesssim\ell^{-2}|u|^{-8-j-\delta/2}.
\end{align}
The Teukolsky equation writes, using the $(\partial_t,\partial_r)$ decomposition \eqref{eq:decteuk} of $\teukhat_s$, 
\begin{align*}
    \teukhat_{-2}^{[\partial_r]}\mcq_{-2}^kT^j\Psihatm&=-\teukhat_{-2}^{[\partial_t]}\mcq_{-2}^kT^j\Psihatm\\
    &=-a^2\sin^2\theta \mcq_{-2}^kT^{j+2}\Psihatm+2(r^2+a^2)\mcq_{-2}^kT^{j+1}\drout\Psihatm-2a\mcq_{-2}^kT^{j+1}\Phi\Psihatm\\
    &\quad-(6r+4ia\cos\theta)\mcq_{-2}^kT^{j+1}\Psihatm.
\end{align*}
Projecting onto an outgoing $(\ell,m)$ mode thus gives, using \eqref{eq:premiereIIell}--\eqref{eq:costhetamodeIIell},
\begin{align}\label{eq:aprojIIell}
    (\teukhat_{-2}^{[\partial_r]}\mcq_{-2}^kT^j\Psihatm)_{\ell,m}^-=O(\ell^{-2}|u|^{-8-j-\delta/2}).
\end{align}
Then, using the expression \eqref{eq:teukhatdr} of the operator $\teukhat_{-2}^{[\partial_r]}$, as well as \eqref{eq:dronddrond'spin} and Lemma \ref{lem:commproj}, \eqref{eq:aprojIIell} rewrites as \eqref{eq:odeIIell} for $|m|\leq \ell$, as stated.
\end{proof}
\begin{rem}\label{rem:maxderlost}
    Note that in order for the bounds \eqref{eq:premiereIIell}--\eqref{eq:costhetamodeIIell} to hold by Propositions \ref{prop:upperboundII} and \ref{prop:almostsharpdrout}, we need $N_k\geq 12$ and $N_j\geq 10$. This is where we lose the maximum number of derivatives in this paper. 
\end{rem}
\subsubsection{Resolution of the ODE \eqref{eq:odeIIell}}
In the setting of Proposition \ref{prop:edoII3ell}, let us consider values of $u\ll -1$ such that the corresponding level sets are in $\deux\cup\trois$. By \eqref{eq:odeIIell}, the ODE satisfied by $g_{m,\ell,u}=G_{m,\ell}(r,u)$ is 
\begin{align}\label{eq:odeodeIIell}
    \Delta g_{m,\ell,u}''(r)-(2iam+2(r-M))g_{m,\ell,u}'(r)+\lambda'_\ell g_{m,\ell,u}(r)=O(\ell^{-2}|u|^{-8-\delta/2}).
\end{align}
As before, we use a basis of solutions to the homogeneous problem 
\begin{align}\label{eq:odeHIIell}
    \Delta \widehat{v}''(r)-(2iam+2(r-M))\widehat{v}'(r)+\lambda'_\ell \widehat{v}(r)=0.
\end{align}
To solve \eqref{eq:odeHIIell}, notice that as in the case $\ell=3$, if $v$ satisfies \eqref{eq:odeHell}, then $\widehat{v}:=v\Delta^2e^{2imr_{mod}}$ satisfies \eqref{eq:odeHIIell}. This proves that a basis $({}^{(m,\ell)}\!\widehat{v_1},{}^{(m,\ell)}\!\widehat{v_2})$ of solutions  of \eqref{eq:odeHIIell} is
\begin{align}
    &{}^{(m,\ell)}\!\widehat{v_1}(r):={}^{(m,\ell)}\!v_1(r)\Delta^2(r)e^{2imr_{mod}(r)},\\
    &{}^{(m,\ell)}\!\widehat{v_2}(r):={}^{(m,\ell)}\!v_2(r)\Delta^2(r)e^{2imr_{mod}(r)}.
\end{align}
As before, by the variation of constants method, we find a particular solution $\widehat{p}_{m,\ell,u}(r)$ of \eqref{eq:odeodeIIell}:
\begin{align}\label{eq:pmucontent}
    \widehat{p}_{m,\ell,u}(r)=&-{}^{(m,\ell)}\!\widehat{v_1}(r)\int_{\rb}^r {}^{(m,\ell)}\!v_2(r')O(\ell^{-2}|u|^{-8-j-\delta/2})\dee r'\nn\\
    &+{}^{(m,\ell)}\!\widehat{v_2}(r)\int_{\rb}^{r}{}^{(m,\ell)}\!v_1(r')O(\ell^{-2}|u|^{-8-j-\delta/2})\dee r',
\end{align}
and we get constants $\widehat{A}_{m,\ell}(u)$, $\widehat{B}_{m,\ell}(u)$ such that in $\deux\cup\trois$ we have
\begin{align}
    {G_{m,\ell}(r,u)}=\widehat{A}_{m,\ell}(u){}^{(m,\ell)}\!\widehat{v_1}(r)+\widehat{B}_{m,\ell}(u){}^{(m,\ell)}\!\widehat{v_2}(r)+\widehat{p}_{m,\ell,u}(r).\label{eq:troizell}
\end{align}
\subsubsection{Computation of the constants $\widehat{A}_{m,\ell}(u)$, $\widehat{B}_{m,\ell}(u)$}
We have by \eqref{eq:troizell},
$$\widehat{A}_{m,\ell}(u){}^{(m,\ell)}\!\widehat{v_1}(\rb)=G_{m,\ell}(\rb,u)=\Delta^2(\rb)e^{2imr_{mod}(\rb)}{F_{m,\ell}}(\rb,\ubar_{\rb}(u))$$
which yields
\begin{align}\label{eq:cm1expell}
    \widehat{A}_{m,\ell}(u)&=\frac{{F_{m,\ell}}(\rb,\ubar_{\rb}(u))}{{}^{(m,\ell)}\!v_1(\rb)}=O_{(m,\ell)}(1).
\end{align}
We also have as before,
\begin{align*}
    \parr{G_{m,\ell}}(\rb,u)&=(4(\rb-M)+2iam)\Delta(\rb)e^{2imr_{mod}(\rb)}{F_{m,\ell}}(\rb,\ubar_{\rb}(u))\\
    &\quad+2(\rb^2+a^2)\Delta(\rb)e^{2imr_{mod}(\rb)}\parubar{F_{m,\ell}}(\rb,\ubar_{\rb}(u))+\Delta^2(\rb)e^{2imr_{mod}(\rb)}\parr{F_{m,\ell}}(\rb,\ubar_{\rb}(u)).
\end{align*}
Moreover, by \eqref{eq:troizell}, using ${}^{(m,\ell)}\!\widehat{v_2}(\rb)=0$, $$\parr{G_{m,\ell}}(\rb,u)=\widehat{A}_{m,\ell}(u){}^{(m,\ell)}\!\widehat{v_1}'(\rb)+\widehat{B}_{m,\ell}(u){}^{(m,\ell)}\!\widehat{v_2}'(\rb),$$
thus using $\partial_\ubar{F_{m,\ell}}(\rb,\ubar_{\rb}(u))=(\mcq_{-2}^kT^{j+1}\psihatm)_{(m,\ell)}(\rb,\ubar_{\rb}(u))$, this gives 
\begin{align*}
    \widehat{A}_{m,\ell}(u){}^{(m,\ell)}\!\widehat{v_1}'(\rb)+\widehat{B}_{m,\ell}&(u){}^{(m,\ell)}\!\widehat{v_2}'(\rb)=(4(\rb-M)+2iam)\Delta(\rb)e^{2imr_{mod}(\rb)}{F_{m,\ell}}(\rb,\ubar_{\rb}(u))\\
    &\quad+\Delta^2(\rb)e^{2imr_{mod}(\rb)}\partial_r{F_{m,\ell}}(\rb,\ubar_{\rb}(u))+O((\mcq_{-2}^kT^{j+1}\psihatm)_{(m,\ell)}(\rb,\ubar_{\rb}(u))).
\end{align*}
Moreover, we have by \eqref{eq:cm1expell},
$$\widehat{A}_{m,\ell}(u){}^{(m,\ell)}\!\widehat{v_1}'(\rb)=\Bigg[(4(\rb-M)+2iam)\Delta(\rb)e^{2imr_{mod}(\rb)}+\Delta^2(\rb)e^{2imr_{mod}(\rb)}\frac{{}^{(m,\ell)}\!v_1'(\rb)}{{}^{(m,\ell)}\!v_1(\rb)}\Bigg]{F_{m,\ell}}(\rb,\ubar_{\rb}(u)),$$
thus using \eqref{eq:derrHmell} and \eqref{eq:yeyeyeell}, we get
\begin{align*}
    \widehat{B}_{m,\ell}(u){}^{(m,\ell)}\!\widehat{v_2}'(\rb)&=\Delta^2(\rb)e^{2imr_{mod}(\rb)}\Bigg[\parr{F_{m,\ell}}(\rb,\ubar_{\rb}(u))-\frac{{}^{(m,\ell)}\!v_1'(\rb)}{{}^{(m,\ell)}\!v_1(\rb)}{F_{m,\ell}}(\rb,\ubar_{\rb}(u))\Bigg]+O(\ell^{-2}|u|^{-8-j-\delta/2})\\
    &=\frac{\Delta^{-1}(\rb)}{{}^{(m,\ell)}\!v_1(\rb)}\int_{r_+}^{\rb}\Delta^2(r')e^{2imr_{mod}(r')}{}^{(m,\ell)}\!v_1(r')O(\ell^{-2}(2(r')^*-u)^{-8-j-\delta})\dee r'\\
    &\quad\quad\quad +O((\mcq_{-2}^kT^{j+1}\psihatm)_{(m,\ell)}(\rb,\ubar_{\rb}(u))).
\end{align*}
Finally, using ${}^{(m,\ell)}\!\widehat{v_2}'(\rb)=\Delta^{-1}(\rb)/{}^{(m,\ell)}\!v_1(\rb)$, we find  
\begin{align}\label{eq:tuuuutu}
    \widehat{B}_{m,\ell}(u)&=\int_{r_+}^{\rb}\Delta^2(r')e^{2imr_{mod}(r')}{}^{(m,\ell)}\!v_1(r')O(\ell^{-2}(2(r')^*-u)^{-8-j-\delta})\dee r'\nn\\
    &\quad\quad+O({}^{(m,\ell)}\!v_1(\rb)(\mcq_{-2}^kT^{j+1}\psihatm)_{(m,\ell)}(\rb,\ubar_{\rb}(u)))\nn\\
    &=\int_{r_+}^{\rb}\Delta^2(r')e^{2imr_{mod}(r')}{}^{(m,\ell)}\!v_1(r')O(\ell^{-2}|u|^{-8-j-\delta})\dee r'\nn\\
    &\quad\quad+O({}^{(m,\ell)}\!v_1(\rb)(\mcq_{-2}^kT^{j+1}\psihatm)_{(m,\ell)}(\rb,\ubar_{\rb}(u))).
\end{align}
We can now estimate the different terms in \eqref{eq:troizell}. First, by \eqref{eq:cm1expell} we have 
\begin{align}
    \left|\widehat{A}_{m,\ell}(u){}^{(m,\ell)}\!\widehat{v_1}(r)\right|&\lesssim_{(m,\ell)}r-r_-.\label{eq:conc3}
\end{align}
Next, the first term in the RHS of \eqref{eq:pmucontent} can be bounded using Proposition \ref{prop:v1ell}, \ref{item:(d)}, as
\begin{align}
    \left|{}^{(m,\ell)}\!\widehat{v_1}(r)\int_{\rb}^r {}^{(m,\ell)}\!v_2(r')O(\ell^{-2}|u|^{-8-j-\delta/2})\dee r'\right|\lesssim_{(m,\ell)} \Delta^2\int_{r}^{\rb}\frac{\dee r'}{(r-r_-)^2}+r-r_-\lesssim_{(m,\ell)}r-r_-.\label{eq:conc2}
\end{align}
Finally, the last term in the RHS of \eqref{eq:troizell}, which is the sum of $\widehat{B}_{m,\ell}(u){}^{(m,\ell)}\!\widehat{v_2}(r)$ with the second line in \eqref{eq:pmucontent}, can be rewritten as 
\begin{align*}
    &{}^{(m,\ell)}\!\widehat{v_2}(r)\int_{r_+}^{r}{}^{(m,\ell)}\!v_1(r')O(\ell^{-2}|u|^{-8-j-\delta/2})\dee r'+O({}^{(m,\ell)}\!v_1(\rb)(\mcq_{-2}^kT^{j+1}\psihatm)_{(m,\ell)}(\rb,\ubar_{\rb}(u)))\\
    &=\frac{{}^{(m,\ell)}\!v_1(r_-)^{-1}}{(4(M-r_-)-2iam)}\left(\int_{r_+}^{r}{}^{(m,\ell)}\!v_1(r')O(\ell^{-2}|u|^{-8-j-\delta/2})\dee r'+O({}^{(m,\ell)}\!v_1(\rb)(\mcq_{-2}^kT^{j+1}\psihatm)_{(m,\ell)}(\rb,\ubar_{\rb}(u)))\right)\\
    &\quad\quad+O_{(m,\ell)}(r-r_-)
\end{align*}
by the expression \eqref{eq:tuuuutu} of $\widehat{B}_{m,\ell}(u){}^{(m,\ell)}$, and by Proposition \ref{prop:v1ell}, \ref{item:(d)}. Finally, by Proposition \ref{prop:v1ell}, \ref{item:(iv)}, we have the bound $|{}^{(m,\ell)}\!v_1(r')|/|{}^{(m,\ell)}\!v_1(r_-)|\leq 1$, thus
\begin{align}
    &\left|\frac{{}^{(m,\ell)}\!v_1(r_-)^{-1}}{(4(M-r_-)-2iam)}\left(\int_{r_+}^{r}{}^{(m,\ell)}\!v_1(r')O(\ell^{-2}|u|^{-8-j-\delta/2})\dee r'+O({}^{(m,\ell)}\!v_1(\rb)(\mcq_{-2}^kT^{j+1}\psihatm)_{(m,\ell)}(\rb,\ubar_{\rb}(u)))\right)\right|\nn\\
    &\lesssim \ell^{-2}|u|^{-8-j-\delta/2}+|(\mcq_{-2}^kT^{j+1}\errhatm)_{(m,\ell)}(\rb,\ubar_{\rb}(u))|,\label{eq:conc1}
\end{align}
where the implicit constant in the bound above does not depend on $(m,\ell)$, and where we used
$$|(\mcq_{-2}^kT^{j+1}\psihatm)_{(m,\ell)}(\rb,\ubar_{\rb}(u)))|\lesssim |(\mcq_{-2}^kT^{j+1}\errhatm)_{(m,\ell)}(\rb,\ubar_{\rb}(u)))|+\ubar^{-9-j},$$
where the term $\ubar^{-9-j}$ is non-zero only for $\ell=4$ by Proposition \ref{prop:modecoupling}, as $k\in\{0,1\}$. Finally, combining \eqref{eq:troizell}, \eqref{eq:conc3}, \eqref{eq:conc2}, \eqref{eq:conc1} yields the bound 
$$G_{m,\ell}(r,u)=O(\ell^{-2}|u|^{-8-j-\delta/2})+O(|(\mcq_{-2}^kT^{j+1}\errhatm)_{(m,\ell)}(\rb,\ubar_{\rb}(u))|)+O_{(m,\ell)}(r-r_-).$$
We deduce that on $\ch$, 
$$|G_{m,\ell}(r_-,u)|\lesssim\ell^{-2}|u|^{-8-j-\delta/2}+|(\mcq_{-2}^kT^{j+1}\errhatm)_{(m,\ell)}(\rb,\ubar_{\rb}(u))|,$$
where the implicit constant does not depend on $(m,\ell)$, and thus by \eqref{eq:definedG},
\begin{align}
    \|(\mcq_{-2}^kT^{j}\Psihatm)_{\ell\geq 4}\|_{L^2(S(u,+\infty))}^2&\lesssim \sum_{\ell\geq 4}\sum_{|m|\leq\ell}\Big(\ell^{-4}|u|^{-2(8+j+\delta/2)}+|(\mcq_{-2}^kT^{j+1}\errhatm)_{(m,\ell)}(\rb,\ubar_{\rb}(u))|^2\Big)\nn\\
    &\lesssim |u|^{-2(8+j+\delta/2)}+\|\mcq_{-2}^kT^{j+1}\errhatm\|_{L^2(S(u,\ubar_{\rb}(u)))}^2\nn\\
    &\lesssim |u|^{-2(8+j+\delta/2)}\label{eq:finalesti}
\end{align}
on $\ch$, by \eqref{eq:borneI}. This yields, by the Sobolev embedding \eqref{eq:sobolev},
\begin{align}\label{eq:lastlast}
    |(\Psihatm)_{\ell\geq 4}(r_-,u,\theta,\phi_-)|&\lesssim\|\mcq_{-2}T^{\leq 2}(\Psihatm)_{\ell\geq 4}\|_{L^2(S(u,+\infty))}\nn\\
    &\lesssim\|(\mcq_{-2}T^{\leq 2}\Psihatm)_{\ell\geq 4}\|_{L^2(S(u,+\infty))}+|u|^{-8-\delta/2} \nn\\
    &\lesssim |u|^{-8-\delta/2} 
\end{align}
on $\ch$. Note that here we used for the second bound the computation on $\ch$ (where $\Psihatm=\errhatm$),
\begin{align*}
    [\mcq_{-2}T^{\leq 2},\mathbf{P}_{\ell\geq 4}]\psihatm&=[\mcq_{-2},\mathbf{P}_{\ell\geq 4}]T^{\leq 2}\errhatm\\
    &=[a^2\sin^2\theta T^2+4ia\cos\theta T,\mathbf{P}_{\ell\geq 4}]T^{\leq 2}\errhatm\\
    &=[a^2\sin^2\theta,\mathbf{P}_{\ell\geq 4}]T^2T^{\leq 2}\errhatm +4ia[\cos\theta,\mathbf{P}_{\ell\geq 4}]TT^{\leq 2}\errhatm
\end{align*}
which is bounded by $|u|^{-8-\delta/2}$ in $L^2(S(u,+\infty))$ by Propositions \ref{prop:upperboundII} and \ref{prop:almostsharpdrout}. By \eqref{eq:lastlast}, to conclude the proof of Theorem \ref{thm:decayhigher}, it is only left to prove the following bound, in $\deux\cup\trois$: 
$$|(\Psihatm)_{\ell\geq 4}(r,u,\theta,\phi_-)-(\Psihatm)_{\ell\geq 4}(r_-,u,\theta,\phi_-)|\lesssim r-r_-.$$
To this end, we use Propositions \ref{prop:upperboundII} and \ref{prop:almostsharpdrout} again, from which we deduce in $\deux\cup\trois$
\begin{align*}
    |\drout(\Psihatm)_{\ell\geq 4}|&\lesssim\|\mcq_{-2}T^{\leq 2}\drout(\Psihatm)_{\ell\geq 4}\|_{L^2(S(u,\ubar))}\\
    &\lesssim \|(\drout^{\leq 1}\mcq_{-2}T^{\leq 2}\Psihatm)_{\ell\geq 4}\|_{L^2(S(u,\ubar))}+|u|^{-8-\delta/2} \\
    &\lesssim \|\drout^{\leq 1}\mcq_{-2}T^{\leq 2}\Psihatm\|_{L^2(S(u,\ubar))}+|u|^{-8-\delta/2}  \lesssim 1.
\end{align*}
Integrating the bound above in $\deux\cup\trois$ from $\ch$, and using \eqref{eq:lastlast} on $\ch$, we obtain the bound 
$$|(\Psihatm)_{\ell\geq 4}(r,u,\theta,\phi_-)|\lesssim |u|^{-8-\delta/2}+\int_{r_-}^r\dee r'\lesssim |u|^{-8-\delta/2}+r-r_-, $$
which concludes the proof of Theorem \ref{thm:decayhigher}.
\appendix
\section{Computations for the energy method}
\subsection{Integration by parts on the spheres}\label{appendix:ippsphII} 
\begin{lem}\label{lem:intS}
    Equation \eqref{eq:unefois} can be rewritten as follows,
\begin{align*} \partial_\ubar\left(\int_{S(u,\ubar)}\mathbf{F}_\ubar[\Psihat]\dee\nu\right)+\partial_u\left(\int_{S(u,\ubar)}\mathbf{F}_u[\Psihat]\dee\nu\right)+\int_{S(u,\ubar)}\mathbf{B}[\Psihat]\dee\nu=\intS\mu\Real\left(\overline{X(\Psihat)}O(\ubar^{-\beta})\right)\dee\nu,
\end{align*}
where 
$$\mathbf{F}_\ubar[\Psihat]:=2(r^2+a^2)f(r)|\ethat\Psihat|^2-\frac{1}{2}\mu g(r)(|\partial_\theta\Psihat|^2+|\mcu\Psihat|^2)+a\sin\theta \mu \mathfrak{R}\left(\overline{X(\Psihat)}\mcu\Psihat)\right),$$
$$\mathbf{F}_u[\Psihat]:=2(r^2+a^2)g(r)|\partial_\ubar\Psihat|^2-\frac{1}{2}\mu f(r)(|\partial_\theta\Psihat|^2+|\mcu\Psihat|^2)-a\sin\theta \mu \mathfrak{R}\left(\overline{X(\Psihat)}\mcu\Psihat)\right),$$
and the bulk term is given by
\begin{align*}
    \mathbf{B}[\Psihat]&=2(r\mu f(r)-\partial_\ubar((r^2+a^2)f(r)))|\ethat\Psihat|^2+2(r\mu g(r)-\partial_u((r^2+a^2)g(r)))|\partial_\ubar\Psihat|^2\\
    &\quad +\frac{1}{2}(\partial_u(\mu f(r))+\partial_\ubar(\mu g(r)))(|\partial_\theta\Psihat|^2+|\mcu\Psihat|^2)\\
    &\quad+\mu g(r)\frac{s\mu ra^2\sin\theta\cos\theta}{(r^2+a^2)^2}\mathfrak{I}(\overline{\Psihat} \mcu\Psihat)\\
    &\quad+4\mu g(r)as\cos\theta\mathfrak{I}(\overline{\partial_\ubar\Psihat}T\Psihat)-4sr\mu g(r)\mathfrak{R}(\overline{\partial_\ubar\Psihat}T\Psihat)+4g(r)cs(r-M)|\partial_\ubar\Psihat|^2\\
    &\quad+\frac{2g(r)ar\mu}{r^2+a^2}\mathfrak{R}(\overline{\partial_\ubar\Psihat}\Phi\Psihat)+g(r)\left(\mu s-s^2\mu\frac{a^4\sin^2\theta\cos^2\theta}{(r^2+a^2)^2}+\mu V\right)\mathfrak{R}(\overline{\partial_\ubar\Psihat}\Psihat)\\
    &\quad-2sg(r)\mu\frac{a^2\sin\theta\cos\theta}{(r^2+a^2)}\mathfrak{I}(\overline{\partial_\ubar\Psihat}\mcu\Psihat)+2r\mu g(r)\Real(\overline{\partial_\ubar\Psihat} \ethat\Psihat)\\
    &\quad+\mu f(r)\frac{s\mu ra^2\sin\theta\cos\theta}{(r^2+a^2)^2}\mathfrak{I}(\overline{\Psihat} \mcu\Psihat)-2sf(r)\mu\frac{a^2\sin\theta\cos\theta}{(r^2+a^2)}\mathfrak{I}(\overline{ \ethat \Psihat}\mcu\Psihat)\\
    &\quad+4\mu f(r)as\cos\theta\mathfrak{I}(\overline{\ethat\Psihat}T\Psihat)-4sr\mu f(r)\mathfrak{R}(\overline{\ethat\Psihat}T\Psihat)+f(r)(2r\mu+4cs(r-M))\mathfrak{R}(\overline{\ethat\Psihat}\partial_\ubar\Psihat)\\
    &\quad-\frac{2f(r)ar\mu}{r^2+a^2}\mathfrak{R}(\overline{ \ethat \Psihat}\Phi\Psihat)+f(r)\left(\mu s-s^2\mu\frac{a^4\sin^2\theta\cos^2\theta}{(r^2+a^2)^2}+\mu V\right)\mathfrak{R}(\overline{ \ethat \Psihat}\Psihat).
    \end{align*}
\end{lem}
\begin{proof}We compute 
$$\intS \Real\left(f(r)\overline{\ethat\Psihat}\mu\teukhat_{s}^{(c,V)}\Psihat\right)+\Real\left(g(r)\overline{\partial_\ubar\Psihat}\mu\teukhat_{s}^{(c,V)}\Psihat\right)\dee\nu$$
by integration by parts on $S(u,\ubar)$. First, using \eqref{eq:ladeuz} we get 
\begin{align*}
    \intS \Real\Big(f&(r)\overline{\ethat\Psihat}\mu\teukhat_s^{(c,V)}\Psihat\Big)\dee\nu\\
    =&4(r^2+a^2)f(r)\int_{S(u,\ubar)}\mathfrak{R}(\overline{\ethat\Psihat} \partial_\ubar \ethat\Psihat)\:\dee\nu+\mu f(r)\int_{S(u,\ubar)}\mathfrak{R}(\overline{\ethat\Psihat} \mcu^2\Psihat)\:\dee\nu\\
    &+\mu f(r)\int_{S(u,\ubar)}\mathfrak{R}\left(\overline{\ethat\Psihat}\frac{1}{\sin\theta}\partial_\theta(\sin\theta\partial_\theta\Psihat)\right)\dee\nu+4\mu f(r) as \int_{S(u,\ubar)}\cos\theta\mathfrak{I}(\overline{\ethat\Psihat} T\Psihat)\dee\nu\\
    &+f(r)(2r\mu+4s(r-M)) \int_{S(u,\ubar)}\mathfrak{R}(\overline{\ethat\Psihat} \partial_\ubar \Psihat)\dee\nu+ 2f(r)r\mu\int_{S(u,\ubar)}|\ethat\Psihat|^2\dee\nu\\
    &-4sr\mu f(r)\int_{S(u,\ubar)}\mathfrak{R}(\overline{\ethat\Psihat}T\Psihat)\dee\nu-\frac{2f(r)ar\mu }{r^2+a^2}\int_{S(u,\ubar)}\mathfrak{R}(\overline{\ethat\Psihat}\Phi\Psihat)\dee\nu\\
    &+f(r)\left(\mu s-s^2\mu\frac{a^4\sin^2\theta\cos^2\theta}{(r^2+a^2)^2}\right)\int_{S(u,\ubar)}\mathfrak{R}(\overline{\ethat\Psihat}\Psihat)\dee\nu\\
    &-2sf(r)\mu\frac{a^2\sin\theta\cos\theta}{(r^2+a^2)}\int_{S(u,\ubar)}\mathfrak{I}(\overline{\ethat\Psihat}\mcu\Psihat)\dee\nu.
\end{align*}
We have
\begin{align*}
    4(r^2+a^2)f(r)&\int_{S(u,\ubar)}\mathfrak{R}(\overline{\ethat\Psihat} \partial_\ubar \ethat\Psihat)\:\dee\nu=2(r^2+a^2)f(r)\partial_\ubar\left(\int_{S(u,\ubar)} |\ethat\Psihat|^2\:\dee\nu\right)\\
    &=\partial_\ubar\left(2(r^2+a^2)f(r)\int_{S(u,\ubar)} |\ethat\Psihat|^2\:\dee\nu\right)-2\partial_\ubar((r^2+a^2)f(r))\int_{S(u,\ubar)} |\ethat\Psihat|^2\:\dee\nu.
\end{align*}
Next, using Lemma \ref{lem:U}, as well as 
$$[\ethat,\mcu]=\frac{1}{2}\mu\partial_r\left(\frac{\Sigma}{r^2+a^2}\right)=\frac{is\mu ra^2\sin\theta\cos\theta}{(r^2+a^2)^2},$$
we get, in view of $T=\partial_\ubar-\partial_u$,
\begin{align*}
    \mu f(r)\int_{S(u,\ubar)}\mathfrak{R}&(\overline{\ethat\Psihat} \mcu^2\Psihat)\:\dee\nu\\
    &=T\left(\mu f(r)\int_{S(u,\ubar)}a\sin\theta\mathfrak{R}(\overline{\ethat\Psihat}\mcu\Psihat)\:\dee\nu\right)-\mu f(r)\int_{S(u,\ubar)}\mathfrak{R}(\overline{\mcu\ethat\Psihat} \mcu\Psihat)\:\dee\nu\\
    &=(\partial_\ubar-\partial_u)\left(\mu f(r)\int_{S(u,\ubar)}a\sin\theta\mathfrak{R}(\overline{\ethat\Psihat}\mcu\Psihat)\:\dee\nu\right)-\frac{1}{2}\partial_u\left(\mu f(r)\int_{S(u,\ubar)}|\mcu\Psihat|^2\:\dee\nu\right)\\
    &\quad\quad+\frac{1}{2}\partial_u(\mu f(r))\int_{S(u,\ubar)}|\mcu\Psihat|^2\:\dee\nu+\mu f(r)\frac{s\mu ra^2\sin\theta\cos\theta}{(r^2+a^2)^2}\int_{S(u,\ubar)}\mathfrak{I}(\overline{\Psihat} \mcu\Psihat)\:\dee\nu,
\end{align*}
where we used $\partial_u=\ethat+a/(r^2+a^2)\Phi$. We also have, using $[\ethat,\partial_\theta]=0$,
\begin{align*}
    \mu f(r)\int_{S(u,\ubar)}\mathfrak{R}\Bigg(\overline{\ethat\Psihat}\frac{1}{\sin\theta}\partial_\theta&(\sin\theta\partial_\theta\Psihat)\Bigg)\dee\nu\\
    &=-\mu f(r)\int_{S(u,\ubar)}\mathfrak{R}\left(\partial_\theta\overline{\ethat\Psihat}\partial_\theta\Psihat\right)\dee\nu\\
    &=-\frac{1}{2}\mu f(r)\int_{S(u,\ubar)}\ethat(|\partial_\theta\Psihat|^2)\dee\nu\\
    &=\partial_u\left(-\frac{1}{2}\mu f(r)\int_{S(u,\ubar)}|\partial_\theta\Psihat|^2\dee\nu\right)+\frac{1}{2}\ethat(\mu f(r))\int_{S(u,\ubar)}|\partial_\theta\Psihat|^2\dee\nu,
\end{align*}
and we put the remaining terms in the bulk term $\mathbf{B}[\Psihat]$. Now using the expression \eqref{eq:laprem} we have
\begin{align*}
    \intS\Real(g(&r)\overline{\partial_\ubar\Psihat}\mu\teukhat_s^{(c,V)}\Psihat)\dee\nu\\
    =&4(r^2+a^2)g(r)\int_{S(u,\ubar)}\mathfrak{R}(\overline{\partial_\ubar\Psihat}\ethat \partial_\ubar\Psihat)\:\dee\nu+\mu g(r)\int_{S(u,\ubar)}\mathfrak{R}(\overline{\partial_\ubar\Psihat} \mcu^2\Psihat)\:\dee\nu\\
    &+\mu g(r)\int_{S(u,\ubar)}\mathfrak{R}\left(\overline{\partial_\ubar\Psihat}\frac{1}{\sin\theta}\partial_\theta(\sin\theta\partial_\theta\Psihat)\right)\dee\nu+4\mu g(r) as \int_{S(u,\ubar)}\cos\theta\mathfrak{I}(\overline{\partial_\ubar\Psihat} T\Psihat)\dee\nu\\
    &+g(r)(2r\mu+4s(r-M)) \int_{S(u,\ubar)}|\partial_\ubar\Psihat|^2\dee\nu+ 2g(r)r\mu\int_{S(u,\ubar)}\mathfrak{R}(\overline{\partial_\ubar\Psihat}\ethat\Psihat)\dee\nu\\
    &-4sr\mu g(r)\int_{S(u,\ubar)}\mathfrak{R}(\overline{\partial_\ubar\Psihat}T\Psihat)\dee\nu+\frac{2g(r)ar\mu }{r^2+a^2}\int_{S(u,\ubar)}\mathfrak{R}(\overline{\partial_\ubar\Psihat}\Phi\Psihat)\dee\nu\\
    &+g(r)\left(\mu s-s^2\mu\frac{a^4\sin^2\theta\cos^2\theta}{(r^2+a^2)^2}\right)\int_{S(u,\ubar)}\mathfrak{R}(\overline{\partial_\ubar\Psihat}\Psihat)\dee\nu\\
    &-2sg(r)\mu\frac{a^2\sin\theta\cos\theta}{(r^2+a^2)}\int_{S(u,\ubar)}\mathfrak{I}(\overline{\partial_\ubar\Psihat}\mcu\Psihat)\dee\nu.
\end{align*}
As before, we have:
\begin{align*}
    4(r^2+a^2)g(r)&\int_{S(u,\ubar)}\mathfrak{R}(\overline{\partial_\ubar\Psihat}\ethat \partial_\ubar\Psihat)\:\dee\nu=2(r^2+a^2)g(r)\partial_u\left(\int_{S(u,\ubar)} |\partial_\ubar\Psihat|^2\:\dee\nu\right)\\
    &=\partial_u\left(2(r^2+a^2)g(r)\int_{S(u,\ubar)} |\partial_\ubar\Psihat|^2\:\dee\nu\right)-2\partial_u((r^2+a^2)g(r))\int_{S(u,\ubar)} |\partial_\ubar\Psihat|^2\:\dee\nu,
\end{align*}

\begin{align*}
    \mu g(r)\int_{S(u,\ubar)}\mathfrak{R}\Bigg(\overline{\partial_\ubar\Psihat}\frac{1}{\sin\theta}&\partial_\theta(\sin\theta\partial_\theta\Psihat)\Bigg)\dee\nu\\
    &=\partial_\ubar\left(-\frac{1}{2}\mu g(r)\int_{S(u,\ubar)}|\partial_\theta\Psihat|^2\dee\nu\right)+\frac{1}{2}\partial_\ubar(\mu g(r))\int_{S(u,\ubar)}|\partial_\theta\Psihat|^2\dee\nu,
\end{align*}
\begin{align*}
    \mu g(r)\int_{S(u,\ubar)}&\mathfrak{R}(\overline{\partial_\ubar\Psihat} \mcu^2\Psihat)\:\dee\nu\\
    &=(\partial_\ubar-\partial_u)\left(\mu g(r)\int_{S(u,\ubar)}a\sin\theta\mathfrak{R}(\overline{\partial_\ubar\Psihat}\mcu\Psihat)\:\dee\nu\right)-\frac{1}{2}\partial_\ubar\left(\mu g(r)\int_{S(u,\ubar)}|\mcu\Psihat|^2\:\dee\nu\right)\\
    &\quad\quad+\frac{1}{2}\partial_\ubar(\mu g(r))\int_{S(u,\ubar)}|\mcu\Psihat|^2\:\dee\nu+\mu g(r)\frac{s\mu ra^2\sin\theta\cos\theta}{(r^2+a^2)^2}\int_{S(u,\ubar)}\mathfrak{I}(\overline{\Psihat} \mcu\Psihat)\:\dee\nu.
\end{align*}
The remaining terms are put in the bulk term, which concludes the proof of Lemma \ref{lem:intS}.
\end{proof}
\subsection{Coercivity of the bulk term} \label{appendix:bulkII}    
\begin{lem}\label{lem:bulkpos}
    For $\alpha>1$, $s=-2$, $c>0$, $V\in\mathbb{R}$, for $p(a,M)>0$ large enough, and $r_\mathfrak{b}(p,a,M)$ close enough to $r_-$, we have in $\{r_-\leq r\leq\rb\}$,
\begin{align}\label{eq:bulkposapp}
    \intS\mathbf{B}[\Psihat]\dee\nu\gtrsim(-\mu)\intS\ener_{\alpha+1}[\Psihat]\dee\nu.
\end{align}
\end{lem}
\begin{proof}
As in \cite{scalarMZ}, we compute 
\begin{align*}
    r \mu g(r)-\partial_u((r^2+a^2) g(r))  &=-\mu p r\left(r^2+a^2\right)^p, \\
r \mu f(r)-\partial_\ubar((r^2+a^2) f(r))  &=-\alpha(r-M-r\mu)|{\log(-\mu)}|^{-\alpha-1},\\
    \ethat(\mu f(r))+\partial_\ubar(\mu g(r))&=(r^2+a^2)^{-1}(r-M-r\mu)\mu(1+\alpha|{\log(-\mu)}|^{-1})|{\log(-\mu)}|^{-\alpha}\\
    &\quad+\mu (r-M+(p-1)r\mu)(r^2+a^2)^{p-1}.
\end{align*}
Thus, taking $\rb$ sufficiently close to $r_-$ such that 
$$M-r+r\mu\gtrsim 1,\quad r\in [r_-, \rb],$$
we then have 
\begin{align*}
r \mu f(r)-\partial_\ubar((r^2+a^2) f(r))  &\gtrsim |{\log(-\mu)}|^{-\alpha-1},\\
    \ethat(\mu f(r))+\partial_\ubar(\mu g(r))&\gtrsim -\mu|{\log(-\mu)}|^{-\alpha}-\mu (r^2+a^2)^{p}.
\end{align*}
We now define the principal bulk term as
\begin{align*}
    {\mathbf{B}}_{pr}[\Psihat]:=2(r\mu g(&r)-\partial_u((r^2+a^2)g(r)))|\partial_\ubar\Psihat|^2+2(r\mu f(r)-\partial_\ubar((r^2+a^2)f(r)))|\ethat\Psihat|^2\\
    &\quad +\dfrac{1}{2}(\partial_u(\mu f(r))+\partial_\ubar(\mu g(r)))(|\partial_\theta\Psihat|^2+|\mcu\Psihat|^2).
\end{align*}
We have shown that in $\{r_-\leq r\leq\rb\}$, 
\begin{align}
    {\mathbf{B}}_{pr}[\Psihat]\gtrsim (-\mu)p&(r^2+a^2)^p|\partial_\ubar\Psihat|^2+|{\log(-\mu)}|^{-\alpha-1}|\ethat\Psihat|^2\nonumber\\
    &\quad +(-\mu)(|{\log(-\mu)}|^{-\alpha}+(r^2+a^2)^{p})(|\partial_\theta\Psihat|^2+|\mcu\Psihat|^2).\label{eq:aobserver}
\end{align}
The only thing left to prove is that we can take $p$ large enough and $\rb$ sufficiently close to $r_-$ such that ${\mathbf{B}}[\Psihat]-{\mathbf{B}}_{pr}[\Psihat]$ can be absorbed in ${\mathbf{B}}_{pr}[\Psihat]$ after integrating on ${S(u,\ubar)}$. This is due to the following mix between weighted Cauchy-Schwarz of the type $|ab|\leq({\varepsilon a^2+\varepsilon^{-1}b^2})/{2}$ and the Poincaré inequality \eqref{eq:poincare}. We first deal with the terms with a factor $g(r)$. We have the bounds:
\begin{align*}
    \Bigg|\intS \mu g(r)\frac{sra^2\mu\cos\theta\sin\theta}{(r^2+a^2)^2}\mathfrak{I}(\overline{\Psihat} \mcu\Psihat)\dee\nu\Bigg|&\lesssim (-\mu)^2(r^2+a^2)^p\left(\intS\ener_{deg}[\Psihat]\dee\nu+\intS\hspace{-0.3cm}|\mcu\Psihat|^2\dee\nu\right)\\
    &\lesssim(-\mu)^2(r^2+a^2)^p\intS\ener_{deg}[\Psihat]\dee\nu,
\end{align*}
\begin{align*}
    \Bigg|\intS4\mu g(r)as\cos\theta\mathfrak{I}&(\overline{\partial_\ubar\Psihat}T\Psihat)-4srg(r)\mu\mathfrak{R}(\overline{\partial_\ubar\Psihat}T\Psihat)\dee\nu\Bigg|\lesssim\\
&(-\mu)(r^2+a^2)^p\varepsilon^{-1}\intS|\partial_\ubar\Psihat|^2\dee\nu+(-\mu)(r^2+a^2)^p\varepsilon\intS\ener_{deg}[\Psihat]\dee\nu,
\end{align*}
where we used \eqref{eq:expreT} which yields $T=O(\mu)e_3+O(1)\partial_\ubar+O(1)\mcu+O(1).$ We continue with 
$$|2r\mu g(r)\mathfrak{R}(\overline{\partial_\ubar\Psihat}\ethat\Psihat)|\lesssim(-\mu)(r^2+a^2)^p|\ethat\Psihat|^2+(-\mu)(r^2+a^2)^p|\partial_\ubar\Psihat|^2,$$
\begin{align*}
    \Bigg|\intS\dfrac{2g(r)ar\mu}{r^2+a^2}&\mathfrak{R}(\overline{\partial_\ubar\Psihat}\Phi\Psihat)\dee\nu\Bigg|\lesssim\\
    &(-\mu)(r^2+a^2)^p\varepsilon^{-1}\intS|\partial_\ubar\Psihat|^2\dee\nu+(-\mu)(r^2+a^2)^p\varepsilon\intS\ener_{deg}[\Psihat]\dee\nu,
\end{align*}
where we used \eqref{eq:exprephi} which yields $\Phi=O(\mu)e_3+O(1)\partial_\ubar+O(1)\mcu+O(1)$. Next, we bound 
\begin{align*}
    \Bigg|\intS g(r)\Bigg(\mu s-s^2\mu&\frac{a^4\sin^2\theta\cos^2\theta}{(r^2+a^2)^2}+\mu V\Bigg)\mathfrak{R}(\overline{\partial_\ubar\Psihat}\Psihat)\dee\nu\Bigg|\lesssim\\
    (-\mu)(r^2+a^2)^p&\varepsilon^{-1}\intS|\partial_\ubar\Psihat|^2\dee\nu+(-\mu)(r^2+a^2)^p\varepsilon\intS\ener_{deg}[\Psihat]\dee\nu,\\
    \Bigg|2sg(r)\mu\dfrac{a^2\sin\theta\cos\theta}{(r^2+a^2)}\mathfrak{I}(\overline{\partial_\ubar\Psihat}&\mcu\Psihat)\Bigg|\lesssim(-\mu)(r^2+a^2)^p\varepsilon^{-1}|\partial_\ubar\Psihat|^2+(-\mu)(r^2+a^2)^p\varepsilon|\mcu\Psi|^2.
\end{align*}

The only term with a factor $g(r)$ in the bulk term that remains to be controlled is $$4g(r)cs(r-M)|\partial_\ubar\Psihat|^2.$$ Note that with the choice of negative spin $s=-2$, for $c>0$ and $\rb\in(r_-,M)$ this term is non-negative in $\{r_-\leq r\leq\rb\}$, and even satisfies
\begin{align}\label{eq:positive}
    4g(r)cs(r-M)|\partial_\ubar\Psihat|^2\gtrsim(r^2+a^2)^p|\partial_\ubar\Psihat|^2.
\end{align}
Next, we control the terms of the bulk term with a factor $f(r)$. We have 
$$\Bigg|\intS\mu f(r)\frac{sra^2\mu\cos\theta\sin\theta}{(r^2+a^2)^2}\mathfrak{I}(\overline{\Psihat} \mcu\Psihat)\dee\nu\Bigg|\lesssim(-\mu)^2|{\log(-\mu)}|^{-\alpha}\intS\ener_{deg}[\Psihat]\dee\nu,$$
$$|2r\mu f(r)\mathfrak{R}(\overline{\ethat\Psihat}\partial_\ubar\Psihat)|\lesssim(-\mu)^2|\partial_\ubar\Psihat|^2+|{\log(-\mu)}|^{-2\alpha}|\ethat\Psihat|^2.$$
As before, using \eqref{eq:expreT} and \eqref{eq:exprephi} gives
\begin{align}
    \Bigg|\intS4\mu f(r)as\cos\theta\mathfrak{I}(\overline{\ethat\Psihat}T\Psihat)-4srf(r)&\mu\mathfrak{R}(\overline{\ethat\Psihat}T\Psihat)\dee\nu\Bigg|\nn\\
    &\lesssim(-\mu)^2\intS\ener_{deg}[\Psihat]\dee\nu+|{\log(-\mu)}|^{-2\alpha}\intS|\ethat\Psihat|^2\dee\nu,\nn\\
    \Bigg|\intS-\dfrac{2f(r)ar\mu}{r^2+a^2}\mathfrak{R}(\overline{\ethat\Psihat}\Phi\Psihat)\dee\nu\Bigg|&\lesssim(-\mu)^2\intS\ener_{deg}[\Psihat]\dee\nu+|{\log(-\mu)}|^{-2\alpha}\intS|\ethat\Psihat|^2\dee\nu,\nn
\end{align}
\begin{align}
    \Bigg|\intS f(r)\Bigg(\mu s-s^2\mu\dfrac{a^4\sin^2\theta\cos^2\theta}{(r^2+a^2)^2}&+\mu V\Bigg)\mathfrak{R}(\overline{\ethat\Psihat}\Psihat)\dee\nu\Bigg|\nn\\
    &\lesssim(-\mu)^2\intS\ener_{deg}[\Psihat]\dee\nu+|{\log(-\mu)}|^{-2\alpha}\intS|\ethat\Psihat|^2\dee\nu,\nn\\
    \Bigg|2sf(r)\mu\dfrac{a^2\sin\theta\cos\theta}{(r^2+a^2)}\mathfrak{I}(\overline{\ethat\Psihat}\mcu\Psihat)\Bigg|&\lesssim(-\mu)^2|\mcu\Psihat|^2+|{\log(-\mu)}|^{-2\alpha}|\ethat\Psihat|^2,\nn\\
    |4cs(r-M)f(r)\mathfrak{R}(\overline{\ethat\Psihat}\partial_\ubar\Psihat)|&\lesssim |{\log(-\mu)}|^{-2\alpha}\varepsilon^{-1}|\ethat\Psihat|^2+\varepsilon|\partial_\ubar\Psihat|^2.\label{eq:mostdifficult}
\end{align}
All these bounds show that we can first chose $p=p(a,M)$ large enough, and $\varepsilon>0$ small enough at each step, such that all the terms with a factor $g(r)$ in $\mathbf{B}[\Psihat]-\mathbf{B}_{pr}[\Psihat]$ can be absorbed in the RHS of \eqref{eq:aobserver}. Then, $p$ is fixed and the remaining bounds on the $f(r)$ terms, except \eqref{eq:mostdifficult}, can be absorbed in the RHS of \eqref{eq:aobserver}, chosing $\rb$ sufficiently close to $r_-$, where we use
\begin{align*}
    0<-\mu\ll 1,\quad0<-\mu\ll|{\log(-\mu)}|^{-\alpha-1},\quad|{\log(-\mu)}|^{-2\alpha}\ll|{\log(-\mu)}|^{-\alpha-1},
\end{align*}
for $r_-\leq r\leq \rb$ with $\rb$ close to $r_-$ and $\alpha>1$.
The last remaining term that is not directly absorbed in \eqref{eq:aobserver} using the inequalities above is \eqref{eq:mostdifficult}. To treat this term, we use the positive term \eqref{eq:positive}. More precisely, now that $p$ if fixed, choosing $\varepsilon>0$ small enough in \eqref{eq:mostdifficult}, we can absorb the term $\varepsilon|\partial_\ubar\Psihat|^2$ in \eqref{eq:positive}, and finally using $|{\log(-\mu)}|^{-2\alpha}\ll|{\log(-\mu)}|^{-\alpha-1}$ we can absorb the term
$$|{\log(-\mu)}|^{-2\alpha}\varepsilon^{-1}|\ethat\Psihat|^2$$
in \eqref{eq:aobserver}. This implies the lower bound \eqref{eq:bulkposapp}, and hence concludes the proof of Lemma \ref{lem:bulkpos}.\end{proof}

\section{Proof of Proposition \ref{prop:v1ell}}
\label{section:v1ell}

We search for a solution ${}^{(m,\ell)}\!v_1(r)$ of \eqref{eq:odeHell} as a polynomial of the type
$${}^{(m,\ell)}\!v_1(r)=(r-r_-)^{\ell-2}+\sum_{k=0}^{\ell-3}a_k(r-r_-)^k.$$
Note that for such a choice of ${}^{(m,\ell)}\!v_1(r)$, for any family $(a_k)$, the coefficient in front of $(r-r_-)^{\ell-2}$ in the LHS of \eqref{eq:odeHell} is
$$(\ell-2)(\ell-3)+6(\ell-2)+\lambda_\ell=(\ell-2)(\ell+3)+\lambda_\ell=0.$$
Then ${}^{(m,\ell)}\!v_1$ satisfies \eqref{eq:odeHell} if and only if the coefficients in front $(r-r_-)^k$, $0\leq k\leq\ell-3$, in the LHS of \eqref{eq:odeHell} all vanish. Namely, if and only if the following identities are satisfied: 
\begin{align*}
    &a_{\ell-3}=\dfrac{(\ell-2)(\ell(r_+-r_-)-2iam)}{(\ell-3)(\ell+2)-(\ell-2)(\ell+3)},\\
    &a_k=\dfrac{(k+1)((k+3)(r_+-r_-)-2iam)}{k(k+5)-(\ell-2)(\ell+3)}a_{k+1},\quad 0\leq k\leq \ell-4.
\end{align*}
Solving the system above gives \eqref{eq:h1defell}, and the estimate ${}^{(m,\ell)}\!v_1(r) \neq 0$ near $r=r_-$, as
$${}^{(m,\ell)}\!v_1(r_-)=\prod_{j=0}^{\ell-3}\frac{(j+1)((j+3)(r_+-r_-)-2iam)}{j(j+5)-(\ell-2)(\ell+3)}\neq 0.$$
Next, we prove \ref{item:(iii)}. Let $r_\ell\in(r_-,r_+)$ be sufficiently close to $r_-$ such that ${}^{(m,\ell)}\!v_1(r)\neq 0$ on $[r_\ell,r_-]$. We define, on $[r_\ell,r_-)$, 
$$h_{m,\ell}(r):={}^{(m,\ell)}\!v_1(r)\int_{r_\ell}^r\frac{\Delta^{-3}(r')e^{-2imr_{mod}(r')}}{{}^{(m,\ell)}\!v_1(r')^2}\dee r'.$$
Note that $h_{m,\ell}$ solves \eqref{eq:odeHell}, and that $h_{m,\ell}$ is linearly independent from  ${}^{(m,\ell)}\!v_1$ on $[r_\ell,r_-)$, so that $(h_{m,\ell},{}^{(m,\ell)}\!v_1)$ is a basis of solutions of \eqref{eq:odeHell} on $[r_\ell,r_-)$. Moreover, 
\begin{align}
    h_{m,\ell}(r)&=\frac{1}{{}^{(m,\ell)}\!v_1(r_-)}\int_{r_\ell}^r{\Delta^{-3}(r')e^{-2imr_{mod}(r')}}\dee r'+O_{(m,\ell)}((r_+-r)^{-1})\nn\\
    &=\frac{-\Delta^{-2}(r)e^{-2imr_{mod}}}{(4(r_--M)+2iam){}^{(m,\ell)}\!v_1(r_-)}+O_{(m,\ell)}((r-r_-)^{-1})\label{eq:usethis}
\end{align}
by Lemma \ref{lem:v2dv}. Now, we define 
\begin{align}\label{eq:defvtilde}
    \widetilde{v}(r):={}^{(m,\ell)}\!v_1(2M-r),
\end{align}
and using that ${}^{(m,\ell)}\!v_1$ satisfies \eqref{eq:odeHell}, we deduce that $\widetilde{v}$ satisfies the equation
$$\Delta\widetilde{v}''(r)+(6(r-M)-2iam)\widetilde{v}'(r)+\lambda_\ell\widetilde{v}(r)=0,$$
namely \eqref{eq:odeHell} with $m$ replaced by $-m$, so that there exist constants $c_1,c_2$ such that 
$$\widetilde{v}(r)=c_1 {}^{(-m,\ell)}\!v_1(r)+c_2 h_{-m,\ell}(r).$$
Using \eqref{eq:usethis} and the boundedness of $\widetilde{v}(r)$ at $r=r_-$ by \eqref{eq:defvtilde} and \eqref{eq:h1defell}, we get $c_2=0$. Moreover, ${}^{(-m,\ell)}\!v_1=\overline{{}^{(m,\ell)}\!v_1}$, so we get
\begin{align}\label{eq:match}
    \widetilde{v}(r)=c_1 \overline{{}^{(m,\ell)}\!v_1(r)}.
\end{align}
Using the fact that $\widetilde{v}(r)={}^{(m,\ell)}\!v_1(2M-r)$ and $\overline{{}^{(m,\ell)}\!v_1}$ are polynomials of degree $\ell-2$, and matching coefficients in front of $r^{\ell-2}$ in \eqref{eq:match}, we obtain $c_1=(-1)^{\ell-2}$, and thus 
$${}^{(m,\ell)}\!v_1(2M-r)=(-1)^{\ell-2}\overline{{}^{(m,\ell)}\!v_1(r)}.$$
In particular, this implies ${}^{(m,\ell)}\!v_1(r_+)=(-1)^{\ell-2}\overline{{}^{(m,\ell)}\!v_1(r_-)}\neq 0,$ and thus ${}^{(m,\ell)}\!v_1(r)\neq 0$ near $r=r_+$. We have proven for now \ref{item:(i)}, \ref{item:(ii)}, and \ref{item:(iii)}. We now prove \ref{item:(iv)}. Multiplying by $\overline{{}^{(m,\ell)}\!v_1'(r)}$ the ODE \eqref{eq:odeHell} satisfied by ${}^{(m,\ell)}\!v_1$ and taking the real part, we obtain
\begin{align*}
    0&=\Real\left(\left(\Delta {}^{(m,\ell)}\!v_1''(r)+(2iam+6(r-M)){}^{(m,\ell)}\!v_1'(r)+\lambda_\ell {}^{(m,\ell)}\!v_1(r)\right)\overline{{}^{(m,\ell)}\!v_1'(r)}\right)\\
    &=\frac{\Delta}{2}\frac{\dee}{\dee r}\left(\left|{}^{(m,\ell)}\!v_1'\right|^2\right)+6(r-M)\left|{}^{(m,\ell)}\!v_1'\right|^2-\frac{(\ell-2)(\ell+3)}{2}\frac{\dee}{\dee r}\left(\left|{}^{(m,\ell)}\!v_1\right|^2\right)\\
    &=\frac{1}{2}\frac{\dee}{\dee r}\left(\Delta\left|{}^{(m,\ell)}\!v_1'\right|^2\right)+5(r-M)\left|{}^{(m,\ell)}\!v_1'\right|^2-\frac{(\ell-2)(\ell+3)}{2}\frac{\dee}{\dee r}\left(\left|{}^{(m,\ell)}\!v_1\right|^2\right),
\end{align*}
where we used $\Delta'(r)=2(r-M)$. Integrating between $r_-$ and $r$, we get 
$$\frac{(\ell-2)(\ell+3)}{2}\left(\left|{}^{(m,\ell)}\!v_1(r_-)\right|^2-\left|{}^{(m,\ell)}\!v_1(r)\right|^2\right)=-\frac{\Delta}{2}\left|{}^{(m,\ell)}\!v_1'(r)\right|^2+5\int_{r_-}^r(M-r')\left|{}^{(m,\ell)}\!v_1'(r')\right|^2\dee r'.$$
Using $\Delta\leq 0$, we deduce $\left|{}^{(m,\ell)}\!v_1(r_-)\right|^2-\left|{}^{(m,\ell)}\!v_1(r)\right|^2\geq 0$ for $r\in[r_-,M]$, and thus 
$$\left|{}^{(m,\ell)}\!v_1(r_-)\right|\geq\left|{}^{(m,\ell)}\!v_1(r)\right|,\quad r\in[r_-,M].$$
By \ref{item:(iii)}, we get that the bound above holds also for $r\in[M,r_+]$, hence \ref{item:(iv)}. We now define the function ${}^{(m,\ell)}\!v_2$. Let $\varepsilon_{(m,\ell)}>0$ such that ${}^{(m,\ell)}\!v_1(r)\neq 0$ on $[r_-,r_--2\varepsilon)\cup [r_+-2\varepsilon,r_+]$. Let $r_0\in (r_-+\varepsilon,r_-+2\varepsilon)$, so that ${}^{(m,\ell)}\!v_1(r_0)\neq 0$. The scalar $\Delta$ is nowhere-vanishing on $I_{(m,\ell)}:=(r_-+\varepsilon,r_+-\varepsilon)$, so the ODE \eqref{eq:odeHell} admits a unique solution ${}^{(m,\ell)}\!v_2(r)$ defined on $I_{(m,\ell)}$ such that 
\begin{align*}
    {}^{(m,\ell)}\!v_2(r_0)=0,\quad {}^{(m,\ell)}\!v_2'(r_0)=\frac{-\Delta^3(r_0)e^{-2imr_{mod}(r_0)}}{v_1(r_0)}.
\end{align*}
Then, the differential equation satisfied by the Wronskian $W_{(m,\ell)}$ implies that on $I_{(m,\ell)}$, 
\begin{align}\label{eq:smooth1}
    W_{(m,\ell)}=-\Delta^3e^{-2imr_{mod}}.
\end{align}
We now extend smoothly ${}^{(m,\ell)}\!v_2(r)$ to $(r_-,r_-+\varepsilon)$ by
\begin{align}\label{eq:asympt1}
    {}^{(m,\ell)}\!v_2(r)={}^{(m,\ell)}\!v_1(r)\int_{r_0}^r\frac{\Delta^{-3}(r')e^{-2imr_{mod}(r')}}{{}^{(m,\ell)}\!v_1(r')^2}\dee r',\quad r\in (r_-,r_-+\varepsilon),
\end{align}
where the extension is smooth by \ref{item:(ii)} and \eqref{eq:smooth1}. Finally, let $r_1\in [r_+-2\varepsilon,r_+-\varepsilon)$. We extend ${}^{(m,\ell)}\!v_2$ to $[r_+-\varepsilon,r_+)$ by 
\begin{align}\label{eq:asympt2}
    {}^{(m,\ell)}\!v_2={}^{(m,\ell)}\!v_1(r)\int_{r_1}^r\frac{\Delta^{-3}(r')e^{-2imr_{mod}(r')}}{{}^{(m,\ell)}\!v_1(r')^2}\dee r'+\frac{v_2(r_1)}{v_1(r_1)}v_1(r),\quad r\in (r_+-\varepsilon,r_+),
\end{align}
where the extension is smooth by \ref{item:(ii)} and \eqref{eq:smooth1}. This definition ensures that \ref{item:(a)}, \ref{item:(b)}, \ref{item:(c)} hold. Moreover, \ref{item:(d)} holds by \eqref{eq:asympt1} and \eqref{eq:asympt2}, similarly as in \eqref{eq:usethis}. This concludes the proof of Proposition \ref{prop:v1ell}.

\setcounter{secnumdepth}{0}
\section{References}
\printbibliography[heading=none]
\end{document}